\documentclass[12pt]{article}

\usepackage{times}
\usepackage{fullpage}

\usepackage{amsfonts}
\usepackage{amssymb}
\usepackage{amsmath}
\usepackage{array}
\usepackage{needspace}
\usepackage{tcolorbox}
\usepackage{tikz}
\usepackage{latexsym}
\usepackage{url}
\usepackage[breaklinks,hidelinks,linktoc=all]{hyperref}

\newcommand{\ADVpm}{\mathrm{ADV}^\pm}

\newcommand{\identity}{\mathbf{1}}
\newcommand{\rk}{\mathrm{rk}}

\newcommand{\poly}{\mathrm{poly}}

\newcommand{\Ess}{\operatorname{Ess}}
\newcommand{\one}[1]{\mathbf{1}[#1]}
\newcommand{\cB}{\mathsf{B}}
\newcommand{\cR}{\mathsf{R}}

\long\def\rem#1{}

\newcommand{\ket}[1]{|#1\rangle}

\newcommand{\RR}{\mathcal{R}}
\newcommand{\LL}{\mathcal{L}}
\newcommand{\JJ}{\mathcal{J}}
\newcommand{\HH}{\mathcal{H}}
\newcommand{\dJ}{D_{\mathcal{J}}}
\newcommand{\dR}{D_{\mathcal{R}}}

\newtheorem{definition}{Definition}

\newtheorem{theorem}{Theorem}
\newtheorem{lemma}[theorem]{Lemma}
\newtheorem{proposition}[theorem]{Proposition}
\newtheorem{corollary}[theorem]{Corollary}

\newtheorem{remark}[theorem]{Remark}

\newenvironment{introexample}[1]{%
	\refstepcounter{theorem}%
	\par\medskip\noindent
	\textbf{Example~\thetheorem\ (#1).}\enspace\ignorespaces
}{\par\medskip}

\newcounter{introresult}
\newenvironment{introresult}{%
	\begin{tcolorbox}[
		colback=black!10,colframe=black!10,boxrule=0pt,sharp corners,
		boxsep=0pt,left=8pt,right=8pt,top=6pt,bottom=6pt,
		before skip=10pt,after skip=10pt]
	\refstepcounter{introresult}%
	\noindent\textbf{Result~\theintroresult.}\enspace\ignorespaces
}{\end{tcolorbox}}

\newcommand{\thmref}[1]{\hyperref[#1]{{Theorem~\ref*{#1}}}}
\newcommand{\lemref}[1]{\hyperref[#1]{{Lemma~\ref*{#1}}}}
\newcommand{\corref}[1]{\hyperref[#1]{{Corollary~\ref*{#1}}}}
\newcommand{\eqnref}[1]{\hyperref[#1]{{Equation~(\ref*{#1})}}}
\newcommand{\claimref}[1]{\hyperref[#1]{{Claim~\ref*{#1}}}}
\newcommand{\remarkref}[1]{\hyperref[#1]{{Remark~\ref*{#1}}}}
\newcommand{\propref}[1]{\hyperref[#1]{{Proposition~\ref*{#1}}}}
\newcommand{\factref}[1]{\hyperref[#1]{{Fact~\ref*{#1}}}}
\newcommand{\defref}[1]{\hyperref[#1]{{Definition~\ref*{#1}}}}
\newcommand{\exampleref}[1]{\hyperref[#1]{{Example~\ref*{#1}}}}
\newcommand{\hypref}[1]{\hyperref[#1]{{Hypothesis~\ref*{#1}}}}
\newcommand{\secref}[1]{\hyperref[#1]{{Section~\ref*{#1}}}}
\newcommand{\chapref}[1]{\hyperref[#1]{{Chapter~\ref*{#1}}}}
\newcommand{\apref}[1]{\hyperref[#1]{{Appendix~\ref*{#1}}}}

\newcommand{\ignore}[1]{}

\newenvironment{proof}[1][Proof.]{
	\par
	\noindent \textbf{#1}
}{
	\unskip
	\nobreak\hfill\penalty50\hskip3pt\hbox{}\nobreak\hfill
	\hbox{$\Box$}\par\bigskip
}

\begin{document}
\title{The quantum query complexity of the semigroup product problem}
\author{
  Troy Lee\thanks{Centre for Quantum Software and Information, University of Technology Sydney. \url{troyjlee@gmail.com}} \and
  Miklos Santha\thanks{Quantinuum, Singapore. \url{miklos.santha@quantinuum.com}}
}
\date{}
\maketitle

\begin{abstract}
We study the quantum query complexity of computing a semigroup product
$x_1\cdots x_n$, when one query reveals one input element and the multiplication
table is given.  For a finite aperiodic semigroup of size $N-1$, the
argument of Aaronson, Grier, and Schaeffer~\cite{AGS19} gives an upper bound of
$\sqrt n\,(N\log(nN+2))^{O(N)}$ queries.

To obtain query bounds that reflect algebraic structure, we study
the \emph{product breadth} $\beta$: the smallest bound such that every
input word has a subsequence of at most $\beta$ letters with the same
product.  For commutative monoids this is
the small Davenport constant~\cite{GHK06,Wang14}.  We extend this
definition to ordered products in noncommutative monoids.

\begin{itemize}
\item For nontrivial finite commutative aperiodic monoids, the bounded-error
quantum query complexity is
$\Theta(\min\{n,\sqrt{n\beta}\})$, and is thus characterized by product breadth.
We further show that if such a monoid $M$ has aperiodicity index $k$
(the least positive integer satisfying $x^k=x^{k+1}$ for every $x\in M$),
then $\beta=O(k\log(|M|+1)\log\log(|M|+2))$.
Products of capped counters nearly match this bound.
\item For monoids with a stable partial order in which the identity is
the minimum element, we prove that the bounded-error quantum query complexity is at most
$\sqrt{n+1}((\beta+2)\log(n+2))^{O(\log(\beta+2))}$.
Applications include quantum algorithms using $\widetilde O(\sqrt n)$
queries for the best time to buy and sell stock problem and for products
of unitriangular tropical matrices of fixed dimension.
\end{itemize}

For arbitrary finite aperiodic semigroups of order $N-1$, we improve
the bound of Aaronson, Grier, and Schaeffer, obtaining a bounded-error
quantum query complexity of at most
\[
\min\left\{n,\sqrt n\, \log^{O((N\log(N+2))^{1/3})}(n+2)\right\}.
\]
The dependence on semigroup size is nearly tight: the bounded-depth Dyck
lower bound of Ambainis et al.~\cite{ABIKPSSV20} yields aperiodic monoids
requiring $\sqrt n\,2^{\Omega(N^{1/3})}$ quantum queries in the relevant parameter range.
\end{abstract}

\clearpage
\begingroup
\fontsize{10}{10}\selectfont
\tableofcontents
\endgroup

\clearpage
\section{Introduction}\label{sec:overview}
Let $\mathbf{S}=(S,\otimes)$ be a semigroup.  In the \emph{semigroup product problem},
we are given query access to $x_1,\ldots,x_n\in S$ and asked to compute
$x_1\otimes\cdots\otimes x_n$.  The multiplication table of $\mathbf{S}$ is
assumed to be given, and a query reveals an entire element $x_i$ (which could
itself be a matrix, for example).  We study the bounded-error quantum query
complexity of this problem, in particular how it depends on both $n$ and the
algebraic structure of $\mathbf{S}$.

Our first motivation comes from divide-and-conquer algorithms.  For example,
mergesort uses \emph{input-independent splitting}: it divides the input into two
blocks without inspecting their values, then sorts the blocks recursively
and merges the results.  Quicksort uses \emph{input-dependent splitting},
since the partition into subproblems depends on the input values and the
choice of pivot.  Semigroup products provide an algebraic setting for
studying divide-and-conquer with input-independent splitting:
associativity allows us to split the input at any boundary and combine the
products of the two parts.

More generally, on input $x = x_1, \ldots, x_n$, it is common for an input-independent
divide-and-conquer algorithm to represent a contiguous block $w = x_i, x_{i+1}, \ldots, x_j$ by a summary
$\sigma(w)$, and to combine summaries of adjacent blocks by an associative
operation satisfying $\sigma(uv)=\sigma(u)\otimes\sigma(v)$.  Computing the
summary of the whole input is then a semigroup product problem.  For
example, to find the maximum of an array, the summary of a block is its
maximum, and the operation combining two summaries is again maximum.

A richer example is the \emph{best time to buy and sell stock} (BTBS)
problem.  Given integer prices $x_1,\ldots,x_n$, we seek the largest profit
from buying on one day and selling on a later day:
\[
  \mathrm{BTBS}(x)=\max_{i<j}(x_j-x_i).
\]
We can summarize each nonempty interval by its minimum price, maximum price, and
best profit.  If adjacent intervals have summaries $(a,b,c)$ and $(d,e,f)$,
their combined summary is
\[
  (a,b,c)\otimes(d,e,f)
  =\bigl(\min(a,d),\max(b,e),\max(c,f,e-a)\bigr).
\]
The third coordinate accounts for the best profit within either interval and
for buying at the minimum of the first interval and selling at the maximum
of the second.  A singleton has no feasible transaction, so its summary is
$(x_i,x_i,-\infty)$, using the convention $\max\varnothing=-\infty$.
This $\otimes$ operation is associative, and multiplying the singleton summaries
gives $\mathrm{BTBS}(x)$ in the third coordinate.
Later we see this operation can be realized as multiplication of $3$-by-$3$
matrices over the tropical semiring $(\max,+)$.

A second motivation for the semigroup product problem comes from regular languages.  A language
$L\subseteq\Sigma^*$ is regular if and only if it is recognized by a finite monoid, meaning there are a finite monoid
$(M,\otimes)$, a subset $F\subseteq M$, and a map $\phi:\Sigma\to M$ such that
\[
  x_1\cdots x_n\in L
  \quad\Longleftrightarrow\quad
  \phi(x_1)\otimes\cdots\otimes\phi(x_n)\in F.
\]
Here a monoid is a semigroup with an identity element, which is also the
product of the empty word.\footnote{Any semigroup $S$ can be made into a
monoid $S^1$ by adjoining an identity.  An upper bound for products in $S^1$
therefore also applies to products in $S$.}
Thus regular-language membership reduces to computing a finite monoid
product and testing whether the result belongs to $F$.

In this setting, Aaronson, Grier, and Schaeffer (AGS)~\cite{AGS19}
established a trichotomy theorem for quantum query complexity.  Informally,
every regular language has quantum query complexity either $O(1)$,
$\widetilde\Theta(\sqrt n)$, or $\Theta(n)$.\footnote{One precise
formulation takes the worst-case query complexity over input lengths at
most $n$.  For inputs of length exactly $n$, the statement must account for
periodic dependence on the length; see~\cite{AGS19}.  The $O(1)$ case
includes languages decidable without any queries.}
The central $\widetilde\Theta(\sqrt n)$ upper bound in their theorem concerns
star-free languages, languages which are recognized by \emph{aperiodic}
monoids.  A monoid $\mathbf M=(M,\otimes)$ is aperiodic if for every
element $a\in M$, there is a $k\geq1$ such that $a^k=a^{k+1}$.
The least $k$ that works for all elements simultaneously, when it exists,
is the \emph{aperiodicity index} $\iota(\mathbf M)$.  Every finite aperiodic
monoid has such an index.  For finite monoids, aperiodicity is equivalent
to containing no nontrivial subgroup.  AGS show the
quantum query upper bound on star-free languages by designing, for each fixed
finite aperiodic monoid $\mathbf{M}$ and target $m\in M$, an algorithm to test
whether $x_1\otimes\cdots\otimes x_n=m$ using
$\widetilde O_M(\sqrt n)$ queries.  Testing all possible targets, with
suitable error reduction, computes the product with the same asymptotic
bound when $|M|$ is constant.

As a direct consequence of the AGS upper bound for aperiodic monoids and elementary lower bounds,
one obtains the following trichotomy for the product problem in a fixed finite monoid
(\thmref{thm:fixed-monoid-trichotomy}):
\[
  Q_{1/3}(\operatorname{Prod}_{M,n})=
  \begin{cases}
    0, & |M|=1,\\[1mm]
    \widetilde\Theta_M(\sqrt n),
      & M\text{ is nontrivial and aperiodic},\\[1mm]
    \Theta_M(n), & M\text{ is nonaperiodic}.
  \end{cases}
\]
Here $Q_{1/3}$ denotes quantum query complexity with error probability
at most $1/3$ on every input, and
$\widetilde\Theta_M(\sqrt n)$ means an $\Omega(\sqrt n)$ lower bound and an
$O_M(\sqrt n\,\log^{O_M(1)}(n+2))$ upper bound.  The lower bounds have simple
explanations.  A nontrivial subgroup, with identity $e$ and an element $g$
of order $k>1$, allows products over $\{e,g\}$ to count modulo $k$, which
requires $\Omega(n)$ queries.  In a nontrivial aperiodic monoid, choose
$a\ne1$.  On inputs over $\{1,a\}$, the product is $1$ exactly when all
letters are $1$, giving the $\Omega(\sqrt n)$ lower bound for unstructured
search.

In this paper we therefore focus on the non-trivial aperiodic case.  
As in the AGS application the size of the monoid $M$ is constant, its influence
on the hidden constant and the power of $\log n$ is absorbed into the $\widetilde O(\cdot)$ notation.  
Tracking this dependence in their argument gives a bound of
the form $\sqrt n\,(|M|\log(n|M|+2))^{O(|M|)}$; we give a more precise
version in \thmref{thm:main-ags}.  Such a bound need not be sublinear
when $M$ grows with $n$.  Moreover, subexponential dependence on the monoid size is
unavoidable: the bounded-depth Dyck lower bounds of Ambainis et
al.~\cite{ABIKPSSV20} yield aperiodic families requiring
$\sqrt n\,2^{\Omega(|M|^{1/3})}$ queries in the parameter range described in
Section~\ref{sec:dyck-lb}.

On the other hand, we have examples like computing the maximum or
BTBS where the product problem can be solved over an \emph{infinite}
aperiodic monoid with only $O(\sqrt{n})$ or $O(\sqrt{n \log n})$ queries,
respectively.  This leads to our fundamental question:
\begin{quote}
  \emph{What structural properties of an aperiodic monoid determine the overhead beyond $\sqrt{n}$ in the quantum query
  complexity of its product problem?}
\end{quote}

We approach this question using a parameter called \emph{product breadth}.
Let $\mathbf{M} = (M, \otimes)$ be an aperiodic monoid and let $G\subseteq M$
be the set of allowed input elements.  For words $u,w\in G^*$, write
$u\preceq w$ if $u$ is a \emph{scattered subword}, or subsequence, of $w$.
Let $[w]$ denote the product of the letters of $w$, with the empty word
having product $1$.  A \emph{product core} of $w$ is a scattered subword
$u\preceq w$ with $[u]=[w]$.  We define the \emph{product breadth} by
\begin{equation}
  \beta_G(\mathbf{M})=\sup_{w\in G^*}
  \min\bigl\{|u|:u\preceq w,\ [u]=[w]\bigr\}.                    \label{eq:monoid-beta}
\end{equation}
Thus $\beta_G(\mathbf{M})$ is the smallest bound such that every word over $G$ has a
product core of at most that length; it may be infinite.  Write
$\beta(\mathbf M)=\beta_M(\mathbf M)$ when all monoid elements are allowed.

For a commutative monoid $M$ with the full alphabet $G=M$, product
breadth is exactly the \emph{small Davenport constant} $d(M)$ of
Geroldinger and Halter-Koch~\cite[Definition~2.8.12]{GHK06}; see also
Wang~\cite[Definition~C]{Wang14}.  With our empty-product convention,
for finite commutative monoids,
\[
  \beta(M)=d(M)=D(M)-1,
\]
where the \emph{large Davenport constant} $D(M)$ is the least integer
such that every word of that length has a proper product-preserving
subsequence~\cite[Definition~B and Proposition~D]{Wang14}.
Our definition also permits restricted input alphabets and retains the
original order of the letters when multiplication is noncommutative.
Our first main result characterizes quantum query complexity in terms
of this classical invariant for finite commutative aperiodic monoids.

For orientation, note that $\beta(\mathbf M)\geq\iota(\mathbf M)$ whenever
$\mathbf M$ is nontrivial and has an aperiodicity index.  Indeed, if
$k=\iota(\mathbf M)$, some $a\in M$ has distinct powers
$1,a,\ldots,a^k$, so the word of $k$ copies of $a$ has no shorter product
core.  The following examples illustrate product breadth; in particular,
the union example shows that this inequality can be strict.

\begin{introexample}{Maximum}\label{ex:max}
The monoid $\mathbf M=(\mathbb{N}_{\geq0},\max)$ is infinite, commutative,
and idempotent, with identity $0$.  Retaining a position attaining the
maximum preserves the product, so $\beta(\mathbf M)=1$.
Quantum maximum finding uses $\Theta(\sqrt n)$ queries, independently of
the magnitudes of the input values~\cite{DH96}.
This is $O(\sqrt{n\beta(\mathbf M)})$.
\end{introexample}

\begin{introexample}{Union}\label{ex:union}
Union is also commutative and idempotent.
For $(2^{[m]},\cup)$, one occurrence containing each atom preserves the
union, and the word of all singleton sets requires $m$ positions.
Thus $\beta=m=\log_2|M|$.  This problem is equivalent to finding all distinct letters
in the input, which has quantum query complexity $\Theta(\min\{n,\sqrt{nm}\})$.  Thus again we have an $O(\sqrt{\beta n})$ upper bound.
\end{introexample}

\begin{introexample}{Minimum spanning forests}\label{ex:spanning-forest}
Consider weighted edges on a fixed set of $v\geq2$ vertices, with a fixed
ordering to break weight ties.  Summarize an edge set by its minimum
spanning forest, and combine two summaries by taking the minimum spanning
forest of their union.  This operation is associative, commutative, and
idempotent, with the empty forest as identity.  The product of the
individual edges is their minimum spanning forest, so the product breadth
is at most $v-1$.  A spanning tree requires all of its edges, so
$\beta=v-1$.  For a graph with $m$ edges, D\"urr et al.~\cite{DHHM06}
give an $O(\sqrt{mv})$-query algorithm in the adjacency-array model
(random access to adjacency lists, with known vertex degrees).
The adjacency arrays have total size $n=2m$, so this is
$O(\sqrt{n\beta})$: the square root of input length times product breadth.
\end{introexample}

\begin{introexample}{Capped addition}\label{ex:capped-addition}
Let $\mathbf{M} = (\{0,\ldots,m\}, \oplus)$ where $a\oplus b=\min\{m,a+b\}$.  This is a commutative
monoid, but is not idempotent.  The aperiodicity index is $m$ as adding 1 to itself $i \le m$ times gives
distinct elements.  The product breadth is also $m$:
$m$ non-zero entries are necessary and sufficient in the worst case.  The product problem in this case is equivalent to
counting up to $m$ and has quantum query complexity $\Theta(\min\{n,\sqrt{mn}\})$, thus again an upper bound of $O(\sqrt{\beta n})$.
\end{introexample}

\begin{introexample}{Best time to buy and sell stock}\label{ex:btbs-breadth}
The BTBS monoid $\mathbf M$ introduced above is noncommutative and satisfies
$s^2=s^3$ for every $s\in M$, so it is aperiodic.  Take $G$ to be the
singleton price summaries.  Retaining a position attaining the minimum
price, one attaining the maximum price, and an optimal buy--sell pair
(when the input has at least two positions) preserves the full summary using
at most four positions.
This bound is attained by the price sequence $(10,2,5,0)$: its endpoints
are the unique extrema, and its middle pair gives the unique optimal
profit.  Thus $\beta_G(\mathbf M)=4$, even though prices can be arbitrary
integers.  Allcock et al.~\cite{ABBLS23} give an
$O(\sqrt{n\log n})$-query algorithm, independently of the price range.
\end{introexample}

\begin{introexample}{Bounded-depth Dyck}\label{ex:dyck-breadth}
Let $\mathbf M_k$ be the transition monoid of a counter on
$\{0,\ldots,k\}$, where $\mathtt{u}$ increments the counter,
$\mathtt{d}$ decrements it, and leaving the interval enters a dead state.
For $k\geq1$, this monoid is noncommutative and aperiodic.
Its product breadth is $\Theta(k)$, both for the parenthesis generators
$G=\{\mathtt{u},\mathtt{d}\}$ and for arbitrary monoid inputs
(\propref{prop:dyck-breadth} in \secref{sec:dyck-breadth}).  Nevertheless, the lower bound of
Ambainis et al.~\cite{ABIKPSSV20} gives
$\Omega(c^k\sqrt n)=\sqrt n\,2^{\Omega(\beta)}$ queries for a constant
$c>1$, even $n\geq2$, and $1\leq k\leq\log_2 n$, using only inputs over $G$.
Thus breadth alone does not ensure polynomial overhead beyond $\sqrt n$;
additional structural assumptions are needed.
\end{introexample}

The examples of maximum, union, and capped addition are all commutative
and aperiodic, with quantum query complexity
$\Theta(\min\{n,\sqrt{n\beta(\mathbf M)}\})$.
Our first result shows that this is true in general.

\begin{introresult}\label{res:commutative}
For every nontrivial finite commutative aperiodic monoid $\mathbf M$,
\[
  Q_{1/3}(\operatorname{Prod}_{M,n})
  =\Theta\!\left(\min\{n,\sqrt{n\beta(\mathbf M)}\}\right),
\]
with universal implicit constants.  Formal statement in
\thmref{thm:commutative-beta}.
\end{introresult}

We also bound breadth in terms of the size and aperiodicity index of a
finite commutative monoid.  Idempotence gives
$\beta(\mathbf M)\leq\log_2|M|$, while the identity $a^3=a^2$ for every
$a\in M$ gives $\beta(\mathbf M)\leq5\log_2|M|$.
More generally, for aperiodicity index $k=\iota(\mathbf M)$, we prove
\[
  \beta(\mathbf M)
  =O\!\left(k\log(|M|+1)\log\log(|M|+2)\right).
\]

The commutative theorem also gives an algorithm for finding a
minimum-weight basis among $n$ weighted input records from a known
rank-$r$ matroid using $O(\min\{n,\sqrt{nr}\})$ queries, where a query
reveals an element and its weight.  The reduction uses a commutative idempotent monoid of greedy
bases whose product breadth is at most $r$.  In particular, this recovers
the minimum spanning tree bound of D\"urr et al.~\cite{DHHM06} given in
\exampleref{ex:spanning-forest}.

The proofs of the commutative bounds appear in \secref{sec:lattice}; the
matroid construction and its applications are developed in
\secref{sec:matroid-bases}.

Our second result generalizes the quantum maximum-finding algorithm of
D\"urr and H\o yer~\cite{DH96} to noncommutative monoid products.
In the noncommutative case, \exampleref{ex:dyck-breadth} shows that the
overhead can be exponential in breadth, thus we need additional
assumptions to obtain efficient algorithms.  The D\"urr and H\o yer algorithm maintains a candidate maximum
and, in each successful round, uses quantum search to choose a uniformly
random input position whose value exceeds the current candidate, replacing
the candidate by that value.

To extend this idea to monoid products, we maintain the product of a
subsequence of the input as our candidate.  In each round, we sample
additional subsequences and merge their positions with those already
selected, obtaining a candidate at least as large as the
previous one.  The adversary method gives a quantum speedup for this
sampling procedure.  For the idea to work, we need that inserting additional
input elements leaves the candidate product unchanged or increases it.
This is guaranteed by a partial order
compatible with multiplication in which the identity is the minimum element.
The BTBS monoid satisfies this condition.  We now give the precise definition.

\begin{definition}[Stable order with minimum element $1$]\label{def:stable-least-order}
A partial order $\leq$ on a monoid $M$ is \emph{stable} if
\[
 a\leq b\quad\Longrightarrow\quad cad\leq cbd
 \qquad(a,b,c,d\in M).
\]
A monoid equipped with such an order is \emph{stably ordered}.
The identity is the \emph{minimum element} if $1\leq a$ for every $a\in M$.
\end{definition}

\begin{introresult}\label{res:ordered}
Let $M$ be a monoid with a known stable partial order in which the identity
is the minimum element, and let $G\subseteq M$ be a finite set of allowed inputs.  For
$1\leq\beta=\beta_G(M)<\infty$,
\[
 Q_{1/3}(\operatorname{Prod}_{M,G,n})
 \leq\min\left\{n,\sqrt{n+1}
       \bigl(C(\beta+2)\log(n+2)\bigr)^{C\log(\beta+2)}\right\},
\]
where $C$ is an absolute constant.
The bound is independent of $|M|$ and also applies to infinite monoids.
See \thmref{thm:ordered-beta-log-product} for the full statement and proof.
\end{introresult}

This result recovers the quantum query bound of Allcock et al.~\cite{ABBLS23}
for BTBS up to polylogarithmic factors, and applies to a much broader class
of problems.

\Needspace{5\baselineskip}
An important example is the monoid $U_k(\mathbb T)$ of upper
unitriangular matrices over the tropical semiring
$\mathbb T=(\mathbb R\cup\{-\infty\},\max,+)$: diagonal entries are zero,
and entries below the diagonal are $-\infty$.  Entrywise order is a stable
partial order in which the identity is the minimum element, and we show in
\corref{cor:unitriangular-beta} that
\[
  \beta\bigl(U_k(\mathbb T)\bigr)
  \leq\binom{k+1}{3}.
\]
Thus, for every fixed $k$, the product
can be computed with $\widetilde O_k(\sqrt n)$ queries, independently of
the magnitudes of the matrix entries.

These matrix products capture several natural optimization problems.
For example, the best profit from exactly $t$ successive stock
transactions is
\[
  \max_{1\leq i_1<j_1<\cdots<i_t<j_t\leq n}
  \sum_{r=1}^t(x_{j_r}-x_{i_r}).
\]
This is the $(1,2t+1)$ entry of a product in $U_{2t+1}(\mathbb T)$:
each price supplies alternating buy and sell weights $-x_i,+x_i$ along
a chain of $2t+1$ states.  More generally, the
fixed signed-sum problem of~\cite{ABBLS23}, which maximizes
$\sum_{j=1}^r\varepsilon_jx_{i_j}$ over $i_1<\cdots<i_r$ for prescribed
signs $\varepsilon_j\in\{-1,1\}$, is a product entry in
$U_{r+1}(\mathbb T)$.  In each case, one input value determines the entire
corresponding matrix, so the query bound is $\widetilde O(\sqrt n)$ for
every fixed number of transactions or signs.  The encodings are given in
\secref{sec:tropical}.

Finally, we give a bound that holds for every finite aperiodic monoid.
Let $N=|M|$.  A quantitative version of the AGS argument gives the upper bound
$\sqrt n\,(N\log(nN+2))^{O(N)}$ (\thmref{thm:main-ags}).
Although this is $\widetilde O(\sqrt n)$ for fixed $N$, the exponent is
linear in $N$, making the bound weak when the monoid grows with the input.
Our third result improves this to a cube-root dependence, up to
a logarithmic factor in the exponent.

\begin{introresult}\label{res:aperiodic-size}
For every finite aperiodic monoid $M$ of size $N$,
\[
 Q_{1/3}(\operatorname{Prod}_{M,n})
 \leq\min\left\{n,\;
   \sqrt n\,\log^{C(N\log(N+2))^{1/3}}(n+2)
 \right\},
\]
where $C$ is an absolute constant.  Formal statement in
\thmref{thm:ags-cuberoot-size}.
\end{introresult}

\Needspace{9\baselineskip}
The dependence on monoid size is nearly tight.  The depth-$k$ Dyck
monoids $M_k$ from \exampleref{ex:dyck-breadth} have
$N=|M_k|=\Theta(k^3)$ elements.  The lower bound of Ambainis et
al.~\cite{ABIKPSSV20} therefore gives
\[
 Q_{1/3}(\operatorname{Prod}_{M_k,n})
 =\Omega\!\left(\sqrt n\,2^{cN^{1/3}}\right)
\]
for an absolute constant $c>0$, $n\geq2$, and $1\leq k\leq\log_2 n$
(\corref{cor:dyck}).  Thus an exponential factor in $N^{1/3}$ is
unavoidable in general.  The remaining gap is the power of $\log n$,
whose exponent in our upper bound depends on $N$.

\paragraph{Lean formalization.}
We have developed a library for quantum query complexity in Lean~4,
registered on Palomar~\cite{palomar-2026-09-29-000001-v1}.
The library formalizes the quantum query model and proves strong duality,
exact composition, and the constant-factor adversary characterization of
bounded-error query complexity for total Boolean functions.  It also
includes the polynomial method.  For the upper bounds used here, the
library verifies the conversion of adversary dual solutions into
bounded-error quantum algorithms for arbitrary finite input alphabets and
output sets, with explicit constants independent of their sizes.

Using this library, we have developed a Lean formalization covering every
named result in the paper~\cite{MonoidProduct26}, save two intermediate
lemmas that are only formalized as part of the theorem in which they are used;
\remarkref{rem:ordered-formalization} explains this qualification.
In particular, all three main results are fully formalized with explicit constants.

\section{Overview of the proofs}\label{sec:proof-overview}

In this section we explain the proof ideas in our three main results.  
For a fixed input $x_1,\ldots,x_n$ and an index set
$U\subseteq[n]$, write $p(U)$ for the product of the selected letters in
their original order.  For the commutative and stably ordered cases, write
$\beta=\beta_G(M)$ and assume $1\leq\beta<\infty$.

\subsection{Commutative monoids: a sharp breadth bound}
\label{sec:overview-commutative}

The upper bound in \hyperref[res:commutative]{Result~\ref*{res:commutative}}
is an application of a new general technique for obtaining quantum query
upper bounds from classical computations
(\thmref{thm:essential-width}), which refines the decision-tree
framework of Beigi and Taghavi~\cite{BT20}.
Let $\Sigma$ and $O$ be finite input
and output alphabets, and let $f:D\to O$, where $D\subseteq\Sigma^{n}$
is the set of allowed input words.  Fix a finite set $Q$ of ``summary
states", possible states of the algorithm after partially reading the input.  
For each fixed input $x\in D$, the summary is a function
\[
 s_x:\mathcal P([n])\longrightarrow Q,
\]
where $\mathcal P([n])$ denotes the power set of $[n]$.  The value $s_x(T)$ is the summary state after revealing
the entries $x_i$ for $i\in T$.  In the commutative monoid application,
$Q=M$ and $s_x(T)=p(T)$, the product of those entries.

The summary starts at a fixed state $s_x(\varnothing)=q_\varnothing$.
When a new position is queried, it can be updated using only the current
state, the queried index, and the value revealed.  We require the
resulting state $s_x(T)$ to depend only on the revealed positions and
their values, independently of the order of queries.  A fixed readout
map $g:Q\to O$ recovers the answer from the final summary:
$f(x)=g(s_x([n]))$.

Call $i\in T$ \emph{essential} if
$s_x(T)\ne s_x(T\setminus\{i\})$.  The maximum number
of essential positions, over all inputs and revealed sets, is the
\emph{essential width of the summary}.  If the essential width is at most $B$, our theorem gives
\[
 \ADVpm(f)\leq16\sqrt{nB},
\]
where $\ADVpm$ is the general adversary bound, known to characterize bounded-error
quantum query complexity up to constant factors~\cite{HLS07,Rei11,LMRSS11}.
This is a general upper bound, allowing arbitrary finite input and output alphabets and
requiring no monoid structure.  If $B=0$, the function is constant and
the bound is immediate; assume $B>0$ in the proof sketch below.

To explain this theorem, we first recall the decision-tree method of Beigi and
Taghavi~\cite{BT20}.  Their basic bound converts a randomized classical
algorithm using $T$ queries into an $O(\sqrt{Tg})$ quantum algorithm
if a procedure predicting the next branch makes at most $g$ mistakes
on every fixed input, in expectation over the algorithm's internal randomness.
For an incremental summary, we can
scan positions in a uniformly random order and always predict that the
summary will stay unchanged.  Conditional on the first $t$ positions
forming a set $T_t$, the last position is uniform in $T_t$ and changes
the summary exactly when it is essential.  Thus
\[
 \Pr\{\text{the summary changes at step }t\mid T_t\}
 \leq\frac{B}{t}.
\]
Summing gives $O(B\log n)$ expected mistakes, so the Beigi-Taghavi bound yields
$O(\sqrt{nB\log n})$ queries.

We remove the logarithmic loss by using the change probability at each
depth before summing.  In the dual adversary construction, a pair of
inputs contributes at the first step where their summaries differ;
recording the full history prevents later contributions even if their
summaries become equal again.  At depth $t$, we assign weights
$\sqrt{t/B}$ to unchanged branches and $\sqrt{B/t}$ to changed branches.
The probability bound $B/t$ then makes the expected contribution at
this depth $O(\sqrt{B/t})$.  The total cost is
\[
 O\!\left(\sum_{t=1}^n\sqrt{\frac{B}{t}}\right)=O(\sqrt{nB}).
\]
This is the general upper bound of \thmref{thm:essential-width}.

For maximum finding, the summary is simply the largest revealed value,
and $B=1$: only a unique maximum can be essential.  The theorem therefore
recovers the optimal $O(\sqrt n)$ bound of D\"urr and H\o yer~\cite{DH96}.
Beigi and Taghavi obtain $O(\sqrt{n\log n})$ for the equivalent minimum
problem~\cite[Proposition~7]{BT20}, and ask whether other choices of
weights can improve their method~\cite[Section~6]{BT20}.  Our
construction answers this question for minimum and maximum finding,
removing the $\sqrt{\log n}$ loss within their decision-tree
framework.\footnote{Cornelissen, Mande, and Patro~\cite{CMP25} give an
optimal weighting scheme for decision trees querying Boolean inputs,
but do not obtain the $O(\sqrt n)$ bound for minimum or maximum finding over a
general ordered alphabet.}

To apply the theorem to a commutative aperiodic monoid, use the subset
product $p(T)$ as the summary.  Commutativity makes it independent of
the reveal order and allows updates by multiplication.  Its essential
width is exactly $\beta$.  Indeed, divisibility gives a stable partial
order: define $a\leq b$ if $b=ac$ for some $c\in M$.  Under this order,
inserting letters cannot decrease the product.  If $D\subseteq T$ indexes a core
of the revealed subword and $i\in T\setminus D$, then
\[
 p(T)=p(D)\leq p(T\setminus\{i\})\leq p(T).
\]
Every core therefore retains every essential position, giving width
at most $\beta$.  Conversely, all positions of a shortest core are
essential in that core.  Taking $B=\beta$ in the general theorem proves
the $O(\sqrt{n\beta})$ upper bound.

For the matching lower bound, choose a shortest core of length $\beta$
and retain its first $r=\min\{\beta,\lfloor(n+1)/2\rfloor\}$ letters.
This prefix has no proper preserving subword: replacing it by such a
subword would shorten the original core without changing its product.
Hide these $r$ letters among $n-r$ identities and permute all positions.
Commutativity makes the product independent of the permutation, while
deleting any one of the $r$ chosen letters changes it.  Each complete
input has $r$ such deletions, and each input missing a letter has
$n-r+1$ positions in which that letter could be restored.  The resulting
adversary bound is
$\Omega(\sqrt{r(n-r+1)})=\Omega(\min\{n,\sqrt{n\beta}\})$.
Together with the option of reading every position, this proves the
claimed tight bound.

The concrete upper bounds on $\beta$ come from counting subset products of
a shortest core.  In the idempotent case, let $D$ index a shortest core
and suppose two distinct subsets $A,B\subseteq D$ have $p(A)=p(B)$.
Swapping them if necessary, choose $i\in A\setminus B$.  Since $i\in A$,
commutativity and idempotence give $p(A)x_i=p(A)$, and hence
\[
 p(B)x_i=p(A)x_i=p(A)=p(B).
\]
Thus the letters indexed by $B$ already absorb $x_i$.  Since
$B\subseteq D\setminus\{i\}$, deleting position $i$ from the core leaves
its product unchanged, contradicting minimality.  All $2^{|D|}$
subset products are therefore distinct, giving $2^{|D|}\leq|M|$ and
hence $\beta\leq\log_2|M|$.

For the more general identity $x^{k+1}=x^k$, distinct subsets can have
the same product.  Let $D$ index a shortest core of length $d$, and
partition the $2^d$ subsets of $D$ into the
families
\[
 \mathcal F_z=\{S\subseteq D:p(S)=z\},
 \qquad z\in M.
\]
Bounding the size of each family therefore bounds $d$ in terms of $|M|$.

We rule out $(k+1)$-sunflowers in each family.  Such a sunflower
consists of sets $C\cup P_0,\ldots,C\cup P_k$, where the petals $P_j$
are pairwise disjoint and disjoint from $C$.  The sets are distinct, so at
least one petal is nonempty; relabel so that $P_0\ne\varnothing$.
Write $c=p(C)$ and $b_j=p(P_j)$.
Their equal products give $cb_0=\cdots=cb_k$.  Commutativity lets us
apply $cb_j=cb_k$ repeatedly: after replacing one factor, we retain $c$
and move the next factor next to it.  For example,
\[
 cb_0b_1=(cb_0)b_1=(cb_k)b_1=(cb_1)b_k=cb_k^2.
\]
Doing this for all $k+1$ petals, or for the $k$ petals other than $P_0$,
rewrites the two products as $cb_k^{k+1}$ and $cb_k^k$, respectively.
The stabilization identity $b_k^{k+1}=b_k^k$ therefore gives
\[
 c\prod_{j=0}^{k}b_j
 =cb_k^{k+1}
 =cb_k^k
 =c\prod_{j=1}^{k}b_j.
\]
Let $U=C\cup P_0\cup\cdots\cup P_k$ be the positions covered by the
sunflower, and let $R=D\setminus U$ be the remaining positions.
The preceding identity says that $p(U)=p(U\setminus P_0)$.
By commutativity, we can group the contribution of $R$ separately:
\[
 p(D)=p(R)p(U)=p(R)p(U\setminus P_0)=p(D\setminus P_0).
\]
Since $P_0$ is nonempty, $D\setminus P_0$ indexes a strictly shorter
preserving subword, contradicting minimality.  Thus each full product
fiber contains no such sunflower.  Applying the
sunflower bound of Bell, Chueluecha, and Warnke~\cite{BCW21} and counting
subsets of size approximately $\log_2|M|$ gives
\[
 \beta=O\!\left(k\log(|M|+1)\log\log(|M|+2)\right).
\]
The detailed arguments, and the complete proof of
\hyperref[res:commutative]{Result~\ref*{res:commutative}}, appear in
\secref{sec:lattice}.

\subsection{Stably ordered monoids: sampling and short subwords}
\label{sec:overview-ordered}

We explain \hyperref[res:ordered]{Result~\ref*{res:ordered}} in two steps:
first, in this subsection, an upper bound with exponent linear in $\beta$, then the improvement
to a logarithmic exponent in \secref{sec:overview-ordered-log}.
Let $M$ be an aperiodic monoid with a stable partial order in which the identity
is the minimum element.  We may assume $\beta<n$, since otherwise we can read the whole input.

The essential-width theorem requires a summary that can be updated as
positions are revealed in arbitrary order.  For stably ordered monoids
with minimum identity, the worst-case number of deletion-essential
positions still equals $\beta$.  However, the product of the revealed
subword need not support these updates: knowing $[x_1x_3]$ and $x_2$
need not determine $[x_1x_2x_3]$.  Retaining additional information may
permit updates, but can increase the essential width.  Thus the equality
of these parameters does not directly extend the commutative query bound
to this setting.

We generalize the D\"urr and H\o yer~\cite{DH96} quantum maximum-finding
algorithm to noncommutative products.  The algorithm proceeds in rounds, maintaining a candidate answer and
trying to increase it in each round.  In our case, the candidate is always a
product $p(K)$, which is initialized to $1 = p(\emptyset)$.  In a round, we update to a set
$K' \supseteq K$ with $p(K')\geq p(K)$, ideally with strict inequality.
The stable order guarantees monotonicity:
since $a1b\leq axb$, we have $p(U)\leq p(V)$ whenever
$U\subseteq V$.  Once the candidate dominates
the product of every subword of length at most $\beta$, we are done.
Indeed, let $D$ index a core of the input, with $|D|\leq\beta$.  Then
\[
 p(D)\leq p(K)\leq p([n])=p(D),
\]
so antisymmetry gives $p(K)=p(D)=p([n])$.

A difference with quantum maximum finding is that here we only have a partial order
rather than a total one.  An index set $A$
can be useful even when $p(A)$ and the current candidate $p(K)$ are incomparable: the product
$p(K\cup A)$ dominates both.  The analogue of an improvement is therefore
finding a subset $A$ with $p(K\cup A)>p(K)$.

We describe the algorithm as a randomized classical procedure.  A slight
generalization of the Beigi--Taghavi construction~\cite{BT20} to allow sub-routine calls (\thmref{thm:bt-subroutine-composition}) then gives
the quantum query upper bound .

To obtain a quantitative progress bound, the proof uses a batch of
$2\beta$ sampled records in each round.  A \emph{record} is a pair
$(V,(x_i)_{i\in V})$: an index set $V\subseteq[n]$
together with the revealed letters at those indices.  These data let us take unions of
index sets and compute their products in the original input order.
Fix the input $x$, and suppose a randomized classical procedure returns
a record with at most $\beta$ positions.  Let $\mu=\mu_x$ be the
distribution of these records over the procedure's random choices.
We write $V\sim\mu$ for the index set of a sampled record.
This distribution may depend on $x$ and need not be uniform.
The base case below samples a single input
position.  In the recursive cases, we combine the index sets returned
by recursive calls on smaller contiguous blocks of the input.  We then
retain at most $\beta$ positions with the same product as the combined set.

Initially, every record that the sampling procedure can return is
\emph{marked}.  A round independently draws $2\beta$ records from $\mu$
conditioned on being marked, keeping the marked set fixed throughout.
Let $A_1,\ldots,A_{2\beta}$ be their index sets.  We form their union
and update the candidate to $p(K')$, where
\[
 B=A_1\cup\cdots\cup A_{2\beta},\qquad K'=K\cup B.
\]
Save the full record of $K'$, and unmark every previously marked record
whose index set $V$ satisfies $p(K'\cup V)=p(K')$.  Once unmarked,
a record stays unmarked.  Every still-marked record has an index set
$V$ satisfying $p(K'\cup V)>p(K')$.  For a discarded record,
$p(V)\leq p(K'\cup V)=p(K')$, so its product remains dominated as more
positions are saved.

The batch size gives a simple progress bound.  For the analysis only,
hold the old set $K$ fixed and imagine one additional independent test
record from the same conditional distribution, with index set $V$.
It remains marked after the round exactly when
$p(K\cup B\cup V)\ne p(K\cup B)$.
Let $D$ index a core of the subword at positions $K\cup B\cup V$,
with $|D|\leq\beta$.  For each index in $D$ outside $K$, keep one sampled
record whose index set contains it.  This keeps at most $\beta$
records.  Omitting any other sampled record from the union leaves $D$ intact
and therefore preserves the product.  Thus at most $\beta$ sampled records
are essential.  Symmetry among the $2\beta+1$ sampled records gives a
survival probability of at most $\beta/(2\beta+1)<1/2$.  This is the expected fraction of the
current marked probability mass that survives the round.  Thus each round
halves the marked probability mass in expectation.  At the end, the subword
indexed by the saved union can be compressed to a core of length at most
$\beta$ without further queries.

We want a record sampled from $\mu$ conditioned on being marked.
We obtain one by classical rejection sampling: repeatedly sample
independent records from $\mu$, reject unmarked records, and accept
the first marked one.  Each sampled record is a \emph{proposal};
a \emph{draw} is the whole process of trying to obtain one marked record.
A complete round consists of $2\beta$ draws.
We allow at most $k$ proposals per draw, reporting failure if all are
rejected.  Conditional on success, the returned record has exactly the
desired conditional distribution.
Fix all random choices and predict that each proposal will be rejected.
A procedure requesting at most $d$ draws examines at most $dk$
proposals and makes at most $d$ prediction mistakes, since each draw
accepts at most one proposal.  If every fixed
proposal function has an adversary dual of cost at most $T$,
the prediction-tree construction and dual composition give cost
$O(T\sqrt{(dk)d})=O(Td\sqrt k)$.

Fix $1\leq h\leq n$; our goal is to reduce the marked mass below $1/(2h)$.
Taking a draw budget $d=\Theta(\beta\log(n+2))$ with a sufficiently
large constant allows $\ell=d/(2\beta)=\Theta(\log(n+2))$ rounds of
$2\beta$ draws each.
For ideal conditional draws, each round halves the marked mass in
expectation, so its expected value after $\ell$ rounds is at most $2^{-\ell}$.
By Markov's inequality, the probability that it is still at least
$1/(2h)$ is at most $2h\,2^{-\ell}=O(1/n)$.
Taking $k=\Theta(h\log(n+2))$ with a sufficiently large constant
also ensures that, over all draws, the probability of a failure while
the marked mass is at least $1/(2h)$ is $O(1/n)$.
Every unmarked record has its product dominated by the candidate, so
the undominated mass is at most the remaining marked mass.  Thus the
candidate index set $K$ satisfies, with probability at least $1-O(1/n)$,
\[
 \Pr_{V\sim\mu}\{p(V)\nleq p(K)\}<\frac1{2h},
 \qquad\text{at dual cost }O(\beta T\sqrt h\,\log^{3/2}(n+2)).
\]
To apply this bound, we now specify how to generate proposal records and
bound the corresponding adversary dual cost $T$, first for single letters
and then recursively for longer subwords.  We choose the sampling
rule so that failure to dominate the product of any target subword would
give probability at least $1/(2h)$ of generating a proposal whose product
is not dominated by the candidate.  The preceding guarantee therefore
ensures that the candidate dominates every target subword.

We do this by induction on the subword length $r$.  An \emph{$r$-summary}
of an interval is a record of at most $\beta$ positions from that interval
whose product dominates every subword of that interval of size at most $r$.
For $r=1$, apply the preceding sampling procedure to an interval of length
$m$ with $h=m$, generating each proposal by sampling a uniformly random
index and revealing its letter.
Each proposal reads just one input position, so $T=O(1)$.
With high probability, the resulting candidate record has undominated mass
less than $1/(2m)$.  Any undominated letter would contribute at least
$1/m$ to this mass, so on this event the record is a $1$-summary.

For larger $r$, split the interval using a balanced binary tree.  For each
nonleaf depth, we run a separate sampling procedure.  Fix one such depth,
containing $h$ nodes.  To generate a proposal, sample a node uniformly,
recursively obtain an $(r-1)$-summary from each of its two children, and
return a record of at most $\beta$ positions whose indexed subword is a
core of the subword on the union of the two children's index sets.
Stability ensures that the product of a
successful sampled record dominates every subword having at most $r-1$
positions in each child.  If any such subword remained undominated,
sampling its node and succeeding in both
children would give undominated mass of order $1/h$, exceeding the
allowed $1/(2h)$.  After running this procedure separately at every nonleaf
depth, we compress the union of the resulting candidates.  This gives an
$r$-summary: every subword of length between $2$ and $r$ splits across the
children of the smallest
tree interval containing its first and last positions, with at most $r-1$
positions in each child.
Child failures are already part of the distribution $\mu$; we need
no union bound over all recursive calls or all target subwords.

The square-root dependence on the interval length is preserved at each
inductive step.  With $h$ nodes, the recursive child intervals have
length $m/(2h)$, so the sampling and recursion costs cancel as
\[
 \sqrt h\,\sqrt{\frac{m}{2h}}=\sqrt{\frac m2}.
\]
Each step contributes a factor $O(\beta)$, a factor $\log^{3/2}(n+2)$
from rejection sampling, and one logarithm from summing over tree depths.
Reaching $r=\beta$ and converting the final dual to an algorithm thus
costs $(C\beta)^\beta\sqrt n\,\log^{5\beta/2}(n+2)$ and determines the full
product.  This intermediate bound is proved in \secref{sec:beta}.
We next show how to double the subword length covered by the summary at each step,
constructing $2r$-summaries from $r$-summaries.

\subsection{Stably ordered monoids: doubling the covered subword length}
\label{sec:overview-ordered-log}

The preceding argument increases the covered subword length $r$ one position at a time.
We now double it, reducing the number of recursive stages from $\beta$
to $O(\log(\beta+2))$ and obtaining
\hyperref[res:ordered]{Result~\ref*{res:ordered}}.  This makes the
dependence on $\beta$ quasipolynomial and the exponent of $\log n$
logarithmic in $\beta$.  At each stage, we construct a summary that
dominates every subword of size at most $2r$, using summaries that
dominate subwords of size at most $r$ on smaller intervals.
The main challenge is to combine these $r$-summaries into a
$2r$-summary without repeatedly paying to compute the same interval summary.

To see why this construction works, consider an arbitrary subword of
size between $2$ and $2r$, indexed by $U\subseteq[n]$, whose product
we need to dominate.  Partition $U$ into $U_1 = U \cap [k]$ and $U_2 = U \setminus U_1$ for
a $k$ such that $|U_1| = |U_2|$ if $|U|$ is even, and otherwise $|U_2| = |U_1| + 1$.
Put a balanced binary tree on $[n]$, and let
$v$ be the lowest common ancestor of $\max U_1$ and $\min U_2$.
The two children of $v$, together with the siblings off the
root-to-$v$ path, partition $[n]$ into intervals, called the
\emph{frontier} of $v$.  In Figure~\ref{fig:beta-median-frontier},
the shaded triangles represent these intervals, and $K_{[a,b]}$
denotes an $r$-summary of the interval $[a,b]$.
Each interval contains indices from at most one of $U_1$ and $U_2$,
and hence at most $r$ indices of $U$.
Replacing those positions by an $r$-summary of the interval can only
increase the product.  The union of these summaries is therefore a
record whose product dominates $p(U)$.
Singletons are covered at their leaf's parent.  Thus, if we combine
the records from all internal vertices and compress their union, we
obtain a $2r$-summary.

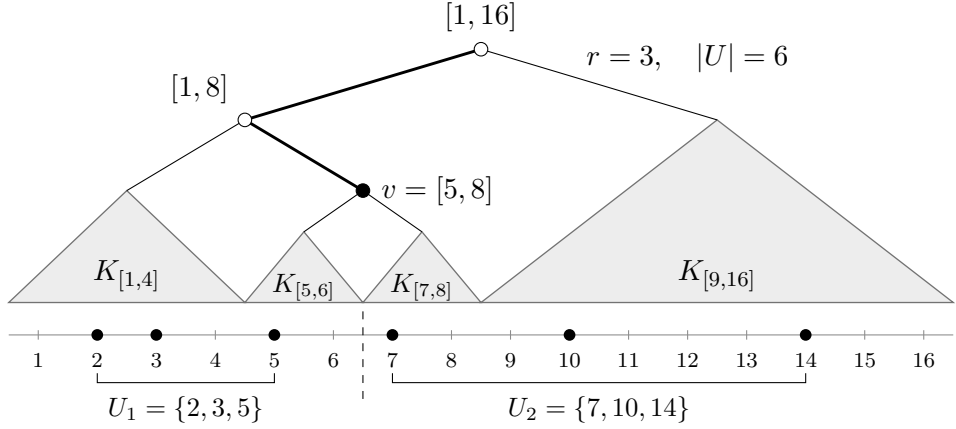
\begin{figure}[htbp]
\centering
\begin{tikzpicture}[x=0.78cm,y=0.78cm,font=\small,
  frontier/.style={draw=black!55,fill=black!7,line width=0.5pt},
  pathnode/.style={circle,draw,fill=white,inner sep=1.8pt}]
  \path[frontier] (2.5,2.6)--(0.5,0.7)--(4.5,0.7)--cycle;
  \path[frontier] (5.5,1.9)--(4.5,0.7)--(6.5,0.7)--cycle;
  \path[frontier] (7.5,1.9)--(6.5,0.7)--(8.5,0.7)--cycle;
  \path[frontier] (12.5,3.8)--(8.5,0.7)--(16.5,0.7)--cycle;
  \draw[line width=1.1pt] (8.5,5)--(4.5,3.8)--(6.5,2.6);
  \draw (8.5,5)--(12.5,3.8);
  \draw (4.5,3.8)--(2.5,2.6);
  \draw (6.5,2.6)--(5.5,1.9) (6.5,2.6)--(7.5,1.9);
  \node[pathnode,label=above:{$[1,16]$}] at (8.5,5) {};
  \node[pathnode,label=above left:{$[1,8]$}] at (4.5,3.8) {};
  \node[pathnode,fill=black,label=right:{$v=[5,8]$}] at (6.5,2.6) {};
  \node at (2.5,1.15) {$K_{[1,4]}$};
  \node[font=\footnotesize] at (5.5,0.96) {$K_{[5,6]}$};
  \node[font=\footnotesize] at (7.5,0.96) {$K_{[7,8]}$};
  \node at (12.5,1.15) {$K_{[9,16]}$};
  \node[anchor=west] at (10.1,4.85) {$r=3,\quad |U|=6$};
  \draw[black!45] (0.5,0.15)--(16.5,0.15);
  \foreach \i in {1,...,16} {
    \draw[black!45] (\i,0.07)--(\i,0.23);
    \node[below,font=\scriptsize] at (\i,0.02) {$\i$};
  }
  \foreach \i in {2,3,5,7,10,14} \fill (\i,0.15) circle (2.2pt);
  \draw[dashed] (6.5,0.55)--(6.5,-1.05);
  \draw (2,-0.50)--(2,-0.66)--(5,-0.66)--(5,-0.50);
  \node[below,font=\footnotesize] at (3.5,-0.70) {$U_1=\{2,3,5\}$};
  \draw (7,-0.50)--(7,-0.66)--(14,-0.66)--(14,-0.50);
  \node[below,font=\footnotesize] at (10.5,-0.70) {$U_2=\{7,10,14\}$};
\end{tikzpicture}
\caption{Doubling the covered subword length from $3$ to $6$.  The filled dots mark
$U=U_1\cup U_2$.  The middle positions $\max U_1=5$ and $\min U_2=7$ have lowest
common ancestor $v=[5,8]$.  The shaded frontier intervals contain
$2,1,1,2$ target positions, respectively, so their $3$-summaries
combine to dominate $p(U)$.  Reading the two child summaries at each
vertex on the bold path supplies the entire frontier; records on
other paths reuse the data at common ancestors.}
\label{fig:beta-median-frontier}
\end{figure}

To obtain all the summaries in the frontier of $v$, follow the path
from the root to $v$, computing the $r$-summaries of the two children
at each vertex.  This supplies every summary in the frontier.
Different paths share many of these computations, and our query bound
must account for this shared work.  Fix the random choices used to
construct each interval summary, so every reuse refers to the same summary.

We combine the frontier summaries in batches, as in
\secref{sec:overview-ordered}.  Write $C_v$ for the union of the index
sets of the frontier summaries at $v$, and let $K$ be the accumulated
index set for our current candidate.  Initially, $K=\varnothing$ and every
internal vertex is marked.  In each round, we independently sample $2\beta$
vertices uniformly from the currently marked set, with replacement,
keeping that set fixed throughout the batch.  We add their frontier unions
to $K$, then update the marks.  A previously marked vertex $v$ remains
marked exactly when
\[
 p(K\cup C_v)>p(K).
\]
If equality holds, we unmark $v$ permanently: its product satisfies
$p(C_v)\leq p(K\cup C_v)=p(K)$ and remains dominated as $K$ grows.
Thus, after each update, the marked vertices give frontier unions whose
insertion would improve the candidate.  Given the known history, the
marked status of $v$ depends only on the summaries supplied along its
root path.  We stop if no marked vertices remain.

The argument from \secref{sec:overview-ordered} shows that each round
halves the expected number of marked vertices.  Since there are $O(n)$
vertices, $O(\log(n+2))$ rounds suffice for all frontier products to be
dominated with constant success probability.  After these rounds, we
compress $K$ to at most $\beta$ positions while preserving its product.

To bound the complexity of combining the frontier summaries, we use
a general theorem for tree search with shared subcomputations
(\thmref{thm:weighted-tree-search}).  Suppose each vertex
$v$ has a deterministic function $g_v(x)$ with adversary dual cost
$t_v>0$, and whether a vertex is marked depends only on the function
values along its root path.  In our application, $g_v$ returns the two
child summaries at each internal vertex $v$.  Define the quantities
$\mathcal C_v$ recursively by
\[
 \mathcal C_v=t_v+
 \sqrt{\sum_{u\text{ child of }v}\mathcal C_u^2},
\]
so $\mathcal C_v=t_v$ at a leaf.  Using a tree instance of the
learning-graph construction of Belovs and Lee~\cite{BL11}, combined
with adversary composition, we prove that the decision problem of
determining whether any vertex is marked has an adversary dual of
cost at most $\mathcal C_\rho$, where $\rho$ is the root.
Adversary tightness then gives quantum query complexity $O(\mathcal C_\rho)$.
Shared work along a prefix adds once, while the costs of child branches
combine through the square root of the sum of their squares.
Different functions $g_v$ may depend on the same input letters,
as happens for nested intervals in our application.

To sample a uniform marked vertex, choose a random permutation of the
vertices and find its first marked entry by binary search.  Each decision
asks whether an initial segment contains a marked vertex.  For this test
only, we keep the whole tree but treat vertices outside the segment as
unmarked.  Since membership in the segment requires no input queries,
the marking condition still depends only on the function values along
the root path, so the same tree dual applies.
For each selected vertex, form its frontier union from the function values
along its path.  Repeating this search supplies the batches of $2\beta$
vertices used in the sampling rounds above.  With the permutations fixed, this is an exact
deterministic computation.  Keeping its transcript makes the dual costs
of its adaptive steps add, without an error-reduction factor per step.

At each stage, compress the union of $O(\log(n+2))$ independent copies
of the resulting summary.  If any copy succeeds, the union succeeds;
every copy describes a subword of the input, even when its domination guarantee fails.
This makes the child summaries in the next stage simultaneously correct
with high probability.  This amplification concerns classical seeds;
we still convert to a quantum algorithm only at the root.

To see the cost of one doubling, suppose the fixed-seed $r$-summary
functions on intervals of length $m$ have dual cost at most $A_r\sqrt m$.
At each tree depth the interval lengths sum to $O(n)$, so the squared
costs $t_v^2$ sum to $O(A_r^2n)$.  The tree has depth
$d=O(\log(n+2))$, so a convenient corollary of the tree-search theorem gives
\[
 \mathcal C_\rho\leq\sqrt{(d+1)\sum_v t_v^2}
 =O(A_r\sqrt n\log(n+2)).
\]
Binary search, the sampling rounds, and amplification
each add one further logarithm, giving the recurrence
\[
 A_{2r}\leq C\beta\log^4(n+2)\,A_r
\]
for the coefficient of $\sqrt m$, with an absolute constant $C$.
After $O(\log(\beta+2))$ stages, the summary describes a subword of the input
whose product dominates that of a core of the input.  Its product therefore
equals the full input product.  A single application of adversary tightness gives the claimed
query bound.  The full argument appears in
Sections~\ref{sec:beta-rank-doubling}--\ref{sec:beta-doubling-cost}.
We know of no example ruling out an $O(\sqrt{n\beta})$ query bound in
this setting; proving even a polynomial dependence on $\beta$ remains
an open problem.

We do obtain a tight dependence on breadth in an important special case:
the monoid $M=\operatorname{UT}_k(\mathbb B)$ of unitriangular matrices
over the Boolean semiring $\mathbb B=(\{0,1\},\vee,\wedge)$.
Equivalently, these are tropical unitriangular matrices with entries in
$\{-\infty,0\}$.  In \secref{sec:boolean-unitriangular}, we show that
$\beta(M)=\binom{k}{2}$ and, for $k\geq2$,
\[
 Q_{1/3}(\operatorname{Prod}_{M,n})
 =\Theta\!\left(\min\{n,\sqrt{n\beta(M)}\}\right)
 =\Theta\!\left(\min\{n,k\sqrt n\}\right).
\]
A consequence of Simon's theorem explains the importance of this
family: every finite $\JJ$-trivial monoid is a homomorphic
image of a submonoid of $\operatorname{UT}_k(\mathbb B)$ for some
$k$~\cite{Sim75,ST88}.  Here $\JJ$-trivial means that distinct elements
generate distinct two-sided ideals; this class includes every finite
monoid with a stable order in which the identity is the minimum element.

\subsection{General aperiodic monoids: the origin of the cube root}
\label{sec:general-aperiodic}

We first review the AGS proof for finite aperiodic monoids, then explain
how to refine it to obtain
\hyperref[res:aperiodic-size]{Result~\ref*{res:aperiodic-size}}.
Throughout, write $N=|M|$ and $L(n)=2+\log_2(n+2)$.

We partition $M$ into equivalence classes, called $\JJ$-classes, by
declaring $x$ and $y$ equivalent when they generate the same two-sided
ideal: $MxM=MyM$, where $MxM=\{axb:a,b\in M\}$.
Equivalently, each can be obtained from the other by multiplying on the
left and right by elements of $M$.
The $\JJ$-classes form a poset under inclusion of these ideals:
\[
 [x]_{\JJ}\leq_{\JJ}[y]_{\JJ}
 \quad\Longleftrightarrow\quad MxM\subseteq MyM.
\]
The class of $1$ is the greatest element, since $M1M=M$; other classes
need not be comparable.  Multiplication can only move downward in this
poset, since $M(xy)M\subseteq MxM$.  In particular, the classes of the
successive prefix products form a descending chain, allowing repetitions.

To test whether the input product $x_1\cdots x_n$ equals a given target
$m\in M$, the AGS proof proceeds by downward induction on the $\JJ$-class
of $m$, starting with the class of $1$.  For the base case, a product in an aperiodic
monoid equals $1$ if and only if all its factors equal $1$.  Standard
Grover search for a nonidentity input element therefore decides this
case using $O(\sqrt n)$ queries.

For the inductive step, fix a target $m\neq1$ and assume we can test
product equality with every target $t$ satisfying $MmM\subsetneq MtM$.
These tests can be applied to any interval of the input.  The induction
makes progress through the ideal order; the intervals themselves need
not become substantially shorter.

To describe the first part of the algorithm, write
$p_i=x_1\cdots x_i$ for the prefix products, with $p_0=1$.
The principal right ideal $p_iM=\{p_i a:a\in M\}$ consists of all
elements obtainable by extending this prefix.  These right ideals
can only shrink:
\[
 M=p_0M\supseteq p_1M\supseteq\cdots\supseteq p_nM.
\]
Just as equality of two-sided ideals defines $\JJ$-classes, equality
of right ideals defines $\RR$-classes:
\[
 R_p=\{q\in M:qM=pM\}.
\]
The prefix products therefore pass through a sequence of phases during
which their right ideal, and hence their $\RR$-class, stays fixed.
Once a phase ends, the product can never return to that class.
In a finite monoid, a strict drop of the right ideal also gives a
strict drop in the $\JJ$-order.

If the full product is $m$, there is a first prefix whose right ideal
is $mM$.  Write $r$ for the preceding prefix product and $a$ for the
letter at this transition.  Then
\[
 rM\supsetneq raM=mM.
\]
AGS enumerates the pairs $(r,a)$ satisfying this relation and tests
for the corresponding transitions.  For a fixed pair, the predicate
$r\in p_iM$ says that $r$ is still
reachable by extending the prefix.  It can change from true to false
only once, so binary search finds the last index $i$ for which it holds.
The required transition occurs exactly when $i<n$, $p_i=r$, and
$x_{i+1}=a$.  To evaluate the predicate, we use the inductively available
equality tests for the elements $s$ with $r\in sM$: all such targets satisfy
\[
 MmM\subsetneq MrM\subseteq MsM.
\]
Thus we can find the transition without computing every prefix product.
Symmetrically, a search from the right uses the descending left ideals
of suffix products to detect a transition into the target's left ideal
$Mm$.

These two boundary conditions ensure that the full product lies in
$mM\cap Mm$, but we must also check that it has not fallen strictly
below the $\JJ$-class of $m$.  Any contiguous interval with product
$z$ satisfying $m\notin MzM$ rules out the answer $m$.
A shortest such interval is either a single letter, which Grover search
can detect, or has the form $awb$, where $a,b\in M$ are its first and
last letters.  Writing $r=[w]$ for the product of the intervening word $w$
(with $r=1$ if $w$ is empty), minimality gives
\[
 m\in MarM,\qquad m\in MrbM,\qquad m\notin MarbM.
\]
The AGS decomposition (\thmref{thm:decomp}) shows that the absence of
these obstructions, together with the two boundary conditions,
characterizes product equality with $m$.  It also shows that every
middle product $r$ in such a triple satisfies $MmM\subsetneq MrM$,
so its equality test is available by induction.

It remains to search for a forbidden interval $awb$ for each
such triple $(a,r,b)$.  Its length is unknown, so we try the scales
$\ell=2,4,8,\ldots$ up to $n$.  If a forbidden interval exists, its
length lies between $\ell$ and $2\ell$ for one of these scales.
At that scale, every cut between consecutive letters of the interval
splits it as $aw_1\cdot w_2b$, where the products $u$ of $w_1$ and
$v$ of $w_2$ satisfy $uv=r$.  Both $u$ and $v$ generate ideals
containing $MrM$, so their equality tests are also available by
induction.  A cut can be tested by enumerating these
factorizations and applying suffix search to the window immediately
to its left and prefix search to the window immediately to its right,
each of length $O(\ell)$.  The forbidden interval supplies
$\Omega(\ell)$ suitable cuts, so quantum search finds one using
$O(\sqrt{n/\ell})$ iterations.
If the inductive equality tests cost $B\sqrt{k}$ on intervals of
length $k$, each cut test costs $B\sqrt\ell$ times a polynomial
factor in $NL(n)$.  The length factors cancel:
\[
 \sqrt{n/\ell}\,B\sqrt\ell=B\sqrt n.
\]
Trying all $O(\log(n+2))$ scales and all possible boundary pairs and
obstruction triples contributes only polynomial factors in $NL(n)$.
Thus each level of the induction preserves the $\sqrt n$ dependence,
while multiplying its coefficient by a polynomial in $NL(n)$.
The depth of this induction can be linear in $N$.  We carry out the
induction with exact adversary duals and combine the duals for all $N$
possible answers at a cost of twice their sum.  Converting the resulting
product dual to a quantum algorithm only at the end avoids error reduction,
giving the quantitative AGS bound (\thmref{thm:main-ags})
\[
 Q_{1/3}(\operatorname{Prod}_{M,n})
 \leq\min\left\{n,\;\sqrt n\,(NL(n))^{O(N)}\right\}.
\]

Our refinement also organizes the computation around phases with a
fixed right ideal.  An $\RR$- or $\JJ$-class is called \emph{regular}
if it contains an idempotent.  For an idempotent $e$, we strengthen the
inductive task to track multiplication within its $\RR$-class
\[
 R_e=\{p\in M:pM=eM\}.
\]
Starting from $e$, the algorithm returns either the final product
$ex_1\cdots x_n$, if it remains in $R_e$, or the first position where
the product leaves $R_e$, together with the product immediately before
that position.  One query to the exiting letter then gives the
starting product for the next right-ideal phase.

The main algebraic step shows that every phase, including one in a
nonregular class, can be simulated inside a regular $\RR$-class
(\lemref{lem:ags-owner-cover}).  There are at most $N$ phases, so
processing them sequentially requires at most $N$ phase computations,
with an additional logarithmic factor for error reduction.
Combining this simulation with the AGS step lets us charge the
recursive cost to chains of \emph{regular} $\JJ$-classes.
For a finite aperiodic monoid of $d$-dimensional matrices, strict
descent between regular $\JJ$-classes lowers matrix rank.  There are
therefore at most $d$ such steps, giving a product algorithm with
query bound $\sqrt n\,(NL(n))^{O(d+1)}$.

To use this dimension-sensitive bound for an arbitrary monoid, we need
suitable matrix representations.  A faithful representation of dimension
$N$ is easy to construct: take a vector space with basis
$\{v_s:s\in M\}$ and represent $m$ by the linear map $v_s\mapsto v_{ms}$.
Its action on $v_1$ distinguishes every element of $M$, but dimension
$N$ leaves us with an exponent linear in $N$.  We therefore seek smaller
representations, together with a way to recover any information they lose.

The rational monoid algebra provides the setting for this decomposition.
Its elements are formal linear combinations $\sum_{m\in M}\alpha_m m$,
with $\alpha_m\in\mathbb Q$, and multiplication is obtained by expanding
products of sums using the multiplication in $M$.  Thus we retain the
original multiplication while gaining the linear structure needed to
study kernels and quotients of representations.  Write $A=\mathbb Q[M]$
if $M$ has no zero.  If it has a zero, we identify its basis element
with the zero vector and write $A$ for the resulting \emph{contracted}
algebra.  In either case, $M$ embeds multiplicatively in $A$ and
$\dim A\leq N$.

Classical Munn--Ponizovski\u\i{} theory~\cite{Munn55,Munn57} gives all
irreducible matrix representations of $A$, whose restrictions to $M$
are maps $\rho_J:M\to M_{d_J}(\mathbb Q)$ indexed by
the nonzero regular $\JJ$-classes.  Extending these maps linearly and
collecting their outputs gives an algebra homomorphism
\[
 \pi:A\longrightarrow\prod_J M_{d_J}(\mathbb Q),
 \qquad R=\ker\pi.
\]
Computing all the matrix products
$\rho_J(x_1)\cdots\rho_J(x_n)$ therefore determines the input product
modulo $R$.  Equivalently, $A/R$ embeds in this product of matrix
algebras.  Each image $\rho_J(M)$ is still a finite aperiodic monoid,
and a query to $x_i$ determines its matrix image, so our matrix-product
bound applies to each coordinate.

The useful property of the lost information is that $R$ is
\emph{nilpotent}: every product of sufficiently many elements of $R$
is zero.  Here the factors need not be equal, and we have the uniform
bound $R^N=0$.  Classically, $R$ is the \emph{Jacobson radical}, the
common kernel of all irreducible representations; its nilpotence is
a general fact about finite-dimensional algebras.
\lemref{lem:ags-munn-decomposition} gives a direct proof of the
properties we need.  Quotienting by $R$ exposes the matrix coordinates,
while nilpotence lets us recover the missing information through a
bounded sequence of finer quotients.  We describe this recovery below.

The dimensions of the matrix coordinates satisfy $|J|\geq d_J^2$.
This gives a tradeoff between the cost of computing a coordinate by
the matrix algorithm and the size of the class associated with it.
Choose a dimension threshold $t$.  Blocks with $d_J\leq t$ can be
computed by the matrix algorithm at cost
$\sqrt n\,(NL(n))^{O(t)}$.  For a block with $d_J>t$, collapse the ideal
generated by $J$ to zero.  The resulting quotient monoid has at least
$t^2$ fewer elements, so it can be handled by induction on monoid size.

Collapsing this ideal loses information, and recovering it is a central
part of the proof.  To recover the $J$-block, we locate the first prefix
entering the ideal and track the product while it remains in $J$.
The recursive equality tests required for this reconstruction have
targets outside the collapsed ideal, where the quotient is injective.
Thus a product algorithm for the smaller quotient suffices, with only
polynomial overhead in $NL(n)$, independent of $d_J$.

After computing all the matrix coordinates, we know the product in
$A/R$.  We recover the original product by lifting it through
successively finer quotients:
\[
 A/R,\qquad A/R^2,\qquad A/R^4,\qquad\ldots,\qquad A.
\]
Each finer quotient retains more information, which we recover using
additional queries to the input.  The map from
$A/R^{2k}$ to $A/R^k$ has kernel $R^k/R^{2k}$, in which every product
of two elements is zero.  For the finite images of $M$, we prove that
an algorithm computing the product in the coarser quotient can be
lifted to one computing it in the finer quotient with only a
polynomial factor in $NL(n)$
(\lemref{lem:ags-square-zero-lift}).  This lifting statement is the
other main technical ingredient.  The underlying local-triviality theorem
is classical~\cite[Lemma~5.4]{AMSV09}; the new ingredient is the quantum
reconstruction and its quantitative bound.  Nilpotence makes the process
terminate: once $2^j\geq N$, the ideal $R^{2^j}$ is zero and the
quotient is $A$ itself.  There are therefore only $O(\log N)$ lifting
steps, for a total factor $(NL(n))^{O(\log N)}$ at each stage of the
induction on monoid size.

Every large-block reduction removes at least $t^2$ elements, so there
are at most $N/t^2$ such reductions along a recursive branch.  Combining
their costs with the small-block bound gives
\[
 Q_{1/3}(\operatorname{Prod}_{M,n})
 \leq\sqrt n\,(NL(n))^{O\left(t+(N/t^2+1)\log N\right)}.
\]
Balancing $t$ against $(N/t^2)\log N$ explains the cube root: take
$t$ of order $(N\log(N+2))^{1/3}$.  To put the resulting estimate in
the form of \hyperref[res:aperiodic-size]{Result~\ref*{res:aperiodic-size}},
compare it with reading all $n$ positions.  If $L(n)\geq N^{1/3}$,
powers of $N$ can be absorbed into powers of $L(n)$; otherwise
$n\leq\sqrt n\,2^{O(N^{1/3})}\leq\sqrt n\,L(n)^{O(N^{1/3})}$.
Either case gives the stated bound after increasing the absolute constant.

\secref{sec:ags-local} proves the localized AGS step and its general
upper bound.  \secref{sec:dyck-lb} derives the elementary monoid-size
lower bound from the Dyck lower bound of Ambainis et al.~\cite{ABIKPSSV20}.
The cube-root proof then develops regular actions and the matrix bound
in Sections~\ref{sec:ags-actions}--\ref{sec:ags-matrix-bound}, followed
by the quotient and lifting arguments in
Sections~\ref{sec:ags-apex-reduction}--\ref{sec:ags-radical-lift}.
\secref{sec:ags-size-recurrence} combines these ingredients and proves
the size bound.  Finally, \secref{sec:ags-structural-examples} explains
how the Dyck and Brandt examples relate to the proof.
\label{page:main-end}

\section{Algebraic preliminaries and basic lower bounds}\label{sec:prelim}

\paragraph{Query model.}
Let $S$ be a semigroup and $G \subseteq S$.   The input is $x = x_1 \cdots x_n \in G^n$.
A query to index $i$ returns the \emph{entire}
element $x_i$ at unit cost---for matrix semigroups this means all entries of the matrix at once.
We write $\operatorname{Prod}_{S,G,n}:G^n\to S$ for the ordered product map over an allowed alphabet
$G\subseteq S$, and abbreviate it to $\operatorname{Prod}_{S,n}$ when $G=S$.  We also write
$p_i = x_1 x_2 \cdots x_i$ for the prefix products, with $p_0 = 1$, and
$[x_i \cdots x_j] = x_i x_{i+1} \cdots x_j$ for the product of an infix.  We write $Q_{1/3}(f)$ for the minimum
number of quantum queries needed to compute $f$ with error probability at most
$1/3$ on every input.

A semigroup $R$ \emph{divides} a semigroup $S$, written $R\prec S$, if there is a
subsemigroup $T\subseteq S$ and a surjective homomorphism $\pi:T\to R$.

\begin{proposition}[Division monotonicity]
\label{prop:division-monotonicity}
  Let $R$ and $S$ be semigroups, let $T\subseteq S$ be a subsemigroup,
  and let $\pi:T\to R$ be a surjective homomorphism, so that $R\prec S$.
  Then, for every $n\ge1$,
  \[
    Q_{1/3}(\operatorname{Prod}_{R,n})\le 2Q_{1/3}(\operatorname{Prod}_{S,n}).
  \]
  More generally, if $G\subseteq R$ and $\sigma:G\to T$ is a section of
  $\pi$ over $G$, meaning that $\pi(\sigma(g))=g$ for every $g\in G$, then
  \[
    Q_{1/3}(\operatorname{Prod}_{R,G,n})
      \le 2Q_{1/3}(\operatorname{Prod}_{S,\sigma(G),n}).
  \]
\end{proposition}

\begin{proof}
  Simulate a coherent query to the lifted letter by querying $x_i$ into a work register,
  adding $\sigma(x_i)$ to the algorithm's answer register by a fixed reversible lookup,
  and unquerying the work register.  Thus each lifted query costs at most two original
  queries.  An
  algorithm for the lifted product returns
  $y=\prod_i\sigma(x_i)\in T$, and free postprocessing gives
  \[
    \pi(y)=\prod_i\pi(\sigma(x_i))=\prod_i x_i.
  \]
  The section need not be multiplicative.  Taking $G=R$ and any set-theoretic section,
  then restricting an algorithm for the full alphabet $S$ to $\sigma(R)$, proves the first
  assertion.
\end{proof}

We use three standard primitives, and one composition principle.  Grover search over $n$ items
finds a marked item (or reports none), and if at least $t$ marked items are \emph{consecutive}
it succeeds with $O(\sqrt{n/t})$ iterations by searching a grid of spacing
$t/2$~\cite{Grover96,BBHT98}.  Minimum finding over $n$ items costs $O(\sqrt{n})$
comparisons~\cite{DH96}.  For evaluations of a fixed function by \emph{bounded-error}
subroutines, including recursive calls, we use the following standard
robust-composition principle: search and minimum finding among $N$ items evaluated by
coherent bounded-error subroutines of query cost $C$ can be performed with
$O(\sqrt{N}\, C \cdot \mathrm{polylog}(N/\epsilon))$ queries and error $\epsilon$, either by
the robust search of H\o yer, Mosca, and de~Wolf~\cite{HMdW03} or by reversible majority-vote
amplification of each evaluation.  We write $\lambda = O(\log(2 + n|M|))$ for a generic
amplification factor in the algorithmic subroutine estimates.  For the
product-breadth bounds and the quantitative AGS bound, we compose adversary
duals and convert to a quantum algorithm only at the end.  In the sampling
constructions, we first fix the random choices to obtain deterministic
functions; see Section~\ref{sec:beta-sampling}.

\paragraph{Ideals and Green's preorders.}
Let $M$ be a monoid.
For $m \in M$, the set $mM = \{m s : s \in M\}$ is the \emph{right ideal} generated by $m$,
$Mm$ the \emph{left ideal}, and $MmM = \{a m b : a, b \in M\}$ the \emph{two-sided ideal}.
Green's equivalences are defined by
\[
	x \mathrel{\RR} y \iff xM = yM, \qquad
	x \mathrel{\LL} y \iff Mx = My, \qquad
	x \mathrel{\JJ} y \iff MxM = MyM ,
\]
and $\HH = \RR \cap \LL$.  A monoid is \emph{$\RR$-trivial} if $xM = yM$ implies $x = y$,
and similarly for the other relations.  In particular, it is \emph{$\HH$-trivial} if
\[
	xM=yM \quad\text{and}\quad Mx=My \quad\Longrightarrow\quad x=y;
\]
the two ideal equalities are required simultaneously.  A finite monoid is aperiodic if and
only if it is $\HH$-trivial; below we only use the direction we prove
(\lemref{lem:singleton}).

Following \cite{AGS19}, define the \emph{rank} of $m \in M$ as
\[
	\rho(m) = |M| - |MmM| \enspace,
\]
the number of elements \emph{outside} the two-sided ideal generated by $m$.  Products generate
smaller ideals, so rank can only grow along a product:

\begin{proposition}\label{prop:rank-mono}
	For any $p, q \in M$ we have $\rho(p) \le \rho(pq)$ and $\rho(q) \le \rho(pq)$.
\end{proposition}

\begin{proof}
	$MpqM \subseteq MpM$ and $MpqM \subseteq MqM$.
\end{proof}

Define $\dJ(M)$ to be the length of the longest chain
\[
	M m_1 M \subsetneq M m_2 M \subsetneq \cdots \subsetneq M m_k M
\]
of principal two-sided ideals, counted as the number of strict containments
$k - 1 \le |M| - 1$.  Similarly $\dR(M)$ denotes the
length of the longest chain of (principal) right ideals.  These are the depth parameters that will
appear in our bounds.  Note that $\dJ$ and $\dR$ can be exponentially smaller than $|M|$: in the
free semilattice $2^{[m]}$ we have $|M| = 2^m$ while $\dR = m$.

\paragraph{Conventions.}
Depth parameters ($\dJ$, $\dR$) count \emph{strict containments} in a chain of ideals; heights
of posets and semilattices ($h$) count the \emph{elements} of a chain.  Throughout, a
bare logarithm $\log x$ in an asymptotic complexity bound is understood
as $\log(2+x)$ when needed to cover small parameters.  An explicitly
shifted argument, such as $\log(n+2)$, is read literally.  Unsubscripted
logarithms are natural logarithms; $\log_2$ denotes base two.

\paragraph{Aperiodicity.}
The following two lemmas isolate what aperiodicity gives us.  The first is a standard
``sandwich'' lemma; the second says that an element of an aperiodic monoid is determined by its
right ideal, its left ideal, and its two-sided ideal.

\begin{lemma}[Sandwich Lemma]\label{lem:sandwich}
	Let $M$ be a finite aperiodic monoid.  If $pqr = q$ then $q = pq$ and $q = qr$.
\end{lemma}

\begin{proof}
	If $pqr = q$, then $p^k q r^k = q$ for any $k > 0$.  Let $N$ be such that $p^N = p^{N+1}$
	and $r^N = r^{N+1}$.  Then $p(p^N q r^N) = p^{N+1} q r^N = p^N q r^N = q$, and
	$p^N q r^N = q$, so $pq = p(p^N q r^N) = q$.  We can similarly argue that $q = qr$.
\end{proof}

\begin{proposition}\label{prop:identity}
	Let $M$ be a finite aperiodic monoid.  If $p_1 p_2 \cdots p_n = 1$ then
	$p_1 = \cdots = p_n = 1$.  Consequently $\rho(m) = 0$ iff $m = 1$.
\end{proposition}

\begin{proof}
	It suffices to handle $n = 2$ and induct.  If $pq = 1$ then $p \cdot 1 \cdot q = 1$, so by
	the Sandwich Lemma $1 = p \cdot 1 = p$ and $1 = 1 \cdot q = q$.

	For the consequence: $\rho(1) = 0$ since $M 1 M = M$.  Conversely if $\rho(m) = 0$ then
	$M m M = M \ni 1$, so $1 = amb$ for some $a, b$, and by the first part $a = m = b = 1$.
\end{proof}

\begin{lemma}[Sch\"utzenberger {\cite[Remark~5]{Sch65}}]\label{lem:singleton}
	Let $M$ be a finite aperiodic monoid and $m \in M$.  Let
	$J_m = \{s \in M : m \not\in MsM\}$.  Then $\{m\} = (mM \cap Mm) \setminus J_m$.
\end{lemma}

\begin{proof}
	$\{m\} \subseteq (mM \cap Mm) \setminus J_m$: clearly $m \in mM$, $m \in Mm$, and
	$m \in MmM$, meaning $m \not\in J_m$.

	$(mM \cap Mm) \setminus J_m \subseteq \{m\}$: let $y \in (mM \cap Mm) \setminus J_m$.  Then
	\begin{enumerate}
		\item As $y \in mM$ we know $y = mq$ for some $q \in M$.
		\item As $y \in Mm$ we know $y = pm$ for some $p \in M$.
		\item As $y \not\in J_m$ we know $m = rys$ for some $r, s \in M$.
	\end{enumerate}
	By (1) and (3) we have $y = m q = (rys)q = r \cdot y \cdot sq$, which by the Sandwich Lemma
	means $y = ry$.  By (2) and (3) we have $y = pm = p(rys) = pr \cdot y \cdot s$, which by the
	Sandwich Lemma means $y = ys$.  Combining with (3), $m = rys = ys = y$.
\end{proof}

\paragraph{Adversary bounds.}
For a function \(f\) with finite output set, we use the general adversary lower bound of
H\o yer, Lee, and \v{S}palek~\cite{HLS07} in the following form.
If \(\Gamma\) is a real symmetric matrix indexed by promise inputs, with
\(\Gamma_{xy}=0\) whenever \(f(x)=f(y)\), and
\[
  (\Delta_i)_{xy}=\mathbf 1[x_i\ne y_i],
\]
then
\[
  Q_{1/3}(f)=\Omega\!\left(
    \frac{\|\Gamma\|}{\max_i\|\Gamma\circ\Delta_i\|}
  \right).
\]
The supremum of the displayed ratio over all such matrices $\Gamma$ is the
negative-weight general adversary bound $\ADVpm(f)$.  The general adversary theorem also
gives the converse direction: under our fixed bounded-error convention, for every function
with finite input and output alphabets,
\begin{equation}
  Q_{1/3}(f)=\Theta\bigl(\ADVpm(f)\bigr),
  \label{eq:adversary-tightness}
\end{equation}
with universal implicit constants~\cite{Rei11,LMRSS11}.  Thus every feasible dual adversary
construction below is, without any further algorithmic conversion or output-size loss, a
quantum query upper bound.  We will sometimes display the sharper adversary estimate and
state its query consequence in the following sentence.

For a deterministic function $f$, an \emph{all-pairs dual of cost $T$}
(or simply a \emph{dual}) consists of vectors satisfying
\begin{equation}\label{eq:all-pairs-dual}
 \sum_{i:x_i\ne y_i}\langle u_{x,i},v_{y,i}\rangle
 =\one{f(x)\ne f(y)}
 \qquad\text{for every }x,y,
\end{equation}
with
\begin{equation}\label{eq:all-pairs-cost}
 \max\left\{\max_x\sum_i\|u_{x,i}\|^2,
              \max_y\sum_i\|v_{y,i}\|^2\right\}\leq T.
\end{equation}
The all-pairs condition includes inputs with the same output, for which
the sum must be zero.  Such vectors certify $\ADVpm(f)\leq T$;
conversely, for finite input and output alphabets a dual exists of cost
at most $2\ADVpm(f)$~\cite{LMRSS11}.  Adversary tightness converts a
cost-$T$ dual into a quantum algorithm with any fixed positive error
using $O(1+T)$ queries.  In the ordered-monoid recursion we compose the
vectors, so this error is incurred only at the final conversion.

We next record the fixed-algebra classification.

\begin{theorem}[Fixed finite-monoid trichotomy]
\label{thm:fixed-monoid-trichotomy}
  Let $S$ be a finite nonaperiodic semigroup.  Then
  \[
    Q_{1/3}(\operatorname{Prod}_{S,n})=\Theta_S(n).
  \]
  The lower bound already holds for a Boolean postprocessing of the product.
  Consequently, for every fixed finite monoid $M$,
  \[
    Q_{1/3}(\operatorname{Prod}_{M,n})=
    \begin{cases}
      0, & |M|=1,\\[2mm]
      \widetilde\Theta_M(\sqrt n),
        & M\text{ is nontrivial and aperiodic},\\[2mm]
      \Theta_M(n), & M\text{ is nonaperiodic}.
    \end{cases}
  \]
  Here $\widetilde\Theta_M(\sqrt n)$ means an $\Omega(\sqrt n)$ lower bound and an
  $O_M(\sqrt n\,(\log(n+2))^{O_M(1)})$ upper bound.
\end{theorem}

\begin{proof}
  A finite semigroup is nonaperiodic precisely when it contains a nontrivial subgroup.
  Choose in such a subgroup its identity $e$ and an element $g\ne e$.  Restrict the input
  letters to $e$ and $g$, and encode them by bits.  Put $r=\lfloor n/2\rfloor$ and promise
  that the Hamming weight is $r$ or $r+1$.  The two possible products are $g^r$ and
  $g^{r+1}$, which are distinct by cancellation in the subgroup; assign them different
  Boolean output values.

  Let $A$ be the inclusion matrix between the weight-$r$ and weight-$(r+1)$ layers of the
  Boolean cube, and use the symmetric adversary with off-diagonal blocks $A$ and
  $A^{\mathsf T}$.
  This bipartite graph is $(n-r,r+1)$-biregular, so
  \[
    \|A\|=\sqrt{(n-r)(r+1)}=\Omega(n).
  \]
  After filtering by any one queried coordinate, its edges form a matching, of norm at
  most $1$.  The adversary lower bound gives $\Omega(n)$ queries; reading the input gives
  the reverse inequality.

  Now let $M$ be a nontrivial aperiodic monoid and choose $a\ne1$.  On the alphabet
  $\{1,a\}$, Proposition~\ref{prop:identity} says that the product is $1$ exactly on the
  all-$1$ input.  Testing the exact product therefore contains \textsc{Or}, and costs
  $\Omega(\sqrt n)$.  The fixed-$M$ upper bound follows, for example, from
  \thmref{thm:ags-cuberoot-size}.  The trivial case is immediate, and the
  nonaperiodic case was proved above.
\end{proof}

Recall that the \emph{aperiodicity index} $\iota(M)$ of a finite aperiodic
monoid $M$ is the least integer $k\ge1$ such that $a^k=a^{k+1}$ for every
$a\in M$.

\begin{theorem}[Aperiodicity-index lower bound]
\label{thm:aperiodicity-index-lower}
	Let \(M\) be a nontrivial finite aperiodic monoid.  For every \(n\ge1\),
	\[
		Q_{1/3}(\operatorname{Prod}_{M,n})
		=\Omega\!\left(\sqrt{n\min\{n,\iota(M)\}}\right)
		=\Omega\!\left(\min\{n,\sqrt{n\iota(M)}\}\right).
	\]
	The lower bound holds on a two-letter alphabet and after Boolean
	postprocessing of the product.
\end{theorem}

\begin{proof}
	Put \(k=\iota(M)\).  If \(k>1\), minimality gives an \(x\in M\) such that
	\(x^{k-1}\ne x^k\).  Equality of two consecutive powers propagates under
	multiplication by \(x\), so
	\[
		x^{j-1}\ne x^j\qquad(1\le j\le k),
	\]
	where \(x^0=1\).
	If \(k=1\), choose any \(x\ne1\), and the same display holds for \(j=1\).

	Restrict the input alphabet to \(\{1,x\}\), encoding a Boolean \(1\) by \(x\)
	and a Boolean \(0\) by the monoid identity.  A Boolean string of Hamming weight \(j\) then has
	product \(x^j\).  For \(1\le r\le\min\{k,n\}\), restrict further to the
	weight-\((r-1)\) and weight-\(r\) layers and give their two distinct products
	opposite Boolean output values.  Let \(A_r\) be the inclusion matrix between
	these layers.  The corresponding bipartite graph is
	\((n-r+1,r)\)-biregular, and hence
	\[
		\|A_r\|=\sqrt{r(n-r+1)}.
	\]
	After filtering on any queried coordinate, its edges form a matching, of norm
	at most one.  The adversary bound therefore gives
	\(\Omega(\sqrt{r(n-r+1)})\) queries.

	Take
	\[
		r=\min\left\{k,\left\lfloor\frac{n+1}{2}\right\rfloor\right\}.
	\]
	If \(k\le(n+1)/2\), the bound is \(\Omega(\sqrt{nk})\); otherwise it is
	\(\Omega(n)\).  This proves both displayed forms.
\end{proof}

\section{General techniques for adversary upper bounds}\label{sec:width}

In this section we give two general techniques for proving adversary upper
bounds.  Both adapt the guessed-decision-tree model of Beigi and
Taghavi~\cite{BT20}.  The first, \thmref{thm:essential-width}, uses
incremental summaries and their essential width to obtain our sharp bound
for commutative monoid products.  The second,
\thmref{thm:bt-subroutine-composition}, allows the decision tree to call
subroutines and is used in the stably ordered case.

\subsection{Incremental summaries and essential width}\label{sec:essential-width}

\begin{definition}[Incremental summary and essential width]\label{def:essential-width}
	Let \(I\) be a finite index set with \(|I|=n\ge1\), let \(\Sigma\) and \(O\)
	be finite input and output alphabets, and let \(f:D\to O\) be a function
	on a nonempty promise domain \(D\subseteq\Sigma^I\).
	An \emph{incremental summary} for \(f\) has a finite state set \(Q\)
	with a fixed initial state \(q_\varnothing\in Q\).
	For each \(i\in I\), it has an update map
	\[
		\delta_i:Q\times\Sigma\longrightarrow Q.
	\]
	It also has a readout map \(g:Q\to O\) and, for each input \(x\in D\)
	and revealed index set \(T\subseteq I\), a summary state \(s_x(T)\in Q\), such that
	\[
		s_x(\varnothing)=q_\varnothing,\qquad
		s_x(T\cup\{i\})=\delta_i(s_x(T),x_i)\quad(i\notin T),\qquad
		f(x)=g(s_x(I)).
	\]
	For input \(x\), the essential positions of a revealed index set \(T\subseteq I\) are
	\[
		\Ess_x(T)=\{i\in T:s_x(T)\ne s_x(T\setminus\{i\})\},
	\]
	and the \emph{essential width} of the summary is
	\[
		\kappa=\max_{x\in D}\max_{T\subseteq I}|\Ess_x(T)|.
	\]
\end{definition}

The state is indexed by the revealed set rather than an ordering of that set; this is the
order-independence condition.

The parameter is deliberately backward-looking.  It does not bound the number of state
changes in every order.  It bounds how many members of an unordered prefix could have
arrived last and changed the state.  This is what makes a uniform random order useful.

We use the all-pairs dual formulation
\eqref{eq:all-pairs-dual}--\eqref{eq:all-pairs-cost}.  The following theorem
is a direct adversary upper bound; its proof does not pass through an
algorithm or the characterization of quantum query complexity.

\begin{theorem}[Essential-width adversary theorem]\label{thm:essential-width}
	If \(f\) has an incremental summary on \(n\ge1\) input positions
	of essential width at most \(B\ge0\), then
	\[
		\boxed{\ADVpm(f)\le16\sqrt{nB}.}
	\]
	Consequently $Q_{1/3}(f)=O(\sqrt{nB})$.
	If the essential width is zero, \(f\) is constant and
	\(\ADVpm(f)=Q_{1/3}(f)=0\).
\end{theorem}

\begin{proof}
	If the essential width is zero, every deletion leaves the summary unchanged.
	Repeated deletion gives \(s_x(I)=s_x(\varnothing)=q_\varnothing\) for every
	input \(x\), so \(f(x)=g(q_\varnothing)\) is constant and
	\(\ADVpm(f)=Q_{1/3}(f)=0\).  We may therefore assume that the essential
	width is positive, and hence \(B\ge1\).

	For a finite label set \(A\), let \(e_\star,(e_a)_{a\in A}\) be orthonormal and put
	\[
		\phi_a=e_\star+e_a,\qquad \psi_a=e_\star-e_a.
	\]
	Then
	\begin{equation}
		\langle\phi_a,\psi_b\rangle=\one{a\ne b},
		\qquad \|\phi_a\|^2=\|\psi_a\|^2=2.                              \label{eq:ineq-gadget}
	\end{equation}
	For state labels, take \(A=Q\), giving the vectors \(\phi_q,\psi_q\)
	for \(q\in Q\).  For output labels, use a separate space with orthonormal
	basis \(\widetilde e_\star,(\widetilde e_o)_{o\in O}\), and define
	\[
		\widetilde\phi_o=\widetilde e_\star+\widetilde e_o,\qquad
		\widetilde\psi_o=\widetilde e_\star-\widetilde e_o
		\qquad(o\in O).
	\]
	The tilded vectors satisfy the same inner-product and norm identities
	as in \eqref{eq:ineq-gadget}.

	Fix a permutation \(\pi=(i_1,\ldots,i_n)\) of \(I\), and write
	\[
		T_t=\{i_1,\ldots,i_t\},\qquad q_t(x)=s_x(T_t),\qquad
		q_0(x)=q_\varnothing.
	\]
	At time \(t\), the branch label is the after-state \(q_t(x)\).  Color the branch black
	if \(q_t(x)=q_{t-1}(x)\), and red otherwise; denote the color by
	\(\chi_t(x)\in\{\cB,\cR\}\).  There is at most one black branch at a node, namely the
	branch whose after-state is the old state.  There may be many red branches, and they
	remain distinguished by their different after-states.  We unfold the tree: a node is
	tagged by the full preceding branch transcript
	\[
		N_t(x)=(q_1(x),\ldots,q_{t-1}(x)),
	\]
	so histories that later reconverge still occupy orthogonal sectors.

	Choose positive weights \(W_t^{\cB},W_t^{\cR}\).  In a two-dimensional color space set
	\begin{align*}
		\alpha_{\cB}&=(W_t^{\cB})^{-1/2}e_{\cB},&
		\alpha_{\cR}&=(W_t^{\cR})^{-1/2}e_{\cR},\\
		\beta_{\cB}&=(W_t^{\cR})^{1/2}e_{\cR},&
		\beta_{\cR}&=(W_t^{\cB})^{1/2}e_{\cB}+(W_t^{\cR})^{1/2}e_{\cR}.
	\end{align*}
	These obey
	\begin{align}
		\langle\alpha_c,\beta_d\rangle
		&=\one{c=\cR\text{ or }d=\cR},                                    \label{eq:color-inner}\\
		\|\alpha_c\|^2&=1/W_t^c,
		&\|\beta_d\|^2&=W_t^{\cR}+\one{d=\cR}W_t^{\cB}.                 \label{eq:color-norms}
	\end{align}
	In mutually orthogonal node sectors define, for the coordinate \(i_t\),
	\begin{align}
		u^\pi_{x,i_t}
		&=e_{N_t(x)}\otimes\alpha_{\chi_t(x)}\otimes
		  \phi_{q_t(x)}\otimes\widetilde\phi_{f(x)},                    \label{eq:scan-u}\\
		v^\pi_{y,i_t}
		&=e_{N_t(y)}\otimes\beta_{\chi_t(y)}\otimes
		  \psi_{q_t(y)}\otimes\widetilde\psi_{f(y)}.                   \label{eq:scan-v}
	\end{align}

	We first check feasibility for this fixed order.  If \(f(x)\ne f(y)\), let \(t\) be
	the first time their after-states differ.  Their node transcripts agree, and the two
	branches cannot both be black, since their old state is common.  The color, state-label,
	and output-label inner products in \eqref{eq:scan-u}--\eqref{eq:scan-v} are therefore all
	one.  Moreover \(x_{i_t}\ne y_{i_t}\): deterministic update from the same old state with
	the same letter would give the same after-state.  Thus this term survives the query
	filter and contributes one.  Earlier terms vanish because their state labels agree;
	later terms vanish because their full node transcripts differ.  If \(f(x)=f(y)\), the
	output gadget kills every term.  Hence the fixed-order vectors satisfy
	\eqref{eq:all-pairs-dual}.  Their squared norms are
	\begin{align}
		\|u^\pi_{x,i_t}\|^2&=4/W_t^{\chi_t(x)},                            \label{eq:scan-u-norm}\\
		\|v^\pi_{y,i_t}\|^2&=4\bigl(W_t^{\cR}
		       +\one{\chi_t(y)=\cR}W_t^{\cB}\bigr).                       \label{eq:scan-v-norm}
	\end{align}

	It remains to average and choose the weights.  Put
	\(Z_t^\pi(x)=\one{\chi_t(x)=\cR}\).  Conditional on the unordered prefix
	\(T_t=T\), the last element \(i_t\) is uniform in \(T\), and
	\[
		Z_t^\pi(x)=1\quad\Longleftrightarrow\quad i_t\in\Ess_x(T).
	\]
	Consequently
	\begin{equation}
		\Pr_\pi[Z_t^\pi(x)=1\mid T_t=T]
		=\frac{|\Ess_x(T)|}{t}\le\frac Bt.                                \label{eq:backward-essential}
	\end{equation}
	Choose
	\[
		W_t^{\cB}=\sqrt{t/B},\qquad W_t^{\cR}=\sqrt{B/t}.
	\]
	Both \eqref{eq:scan-u-norm} and \eqref{eq:scan-v-norm} are bounded by
	\[
		4\left(\sqrt{B/t}+Z_t^\pi(x)\sqrt{t/B}\right).
	\]
	By \eqref{eq:backward-essential}, its expectation over \(\pi\) is at most
	\(8\sqrt{B/t}\).

	Finally take the orthogonal direct sum over all permutations, scaling every block by
	\(1/\sqrt{n!}\).  Equation~\eqref{eq:all-pairs-dual} is preserved, while squared norms
	become averages.  Thus both squared-norm bounds in \eqref{eq:all-pairs-cost} are at most
	\[
		8\sqrt B\sum_{t=1}^n t^{-1/2}\le16\sqrt{nB}.
	\]
\end{proof}

The same proof has a useful weighted form.  If coordinate \(i\) has positive cost \(c_i\),
then for \(B>0\), choosing
\(W^{\cB}_{i,t}=\lambda_t/c_i\), \(W^{\cR}_{i,t}=c_i/\lambda_t\), where
\(\lambda_t=\|c\|_2\sqrt{t/(nB)}\), gives
\begin{equation}
	\ADVpm_c(f)\le16\sqrt B\left(\sum_i c_i^2\right)^{1/2}.              \label{eq:weighted-width}
\end{equation}
For \(B=0\), the function is constant and the same bound is immediate.
The fixed-tree first-divergence vectors are a specialization of the generalized nonbinary
guessed-decision-tree construction of Beigi and Taghavi~\cite{BT20}.  Applying only their
coarse \(O(\sqrt{TG})\) estimate to the random scan would count an expected
\(G\le BH_n\) red branches and leave a \(\sqrt{\log n}\) loss.  The depth-dependent
weights above use the sharper hazard estimate \(\Pr[Z_t=1]\le B/t\) before summing the
norms, which removes that loss.

For maximum finding over a nonempty finite totally ordered alphabet $\Sigma$,
the summary is the largest revealed value, with a new least element for the
empty set.  Only a unique maximum can be essential, so the essential width
is at most one.  Thus \thmref{thm:essential-width} gives
$\ADVpm(\operatorname{MAX}_{n,\Sigma})\le16\sqrt n$ and hence
$Q_{1/3}(\operatorname{MAX}_{n,\Sigma})=O(\sqrt n)$.
The same argument applies to minimum finding; to return an attaining index,
we can include the index in the summary and break ties by index.
Beigi and Taghavi obtain $O(\sqrt{n\log n})$ for minimum finding and ask
whether other choices of weights can improve their
bounds~\cite[Proposition~7 and Section~6]{BT20}.
Our weights answer this question for minimum and maximum finding,
removing the $\sqrt{\log n}$ loss within their decision-tree framework.

There is also a fixed-order form in which the relevant parameter is the total number of
state changes rather than backward essential width.  We use this form for the
$\RR$-trivial product bound in \secref{sec:rtrivial} and the sharp Boolean
unitriangular bound in \secref{sec:boolean-unitriangular}.

\begin{lemma}[Bounded-change scan]\label{lem:bounded-change-scan}
	Let $\varnothing\ne\mathcal D\subseteq\Sigma^n$ be finite.  Suppose that
	$f:\mathcal D\to O$ is computed by a deterministic sequential summary
	\[
		q_0\in Q,\qquad q_i(x)=\delta_i(q_{i-1}(x),x_i),\qquad f(x)=g(q_n(x)).
	\]
	If every trajectory has at most $C\ge1$ indices $i$ for which
	$q_i(x)\ne q_{i-1}(x)$, then
	\[
		\ADVpm(f)\le8\sqrt{nC}.
	\]
	If $C=0$, then $f$ is constant and its adversary bound is zero.
\end{lemma}

\begin{proof}
	Use the state and output inequality gadgets from \eqref{eq:ineq-gadget}, and unfold the
	natural-order computation by the preceding transcript
	$N_i(x)=(q_1(x),\ldots,q_{i-1}(x))$.  Color step $i$ black when
	$q_i(x)=q_{i-1}(x)$ and red otherwise.  In \eqref{eq:color-inner}--
	\eqref{eq:color-norms}, take the fixed weights
	\[
		W^{\cB}=\sqrt{n/C},\qquad W^{\cR}=\sqrt{C/n},
	\]
	and define $u_{x,i},v_{y,i}$ exactly as in \eqref{eq:scan-u}--
	\eqref{eq:scan-v}, with $i_t=i$.  At the first divergence of two state trajectories,
	the old states agree, the current letters must differ, and the two branches cannot both
	be black.  That term contributes one; earlier terms vanish through the state-label
	gadget and later terms through the transcript tag.  The output gadget kills every term
	when the outputs agree.  Thus \eqref{eq:all-pairs-dual} holds.

	If a trajectory has $r\le C$ red steps, the two total squared norms are at most
	\begin{align*}
		4\left(\frac{n-r}{W^{\cB}}+\frac r{W^{\cR}}\right)
		&\le8\sqrt{nC},\\
		4\left(nW^{\cR}+rW^{\cB}\right)
		&\le8\sqrt{nC}.
	\end{align*}
	Equation~\eqref{eq:all-pairs-cost} proves the claim.  If $C=0$, every trajectory stays
	at $q_0$, so the output is constant.
\end{proof}

\subsection{Decision trees with subroutine calls}\label{sec:bt-subroutines}

In our stably ordered algorithm in \secref{sec:beta}, the classical
sampling procedure calls subroutines that summarize smaller input
intervals.  After fixing all random choices, these summaries are
deterministic functions whose adversary bounds are supplied by the
recursive analysis.  We need to use these bounds while counting the
prediction mistakes made by the surrounding procedure.  For example,
in the rejection sampler of \secref{sec:beta-sampling}, we predict that
each sampled summary will be rejected; accepting one is a prediction
mistake.

The following theorem combines the Beigi--Taghavi construction with
adversary dual composition to analyze such procedures.  We treat the
subroutine outputs as coordinates of a virtual input and apply the
decision-tree construction to the calls.  Composing with duals for the
subroutines then charges each call according to its adversary bound.
This preserves the quantum savings obtained at earlier recursive stages.
We compose the duals through the entire recursion and convert to a
quantum algorithm only at the end.

\begin{theorem}[Beigi--Taghavi with subroutine calls]
\label{thm:bt-subroutine-composition}
Let $f:\mathcal D\to O$, where $\mathcal D\subseteq\Sigma^n$, and let
$g_j:\mathcal D\to\Gamma_j$, for $j$ in a finite index set $J$, be
functions with $\ADVpm(g_j)\leq T$ for some $T\geq0$.
All input and output alphabets are finite.
Suppose a deterministic procedure computes $f(x)$, accessing $x$ only
through calls to the functions $g_j$.  The choice of which function to
call may depend on earlier answers.  At each call, designate at most
one outgoing branch as predicted.  Return values may be grouped into
one branch provided the subsequent computation, including retained
data, depends only on that branch.

Suppose every execution on $x\in\mathcal D$ makes at most $q$ calls
and takes an unpredicted branch at most $G$ times, where $q,G$ are
nonnegative integers.  Then $f$ has a dual of cost $O(T\sqrt{qG})$, and
\[
 Q_{1/3}(f)=O(T\sqrt{qG}).
\]
The called functions may depend on overlapping input positions.
If $q=0$, $G=0$, or $T=0$, then $f$ is constant and the cost is zero.
\end{theorem}

\begin{proof}
If $q=0$ or $G=0$, all inputs follow the same path to a leaf.  If $T=0$,
every $g_j$ is constant, so the output is again constant.  Assume henceforth
that $q,G,T>0$.

Regard $g(x)=(g_j(x))_{j\in J}$ as a virtual input.  The procedure becomes
a decision tree computing $H$ on the promise domain
$Z=\{g(x):x\in\mathcal D\}$, with $f=H\circ g$.  Its execution uses at
most $q$ virtual queries and makes at most $G$ prediction mistakes.
The Beigi--Taghavi construction~\cite{BT20}, in the all-pairs form used
in \secref{sec:essential-width}, therefore gives a dual
$(a_{z,j},b_{z,j})$ for $H$ of cost $O(\sqrt{qG})$.
No bound on executions outside $Z$ is needed.

For each $j$, choose a dual $(c^j_{x,i},d^j_{x,i})$ for $g_j$ of cost
at most $2T$.  The standard composition construction~\cite{LMRSS11} is
\[
 u_{x,i}=\bigoplus_j a_{g(x),j}\otimes c^j_{x,i},
 \qquad
 v_{y,i}=\bigoplus_j b_{g(y),j}\otimes d^j_{y,i}.
\]
The inner all-pairs identities turn the query-filtered inner product into
\[
 \sum_{j:g_j(x)\ne g_j(y)}
      \langle a_{g(x),j},b_{g(y),j}\rangle
 =\one{f(x)\ne f(y)}.
\]
Each squared-norm sum grows by at most a factor $2T$, proving the dual
bound.  The construction uses no disjointness assumption on the inputs
to the called functions.  Weak duality and adversary tightness give the
query bound.
\end{proof}

For randomized procedures, we fix all seeds, including those of recursive
calls, before applying the theorem.  We bound the dual cost of each resulting
output function uniformly in the seeds, and analyze the classical success
probability separately.

\subsection{Bounded-change prefix tracking}\label{sec:rtrivial}

Every finite $\RR$-trivial monoid is aperiodic: distinct elements on a nontrivial cycle of
powers would generate the same principal right ideal.  In an $\RR$-trivial monoid, every
change in the prefix product strictly descends through the principal right ideals.  The
bounded-change scan therefore gives a log-free bound governed by the right-ideal depth
$\dR(M)$ rather than by $|M|$.

\begin{theorem}[Log-free $\RR$-trivial product bound]\label{thm:rtrivial}
	Let $M$ be a finite $\RR$-trivial monoid.  Then
	\[
		\ADVpm(\operatorname{Prod}_{M,n})
		\le8\sqrt{n\min\{n,\dR(M)\}},
		\qquad
		Q_{1/3}(\operatorname{Prod}_{M,n})
		=O\!\left(\min\{n,\sqrt{n\dR(M)}\}\right).
	\]
	More generally, let $\varnothing\ne\mathcal D\subseteq M^n$ and define
	$f:\mathcal D\to M$ by $f(x)=x_1\cdots x_n$.  Put
	\[
		C_{\mathcal D}=\max_{x\in\mathcal D}
		\bigl|\{0\le i<n:p_i(x)\ne p_{i+1}(x)\}\bigr|.
	\]
	Then
	$\ADVpm(f)\le8\sqrt{nC_{\mathcal D}}$,
	with value zero when $C_{\mathcal D}=0$.
\end{theorem}

\begin{proof}
	Use the sequential summary $q_i(x)=p_i(x)=x_1\cdots x_i$.  We always have
	$p_iM\subseteq p_{i-1}M$.  If the prefix changes, equality of these two right ideals
	would contradict $\RR$-triviality, so every change produces a strict containment.
	Consequently
	$C_{\mathcal D}\le\min\{n,\dR(M)\}$.  \lemref{lem:bounded-change-scan} proves the
	adversary statements, adversary tightness gives the query bound, and querying all
	letters gives the cap by $n$.
\end{proof}

For a finite semilattice $L$, $\dR(L) + 1$ is its height (recall $\dR$ counts strict
containments while height counts elements), so \thmref{thm:rtrivial} gives
$O(\sqrt{n(h-1)})$.  For the union semilattice
$2^{[m]}$ this is $O(\sqrt{nm})$---tight up to constants, and exponentially better in
the monoid parameters than the original AGS bound in \thmref{thm:main-ags}.
The dependence on $\dR$ is optimal up to constants: the capped-addition monoid
$M_K=(\{0,\ldots,K\},\mathbin{\oplus})$, where $a\oplus b=\min\{a+b,K\}$, has
$\dR(M_K)=K$, and its product problem on binary inputs requires
$\Omega(\sqrt{nK})$ queries for $K\le n/2$ by
\propref{prop:jtrivial-log-fails}.

\section{Commutative monoids}\label{sec:lattice}

We now prove the sharp breadth bound for commutative monoids.  The upper
bound follows from the essential-width adversary theorem
(\thmref{thm:essential-width}).  For commutative products, essential width
equals product breadth.  We then bound breadth using commutativity and
stabilization identities.

\begin{theorem}[Commutative product bound]\label{thm:commutative-beta}
	For every finite commutative aperiodic monoid \(M\), alphabet \(G\subseteq M\), and
	\(n\ge1\),
	\begin{equation}
		\ADVpm(\operatorname{Prod}_{M,G,n})
		\le16\sqrt{n\min\{n,\beta_G(M)\}}.                             \label{eq:comm-width}
	\end{equation}
	Consequently,
	\[
		Q_{1/3}(\operatorname{Prod}_{M,G,n})
		=O\!\left(\min\{n,\sqrt{n\beta_G(M)}\}\right),
	\]
	with a universal implicit constant.
	If $1\in G$, this bound is tight:
	\[
		Q_{1/3}(\operatorname{Prod}_{M,G,n})
		=\Theta\!\left(\min\{n,\sqrt{n\beta_G(M)}\}\right).
	\]
\end{theorem}

\subsection{Commutative products and product breadth}

For the rest of this section, \(M\) is a finite commutative aperiodic monoid and
\(G\subseteq M\) is the allowed alphabet.  For \(x\in G^n\) and \(T\subseteq[n]\), put
\begin{equation}
	 p_x(T)=\prod_{i\in T}x_i.                                             \label{eq:subset-product}
\end{equation}
The empty product is the identity.  Commutativity makes this an order-independent summary,
with update \(\delta(p,a)=pa\), and its final state is the desired product.  Thus
\begin{equation}
	\Ess_x(T)=\left\{i\in T:p_x(T)\ne p_x(T\setminus\{i\})\right\}.       \label{eq:product-essential}
\end{equation}
We first recall the structural consequence of commutative aperiodicity.

\begin{proposition}\label{prop:comm-jtrivial}
	A finite commutative aperiodic monoid is \(\JJ\)-trivial (hence \(\RR\)- and
	\(\LL\)-trivial).
\end{proposition}

\begin{proof}
	In a commutative monoid \(xM=Mx=MxM\), so \(\RR=\LL=\JJ\).  Suppose \(xM=yM\), say
	\(x=yu\) and \(y=xv\).  Then \(x=xvu=v\cdot x\cdot u\) by commutativity, and the
	Sandwich Lemma (\lemref{lem:sandwich}) gives \(x=vx\).  Hence
	\(y=xv=vx=x\).
\end{proof}

\begin{lemma}[Preserving subwords and essential width]\label{lem:comm-beta-width}
	For the subset-product summary, define its essential width over all word lengths by
	\begin{equation}
		\kappa_G(M)=\sup_{s\ge0}\ \sup_{x\in G^s}|\Ess_x([s])|.          \label{eq:monoid-width}
	\end{equation}
	Then \(\kappa_G(M)=\beta_G(M)\).  More precisely, every deletion-essential position
	of a revealed subword is retained by every product-preserving subword of it.
\end{lemma}

\begin{proof}
	A shortest preserving subword has every position deletion-essential, since otherwise
	it could be shortened further.  Applying the definition of \(\kappa_G(M)\) to that
	subword gives \(\beta_G(M)\le\kappa_G(M)\).

	Conversely, fix \(x\) and \(T\), and let \(S\subseteq T\) preserve its product:
	\(p_x(S)=p_x(T)\).  If \(i\in T\setminus S\), commutativity gives the inclusions
	\[
		p_x(S)M\supseteq p_x(T\setminus\{i\})M
		\supseteq p_x(T)M=p_x(S)M.
	\]
	These principal ideals are equal, so Proposition~\ref{prop:comm-jtrivial} gives
	\(p_x(T\setminus\{i\})=p_x(T)\).  Thus \(\Ess_x(T)\subseteq S\).
	Choosing \(S\) shortest proves
	\(|\Ess_x(T)|\le |S|\le\beta_G(M)\), and hence the reverse inequality.
\end{proof}

The equality concerns worst-case parameters.  The size of a shortest preserving subword
of one full input need not bound the essential width of all its revealed subwords.
For example, the union word \(\{1\},\ldots,\{m\},[m]\) is preserved by its last
position alone, but revealing just the first \(m\) positions makes all of them essential.

\begin{proof}[Proof of \thmref{thm:commutative-beta}.]
	The subset-product summary has essential width at most \(\min\{n,\beta_G(M)\}\) by
	\lemref{lem:comm-beta-width}.  Apply \thmref{thm:essential-width} and adversary
	tightness.  When \(\beta_G(M)=0\), every input product is the identity and no query is
	needed.

	For the lower bound, suppose $1\in G$ and $\beta_G(M)\geq1$, and put
	\[
		r=\min\left\{\beta_G(M),\left\lfloor\frac{n+1}{2}\right\rfloor\right\}.
	\]
	Choose a word $w=a_1\cdots a_r\in G^r$ with no proper
	product-preserving subword.  Such a word is obtained by taking a prefix
	of length $r$ of a shortest product core of length at least $r$.
	In particular, none of the $a_i$ is the identity.
	Let $X$ consist of all length-$n$ permutations of the letters of $w$
	and $n-r$ identities.  Let $Y$ consist of the inputs obtained by
	replacing one nonidentity letter of an input in $X$ by the identity.
	Commutativity implies that every input in $X$ has product $[w]$,
	whereas every input in $Y$ has a different product by the choice of $w$.

	Let $A$ be the bipartite adjacency matrix joining inputs in $X$ and $Y$
	that differ at one coordinate.  Every input in $X$ has $r$ neighbours.
	The multiset of letters of an input in $Y$ uniquely determines its
	missing letter, which can be restored at any of its $n-r+1$ identity
	positions.  Thus the graph is $(r,n-r+1)$-biregular and
	\[
		\|A\|=\sqrt{r(n-r+1)}.
	\]
	Filtering by any queried coordinate leaves a matching, of norm at most
	one.  The adversary lower bound therefore gives
	$\Omega(\sqrt{r(n-r+1)})
	 =\Omega(\min\{n,\sqrt{n\beta_G(M)}\})$ queries, proving the claim.
\end{proof}

It remains to bound \(\beta_G(M)\).  In each argument below we take a shortest preserving
subword of an arbitrary input and bound its length.  No proper scattered subword of this
chosen word has the same product.

\subsection{The idempotent case}

Let \(M\) be idempotent.  A commutative idempotent monoid is a join-semilattice, with
\(a\le b\) when \(ab=b\); we write the operation as \(\vee\).  Let
\begin{equation}
	L_G=\left\{\bigvee H:\varnothing\ne H\subseteq G\right\}             \label{eq:nonempty-join-closure}
\end{equation}
be the nonempty join-closure of the allowed letters.  The monoid of reachable summary
states is contained in \(L_G\cup\{1\}\), where \(1\) represents the empty join.

\begin{lemma}[Subset joins of a shortest preserving word]\label{lem:critical}
	Let \(a_1\cdots a_B\) be a shortest preserving subword of a \(G\)-word.  Its
	\(2^B\) subset joins, including the empty join, are all distinct.
\end{lemma}

\begin{proof}
	We show that the map
	\[
		A\longmapsto\bigvee_{i\in A}a_i,\qquad A\subseteq[B],
	\]
	with the empty set sent to the identity, is injective.  If $A\ne C$, choose, after
	interchanging the two sets if necessary, an $i\in A\setminus C$.  Equality of the two
	joins would imply
	\[
		a_i\le\bigvee_{j\in C}a_j
		\le\bigvee_{j\in[B]\setminus\{i\}}a_j,
	\]
	so deleting \(i\) would preserve the product, contradicting minimality.
\end{proof}

\begin{theorem}[Commutative idempotent product]\label{thm:semilattice-product}
	Let $M$ be finite, commutative, and idempotent.  Let $G\subseteq M$, and let $L_G$
	be its nonempty join-closure.  Then
	\begin{equation}
		\beta_G(M)\le\left\lfloor\log_2(|L_G|+1)\right\rfloor.          \label{eq:semilattice-beta}
	\end{equation}
	Consequently,
	\begin{equation}
		\ADVpm(\operatorname{Prod}_{M,G,n})
		\le16\sqrt{n\min\left\{n,\left\lfloor\log_2(|L_G|+1)\right\rfloor\right\}}.
		\label{eq:semilattice-adv}
	\end{equation}
	In particular,
	$Q_{1/3}(\operatorname{Prod}_{M,G,n})=
	O\!\left(\sqrt{n\min\{n,\lfloor\log_2(|L_G|+1)\rfloor\}}\right)$.
\end{theorem}

\begin{proof}
	For a shortest preserving word of length \(B\), \lemref{lem:critical} gives
	\(2^B\) distinct subset joins, all in \(L_G\cup\{1\}\).  Thus
	\(2^B\le |L_G|+1\), proving the bound on \(\beta_G(M)\).
	The adversary and query bounds follow from \thmref{thm:commutative-beta}.
\end{proof}

Here \(L_G\) is generated by the entire allowed alphabet.  A similar bound
holds under a promise on the inputs: if every input generates at most \(K\)
nonempty joins, then the same proof gives an adversary bound of
\(16\sqrt{n\min\{n,\lfloor\log_2(K+1)\rfloor\}}\).

\subsection{The identity \texorpdfstring{\(x^3=x^2\)}{x3=x2}}

We next bound product breadth when every element satisfies \(x^3=x^2\).
The proof is an elementary collision count on the Boolean cube.

\begin{theorem}[Index two]\label{thm:index-two-width}
	Let \(M\) be a finite commutative monoid satisfying \(x^3=x^2\) for every \(x\in M\).
	Then
	\begin{equation}
		\beta_G(M)\le5\log_2|M|,                                          \label{eq:index-two-width}
	\end{equation}
	and, when \(|M|\ge2\), the strict inequality \(\beta_G(M)<5\log_2|M|\) holds.
	Consequently
	\begin{equation}
		\ADVpm(\operatorname{Prod}_{M,G,n})
		\le16\sqrt{n\min\{n,5\log_2|M|\}}.                              \label{eq:index-two-adv}
	\end{equation}
	In particular,
	$Q_{1/3}(\operatorname{Prod}_{M,G,n})=O(\sqrt{n\min\{n,\log(|M|+1)\}})$.
\end{theorem}

\begin{proof}
	Let \(a_1\cdots a_B\) be a shortest preserving subword of an arbitrary \(G\)-word, and put
	\[
		p(S)=\prod_{i\in S}a_i,\qquad S\subseteq[B].
	\]
	We have \(p([B])\ne p([B]\setminus\{i\})\) for every \(i\).  In fact every edge of
	the subset cube is strict: an equality \(p(S)=p(S\setminus\{i\})\), multiplied by the
	complement of \(S\), would collapse the corresponding top edge.

	We need a collision lemma.  Fix \(R\subseteq[B]\), and suppose \(A,C\subseteq R\) obey
	\(p(A)=p(C)\).  Write \(K=A\cap C\), \(X=A\setminus C\),
	\(Y=C\setminus A\), and \(D=R\setminus(A\cup C)\), and use the same lower-case letter
	for the product over a set.  The equality \(kx=ky\) gives
	\[
		p(R)=kxyd=ky^2d.
	\]
	Since \(y^3=y^2\), we have \(p(R)y=p(R)\); symmetrically \(p(R)x=p(R)\).
	If \(i\in Y\), write \(y=a_i y'\).  Then \(p(R)a_i\) and \(p(R)\) generate the same
	principal ideal: one divisibility is automatic and the reverse follows from
	\(p(R)=p(R)a_i y'\).  Proposition~\ref{prop:comm-jtrivial} therefore gives
	\(p(R)a_i=p(R)\).  The same holds for \(i\in X\).  Thus
	\begin{equation}
		p(A)=p(C)\quad\Longrightarrow\quad
		p(R)a_i=p(R)\quad(i\in A\mathbin\triangle C).                      \label{eq:index-two-collision}
	\end{equation}

	Define
	\[
		N_R=\{i\in R:p(R)a_i\ne p(R)\}.
	\]
	The map \(A\mapsto p(A)\) is injective on subsets of \(N_R\), by
	\eqref{eq:index-two-collision}.  Hence \(|N_R|\le L:=\log_2|M|\).

	Now partition all \(2^B\) subsets according to their product.  For \(z\in M\), let
	\[
		\mathcal F_z=\{S\subseteq[B]:p(S)=z\},
		\qquad A_z=\{i\in[B]:za_i=z\}.
	\]
	Every \(S\in\mathcal F_z\) contains \(A_z\), since an absent \(i\in A_z\) would collapse
	the edge from \(S\) to \(S\cup\{i\}\).  Moreover
	\(S\setminus A_z=N_S\), so \(|S\setminus A_z|\le\lfloor L\rfloor\).  Therefore
	\begin{equation}
		2^B\le |M|\sum_{j=0}^{\lfloor L\rfloor}\binom Bj.                 \label{eq:index-two-ball}
	\end{equation}

	If \(M\) is trivial then \(B=0\).  Otherwise \(L\ge1\).  Put
	\(\ell=\lfloor L\rfloor\).  If \(\ell\ge B/2\), then \(B\le2L\).  Otherwise the
	standard Hamming-ball estimate, and monotonicity of
	\(u\mapsto u\log(eB/u)\) for \(u<B\), give
	\[
		\sum_{j=0}^{\ell}\binom Bj
		\le(eB/\ell)^\ell\le(eB/L)^L.
	\]
	Taking base-two logarithms in \eqref{eq:index-two-ball}, and writing \(t=B/L\), yields
	\[
		t\le1+\log_2(et).
	\]
	This excludes \(t\ge5\): at \(t=5\) the right side is smaller than five because
	\(5e<16\), and thereafter its derivative is less than one.  Hence \(B<5L\).
	This bounds every shortest preserving subword, proving the assertion about
	\(\beta_G(M)\).  The query bounds follow from \thmref{thm:commutative-beta}.
\end{proof}

The special proof is doing more than bounding multiplicities.  It uses a collision at an
arbitrary subset to identify factors that already stabilize a larger product, and then
counts the remaining optional coordinates in every product fiber.  This is the step that
does not extend verbatim to larger aperiodicity index.

\subsection{General aperiodicity index}

Assume now that, for some \(k\ge1\),
\begin{equation}
	 x^{k+1}=x^k\qquad\text{for every }x\in M.                             \label{eq:index-k}
\end{equation}
The identity implies aperiodicity.  We use the Bell--Chueluecha--Warnke form of the modern
sunflower bound~\cite{BCW21}: there is an absolute constant \(C_{\rm sf}\) such that every
\(t\)-uniform family with no \(r\)-sunflower has size less than
\begin{equation}
	\bigl(C_{\rm sf}r\log(2t)\bigr)^t                                    \label{eq:bcw-bound}
\end{equation}
for \(r,t\ge2\).  This bound has also been formalized in Lean~\cite{Lee26},
with explicit constants that allow \(C_{\rm sf}=2^{61}\).

\begin{lemma}[Product fibers contain no large sunflower]\label{lem:no-product-sunflower}
	Let \(a_1\cdots a_B\) be a shortest preserving subword of a \(G\)-word, and put
	\(p(S)=\prod_{i\in S}a_i\).  For every \(z\in M\), the product fiber
	\[
		\{S\subseteq[B]:p(S)=z\}
	\]
	contains no \((k+1)\)-sunflower.
\end{lemma}

\begin{proof}
	Suppose \(C\cup P_0,\ldots,C\cup P_k\) were such a sunflower, with pairwise disjoint
	petals \(P_j\) disjoint from \(C\).  Since the \(k+1\ge2\) sets are distinct, at least
	one petal is nonempty; relabel so that \(P_0\ne\varnothing\).  Write \(c=p(C)\) and
	\(b_j=p(P_j)\), with \(p(\varnothing)=1\).  Equality of the fiber values gives
	\[
		cb_0=cb_1=\cdots=cb_k.
	\]
	Using commutativity, these equalities, and \eqref{eq:index-k} gives
	\[
		c\prod_{j=0}^k b_j
		=cb_k^{k+1}=cb_k^k
		=c\prod_{j=1}^k b_j.
	\]
	Let \(D=[B]\setminus(C\cup P_0\cup\cdots\cup P_k)\) be the remaining index set.
	Multiplying both sides by \(p(D)\) gives \(p([B])=p([B]\setminus P_0)\).
	Thus deleting the nonempty petal \(P_0\) preserves the full product, contradicting
	minimality of the chosen word.
\end{proof}

\begin{theorem}[General index \(k\)]\label{thm:index-k-width}
	There is an absolute constant \(C\) such that every finite commutative monoid satisfying
	\eqref{eq:index-k} obeys
	\begin{equation}
		\beta_G(M)
		\le C(k+1)(\log_2|M|+1)\log(\log_2|M|+2).                          \label{eq:index-k-width}
	\end{equation}
	Consequently
	\begin{equation}
		\ADVpm(\operatorname{Prod}_{M,G,n})
		=O\!\left(\sqrt{n\min\{n,
		k\log(|M|+1)\log\log(|M|+2)\}}\right).                           \label{eq:index-k-adv}
	\end{equation}
	The same asymptotic expression upper-bounds
	$Q_{1/3}(\operatorname{Prod}_{M,G,n})$.
\end{theorem}

\begin{proof}
	Take a shortest preserving subword of an arbitrary \(G\)-word, let its length be
	\(B\), and set
	\(t=\max\{2,\lceil\log_2|M|\rceil\}\).  If \(B<t\), the desired estimate is immediate.
	Otherwise, partition the \(t\)-subsets of \([B]\) by their product.  Each resulting
	family is \(t\)-uniform and contains no \((k+1)\)-sunflower by
	\lemref{lem:no-product-sunflower}, so \eqref{eq:bcw-bound} gives
	\[
		\binom Bt
		<|M|\bigl(C_{\rm sf}(k+1)\log(2t)\bigr)^t.
	\]
	Since \(\binom Bt\ge(B/t)^t\) and \(|M|^{1/t}\le2\),
	\[
		B<2C_{\rm sf}(k+1)t\log(2t),
	\]
	which bounds the length of every shortest preserving subword and proves
	\eqref{eq:index-k-width}.  Equation~\eqref{eq:index-k-adv} follows from
	\thmref{thm:commutative-beta}.
\end{proof}

For comparison, the classical Erd\H{o}s--Rado sunflower lemma gives the elementary estimate
\begin{equation}
	\beta_G(M)<2k(\log_2|M|+1)^2,                                        \label{eq:elementary-index-k}
\end{equation}
To see this, let \(B\) be the length of a shortest preserving subword and choose
\(t=\lfloor\log_2|M|\rfloor+1\).  If \(B<t\), the estimate is immediate.
Otherwise, partition the \(t\)-subsets of \([B]\) by their product.  Each family
contains no \((k+1)\)-sunflower by \lemref{lem:no-product-sunflower}, so the
classical lemma bounds its size by \(t!k^t\).  Consequently,
\[
	(B/t)^t\le\binom Bt\le |M|\,t!k^t.
\]
Taking \(t\)-th roots and using \(|M|^{1/t}<2\) and \((t!)^{1/t}\le t\) gives
\[
	B\le kt\,|M|^{1/t}(t!)^{1/t}
	 <2kt^2\le2k(\log_2|M|+1)^2.
\]
The factor \(t\) from \((t!)^{1/t}\le t\) accounts for the second \(\log|M|\).
The BCW bound gives \(O(\log(2t))\) in its place, yielding
\(\log\log(|M|+2)\).

\begin{remark}[Where the hypotheses enter]\label{rem:scope-records}
	Commutativity permits the subset-product manipulations and makes that product an
	order-independent incremental summary.  In the sunflower argument, the identity
	\(x^{k+1}=x^k\) is used only to delete a petal and contradict minimality.
	Idempotence and the case \(k=2\) admit stronger collision arguments than the general
	theorem.  Without a bounded aperiodicity index, a logarithmic bound on \(\beta\) is
	false: the capped counter of height \(K\) needs \(K\) copies of its unit generator
	to preserve their product, but has only \(K+1\) elements.
\end{remark}

\subsection{Capped counters: a matching lower bound}

The preceding dependence is close to optimal.  For integers \(k,r\ge1\), let
\[
	C_{k+1}=\{0,1,\ldots,k\},\qquad a\oplus b=\min\{a+b,k\},
\]
and put \(M_{k,r}=C_{k+1}^r\).  Then \(|M_{k,r}|=(k+1)^r\), every element satisfies
\(x^{k+1}=x^k\), and the least such uniform exponent is \(k\).

\begin{theorem}[Capped-counter products]\label{thm:capped-counter-product}
	The product-preserving subword parameter is
	\[
		\beta_{M_{k,r}}(M_{k,r})=kr.
	\]
	Among words of length \(n\), the largest length of a shortest preserving subword is
	\(\min\{n,kr\}\).
	Moreover,
	\begin{equation}
		Q_{1/3}(\operatorname{Prod}_{M_{k,r},n})
		=\Theta\!\left(\sqrt{n\min\{n,kr\}}\right).                       \label{eq:capped-counter-theta}
	\end{equation}
\end{theorem}

\begin{proof}
	For each coordinate, retain at most \(k\) positions preserving its capped sum.
	If its uncapped total is less than \(k\), retain all positions positive in that
	coordinate; there are at most \(k-1\).  Otherwise retain positive positions until
	their sum reaches \(k\), which takes at most \(k\) positions.  The union of these
	choices has at most \(kr\) positions and preserves every coordinate: adding retained
	positions from other coordinates cannot lower a capped sum or exceed its full-input
	value.  Hence \(\beta_{M_{k,r}}(M_{k,r})\le kr\).

	For equality, take \(k\) copies of each coordinate-unit vector.  Any proper subword
	reduces at least one coordinate below \(k\), so all \(kr\) positions are required.
	For the length-\(n\) statement, take \(\min\{n,kr\}\) coordinate-unit vectors, at most
	\(k\) copies per coordinate, and pad with zero vectors to length \(n\).  Every unit
	vector must be retained.

	The monoid \(M_{k,r}\) is commutative and aperiodic, and its full alphabet contains
	the identity, the zero vector.  Thus \thmref{thm:commutative-beta} gives
	\eqref{eq:capped-counter-theta}.  By \eqref{eq:adversary-tightness}, the same
	asymptotic bound holds for \(\ADVpm(\operatorname{Prod}_{M_{k,r},n})\).
\end{proof}

We now use capped counters to assess the tightness of the general upper
bound on \(\beta\) obtained from sunflower bounds.
Since \(kr=(k/\log_2(k+1))\log_2|M_{k,r}|\), the nonsaturated regime of
\eqref{eq:capped-counter-theta} is
\begin{equation}
	\Theta\!\left(
	\sqrt{n\,\frac{k}{\log_2(k+1)}\log_2|M_{k,r}|}
	\right).                                                              \label{eq:capped-counter-size}
\end{equation}
Thus our general bound on \(\beta\), derived from the sunflower theorem
of~\cite{BCW21}, is optimal, as a function of \(k\) and \(|M|\), up
to a factor \(O(\log(k+1)\log\log(|M|+2))\); for query complexity the remaining ratio is
the square root of that factor.  For \(k=2\), \thmref{thm:index-two-width} removes the
\(\log\log |M|\) loss and matches the capped cube up to an absolute constant.

The one-coordinate case also rules out a bound depending only on the
logarithm of the monoid size.

\begin{proposition}[Logarithmic monoid size does not suffice]
\label{prop:jtrivial-log-fails}
	Let $K\ge 1$.  The capped counter $C_{K+1}$ is a finite commutative
	$\JJ$-trivial monoid
	with $\dR(C_{K+1})=K$, yet on inputs from $\{0,1\}^n$ its product problem requires
	$\Omega(\sqrt{nK})$ quantum queries whenever $K\le n/2$.
\end{proposition}

\begin{proof}
	The monoid is commutative, and the principal right and two-sided ideals generated by
	$a$ both equal
	$aC_{K+1}=C_{K+1}aC_{K+1}=\{a,a+1,\ldots,K\}$.  Thus distinct elements generate distinct principal
	ideals, and these ideals form a chain with $K$ strict containments.  Hence $C_{K+1}$ is
	$\JJ$-trivial and $\dR(C_{K+1})=K$.  The capped-counter breadth calculation gives
    $\beta_{\{0,1\}}(C_{K+1})=K$: the word of $K$ copies of $1$ needs every
    letter, and the upper bound holds for all counter letters.  Applying
    \thmref{thm:commutative-beta} to the alphabet $\{0,1\}$ therefore gives
    $\Omega(\sqrt{nK})$ queries when $K\leq n/2$.
\end{proof}

Consequently, there is no uniform $\widetilde O(\sqrt{n\log |M|})$ bound
for finite $\JJ$-trivial monoids: taking $K=\lfloor\sqrt n\rfloor$ gives
an $\Omega(n^{3/4})$ lower bound, whereas $\sqrt{n\log|C_{K+1}|}$ times
any polylogarithmic factor is $n^{1/2+o(1)}$.

\section{Application: minimum-weight matroid bases}\label{sec:matroid-bases}

The commutative product theorem also applies to optimization problems.
In this section, we show that a minimum-weight basis of a known rank-$r$
matroid can be computed from $n$ weighted input records using
$O(\min\{n,\sqrt{nr}\})$ quantum queries.  The key observation is that
greedy bases can be merged: the greedy basis of a union is determined by
the greedy bases of its two parts.  This gives a commutative idempotent
monoid whose product breadth is the matroid rank.

\subsection{The monoid of greedy bases}

Recall that a finite matroid $\mathcal M=(E,\mathcal I)$ consists of a
finite ground set $E$ and a family $\mathcal I$ of independent sets.
This family contains the empty set, is closed under taking subsets, and
satisfies the exchange axiom: if $A,B\in\mathcal I$ and $|A|<|B|$, some
$e\in B\setminus A$ has $A\cup\{e\}\in\mathcal I$.  A \emph{basis of
$S\subseteq E$} is a maximal independent subset of $S$; all such bases
have the same size, denoted $\rk(S)$.  Write $r=\rk(E)$ and
$\operatorname{cl}(S)=\{e\in E:\rk(S\cup\{e\})=\rk(S)\}$ for the
closure of $S$.

Fix a total order on $E$.  For $S\subseteq E$, let $B(S)$ be the
\emph{greedy basis}: scan the elements of $S$ in increasing order,
retaining an element exactly when it preserves independence.  For every
initial segment $P$ of the order, $B(S)\cap P$ is a basis of $S\cap P$.
Indeed, each rejected element lies in the closure of the previously
retained elements.  Consequently,
\begin{equation}\label{eq:matroid-prefix-closure}
 \operatorname{cl}(B(S)\cap P)=\operatorname{cl}(S\cap P).
\end{equation}
In particular, membership in $B(S)$ is determined by the ranks of these
prefixes: $e$ is retained exactly when the rank increases between the
prefix strictly before $e$ and the prefix ending at $e$.

\begin{lemma}[Greedy bases merge]\label{lem:matroid-greedy-merge}
 For a finite matroid with a fixed total order on its ground set, the
 greedy bases satisfy
 \[
  B(S\cup T)=B\bigl(B(S)\cup B(T)\bigr)
  \qquad(S,T\subseteq E).
 \]
\end{lemma}

\begin{proof}
 Fix an initial segment $P$.  By~\eqref{eq:matroid-prefix-closure},
 $B(S)\cap P$ and $S\cap P$ have the same closure, as do $B(T)\cap P$
 and $T\cap P$.  Taking unions preserves equality of closures, so
 \[
  \operatorname{cl}\bigl((B(S)\cup B(T))\cap P\bigr)
  =\operatorname{cl}\bigl((S\cup T)\cap P\bigr).
 \]
 The two sets therefore have the same rank on every prefix.  Their
 greedy bases select exactly the same elements, by the rank-increase
 characterization above.
\end{proof}

\begin{proposition}[Greedy-basis monoid]\label{prop:matroid-greedy-monoid}
 Let $\mathcal M=(E,\mathcal I)$ be a finite rank-$r$ matroid with a
 fixed total order on $E$.  The independent sets form a commutative
 idempotent monoid $H$ under
 \[
  A*C=B(A\cup C),
 \]
 with identity $\varnothing$.  For the alphabet
 $G=\{\varnothing\}\cup\{B(\{e\}):e\in E\}$,
 \[
  \beta_G(H)=r.
 \]
\end{proposition}

\begin{proof}
 Greedy leaves every independent set unchanged.  Thus $\varnothing$ is
 an identity and $A*A=A$.  Commutativity follows from that of union,
 and \lemref{lem:matroid-greedy-merge} identifies both bracketings of a
 triple product with $B(A\cup C\cup D)$, proving associativity.

 A word over $G$ has product $B(S)$, where $S$ is the union of its
 singleton letters.  Retain one occurrence of each element of $B(S)$.
 The resulting scattered subword has at most $r$ letters and the same
 product, so $\beta_G(H)\le r$.  Conversely, the singleton letters of
 a basis of $E$ give a word of length $r$ whose product changes under
 every proper deletion.  Hence $\beta_G(H)\ge r$.
\end{proof}

\subsection{Weighted records and the query bound}

Let $\mathcal M=(E,\mathcal I)$ be a known finite matroid of rank $r$,
and let $W\subseteq\mathbb R$ be a finite weight alphabet.  The input
consists of $n$ records $x_i=(e_i,w_i)\in E\times W$, with an optional
null record $\bot$.  One query reveals the whole record.  Matroid
computations and arithmetic on revealed weights are free in this query
model.

We seek a minimum-weight basis of the elements that occur in the input,
specified by input indices.  Repeated occurrences of an element are
treated as parallel copies: an index set $I$ is independent if all its
records are nonnull, the elements $e_i$ for $i\in I$ are distinct, and
$\{e_i:i\in I\}\in\mathcal I$.  A basis is a maximal such set, and its
weight is $\sum_{i\in I}w_i$.  We choose a canonical answer by applying
the greedy algorithm in increasing order of $(w_i,i)$.

To obtain a monoid fixed independently of the input, take the public
record universe
\[
 U=E\times W\times[n].
\]
The position label distinguishes repeated occurrences of the same record
$(e,w)$ and provides the input indices used for tie-breaking and for
specifying the output.
We define a matroid on $U$ where a set of records is independent exactly
when the corresponding elements of $E$ are distinct and form an independent
set in $\mathcal M$.
This matroid has rank at most $r$.  Order its elements lexicographically
by $(w,i,e)$, using any fixed total order on $E$ for the last tie-break.
On an actual input, this is precisely the order $(w_i,i)$, since each
position supplies only one record.  Let $B$ and $H$ be its greedy-basis
map and monoid from the preceding subsection.

\begin{theorem}[Minimum-weight matroid bases]\label{thm:matroid-basis-query}
 Let $\mathcal M$ be a known finite matroid of rank $r$, and let
 $W\subseteq\mathbb R$ be finite.  For $n\ge1$, the canonical
 minimum-weight basis of $n$ input records can be computed with bounded
 error using
 \[
  O\!\left(\min\{n,\sqrt{nr}\}\right)
 \]
 quantum queries.  The implicit constant is universal.  The same bound
 permits returning the selected records as well as their indices.
 If $r=0$, the answer is empty and no query is needed.
\end{theorem}

\begin{proof}
 Encode a nonnull record $x_i=(e_i,w_i)$ by the monoid letter
 $B(\{(e_i,w_i,i)\})$, and a null record by the identity.  Their product
 is the greedy basis of all the labelled input records, by
 \lemref{lem:matroid-greedy-merge}.

 This basis has minimum weight.  To see this, compare it with any other
 basis of the input records.  In each order prefix, greedy retains a
 basis of that prefix and hence at least as many elements as the
 competing basis.  Its $j$th selected element therefore precedes or
 equals the competing basis's $j$th element.  The order respects
 weights, so summing these comparisons proves minimum total weight.

 By \propref{prop:matroid-greedy-monoid}, the encoding alphabet has
 product breadth at most $r$ in the finite commutative idempotent
 monoid $H$.  Apply \thmref{thm:commutative-beta}.  A query to an encoded
 letter costs only a constant number of input queries, since the
 position label $i$ is public.  The product supplies the selected
 records and their indices.  When $r=0$, every encoded letter is the
 identity.
\end{proof}

\paragraph{Examples.}
For a uniform rank-$r$ matroid on $n$ distinct input elements, this
computes the $r$ smallest weights.  A partition matroid with capacity one
in each part selects a minimum-weight representative of each occurring
type; truncating its rank to $r$ selects the $r$ cheapest distinct types,
or all types if fewer occur.

For the graphic matroid of the complete graph on a known set of $v$
vertices, the input records are weighted edges and the rank is at most
$v-1$.  The theorem finds a minimum spanning forest using
$O(\min\{n,\sqrt{nv}\})$ record queries.
It gives $O(v^{3/2})$ queries in the adjacency-matrix model and
$O(\sqrt{vm})$ in the adjacency-array model with $m$ edges and public
degrees.  These recover the bounds of D\"urr et al.~\cite{DHHM06}, who
also give the corresponding selection algorithms and matching lower
bounds.  Their time bounds require additional data structures; here we
count only queries.

For the vector matroid on $\mathbb F_q^d$, a query reveals a vector and
its weight.  We obtain a minimum-weight subset forming a basis of the
input span using $O(\min\{n,\sqrt{nd}\})$ queries, uniformly in $q$.
With equal weights, this also computes the span itself, represented by
a basis.  The cost is for whole-vector queries.

\paragraph{Relation to earlier methods.}
The construction gives a common algebraic derivation of these bounds.
It is also related to classical sampling methods for matroid
optimization.  If each element is sampled independently with probability
$1/2$, Karger~\cite[Theorem~2.2]{Karger98} shows that at most $2r$ elements
in expectation either belong to the sample's greedy basis or change that
basis when inserted.

A direct application of the randomized decision-tree theorem of Beigi
and Taghavi~\cite[Theorem~4]{BT20} gives a slightly weaker bound for hidden
weights.  In a uniformly random scan, the basis changes at step $t$
exactly when the newly revealed record belongs to the greedy basis of
the first $t$ records.  Conditional on the set of $t$ revealed positions,
the last position is uniform, so this happens with probability at most
$\min\{1,r/t\}$.  For $1\le r\le n$, the expected number of changes is
therefore
$O(r(1+\log(n/r)))$, giving
$O(\sqrt{nr(1+\log(n/r))})$ queries.  The commutative product theorem
removes this logarithmic loss.  If instead the weights are public and
queries test independence in an unknown matroid, greedy can scan in
weight order and accept at most $r$ elements; the same decision-tree
theorem already gives $O(\sqrt{nr})$ independence queries.

\section{Product breadth and ordered monoid products}\label{sec:beta}

In this section, we prove our second main result,
\hyperref[res:ordered]{Result~\ref*{res:ordered}}, which bounds the quantum
query complexity of products in stably ordered monoids whose identity is
the minimum element.  The bound depends on the product breadth $\beta$ and
the input length, independently of the size of the monoid.  We first prove
a bound with an exponent of $\log n$ linear in $\beta$, then improve this
exponent to $O(\log(\beta+2))$ in
Sections~\ref{sec:beta-rank-doubling}--\ref{sec:beta-doubling-cost}.

We obtain the bound with a linear exponent by analyzing a classical
sampling procedure through the adversary method.
Throughout this section, let $M$ be a possibly infinite monoid with known
multiplication and a known stable partial order in which the identity is the
minimum element.
Fix a finite $G\subseteq M$, and put $\beta=\beta_G(M)$.  The case $\beta=0$ is immediate because every
allowed letter is the identity.  Assume $1\leq\beta<\infty$.
Product breadth and stable order are defined in
\eqnref{eq:monoid-beta} and \defref{def:stable-least-order}, respectively.

\paragraph{The idea.}
We inductively build an algorithm, for integers $r=1,\ldots,\beta$, that finds
a short subword whose product is at least the product of every subword of
length at most $r$.
For $r=1$ this means dominating every individual letter.  To go from $r-1$
to $r$, apply the algorithm for $r-1$ on the two sides of a split and combine
the resulting subwords.  Rejection sampling selects candidate subwords to add
to a growing union, which we finally compress to at most $\beta$ positions.
A prediction-tree argument bounds the adversary cost of the resulting
function after all random choices have been fixed.  These bounds compose
through the recursion; we extract a quantum algorithm only at the root.
At $r=\beta$, dominating all short subwords means having the whole product.
As in quantum maximum finding, the saved candidate never decreases and
the marked sampling mass shrinks.

\subsection{Two elementary observations}

Fix an input $x_1,\ldots,x_n$.  For an index set $U\subseteq[n]$, write
\[
 p(U)=[x_U],
\]
where $x_U$ lists the selected letters in increasing index order.  Inserting a letter
can only increase a product: $1\leq a$ gives $cd\leq cad$.
Consequently
\begin{equation}\label{eq:beta-insertion}
 U\subseteq V\quad\Longrightarrow\quad p(U)\leq p(V).
\end{equation}
In particular, if $D\subseteq U\subseteq V$ and $p(D)=p(V)$, then
$p(U)=p(V)$ by antisymmetry.

Whenever the letters at a set $U$ are known, we can find a set
$D\subseteq U$ with $|D|\leq \beta$ and $p(D)=p(U)$ without further queries:
try subsets in a fixed order until one works.  We call this \emph{compressing}
$U$.

A \emph{record} consists of selected input positions together with their
queried values.  Even when a classical procedure fails to produce a dominating
subword, its returned record still describes a subword of the input.
We denote records by their index sets, with the revealed letters understood.
Every final classical summary below has at most $\beta$ positions, for every
choice of its random seed.  Intermediate records may be larger.  Unions and
compression are functions of these records alone.  We always multiply in
the original position order.

We retain the full records of the accumulated unions after each completed
round until all insertion tests are finished.  The products of the corresponding
subwords, or records describing their cores, need not contain enough information
to evaluate the effect of inserting another subword in the middle.

\begin{lemma}[Only $\beta$ members of a list can be essential]
\label{lem:beta-essential-subword-list}
Let $U_1,\ldots,U_k\subseteq[n]$ be any list of index sets, allowing
overlap and repetitions.  At most $\beta$ indices $j$ satisfy
\[
 p\!\left(\bigcup_{i\ne j}U_i\right)
 \ne p\!\left(\bigcup_iU_i\right).
\]
\end{lemma}

\begin{proof}
Choose a preserving set $D$ of at most $\beta$ positions in the full union.
For each index in $D$, keep one list member containing it.  This keeps
at most $\beta$ members.  Deleting any other member leaves $D$ in the union,
so cannot change its product by~\eqref{eq:beta-insertion}.
\end{proof}

\subsection{Combining a distribution of short subwords}\label{sec:beta-sampling}

We apply \thmref{thm:bt-subroutine-composition} to classical rejection
sampling after fixing all random seeds, including those of recursive
calls.  The following specialization bounds the dual cost of the resulting
deterministic procedure.

\begin{lemma}[Bounded rejection sampling]
\label{lem:beta-rejection-dual}
Consider a deterministic procedure making at most $D\geq1$ rejection
draws, each of which examines at most $k\geq1$ proposals and returns the
first marked record, or reports failure.  Each proposal is a fixed function
of the original input with a dual of cost at most $T\geq1$.  The marking
predicate may depend on previously accepted records.  Rejected record
values are forgotten and do not otherwise affect the continuation.
Then the output of the procedure has a dual of cost $O(TD\sqrt k)$.
The proposal functions may depend on overlapping input positions.
\end{lemma}

\begin{proof}
At each proposal, coalesce all rejected answers into one branch and
predict that branch.  There are at most $Dk$ calls and at most $D$
prediction mistakes: each accepted record ends its draw.
Apply \thmref{thm:bt-subroutine-composition} with $q=Dk$ and $G=D$ to
obtain cost $O(T\sqrt{(Dk)D})=O(TD\sqrt k)$.
For adversary composition, a virtual input may specify conflicting letters
at the same position in different records.  Define their union by keeping
the first recorded letter at each position.  This makes the virtual procedure
total and has no effect on records produced from the same input.
\end{proof}

\begin{lemma}[A short subword dominating almost all samples]
\label{lem:beta-sampling}
Fix an ambient length $N\geq2$ that is a power of two, with
$\beta\leq N$, and put $L=1+\lceil\log_2 N\rceil$.
Suppose a finite random seed $\omega$
specifies a record $U_\omega(x)$ of at most $\beta$
positions.  For each fixed input $x$, write $\mu$ for the distribution
of $U_\omega(x)$ over the random choice of $\omega$.
Assume that, for every fixed $\omega$, the function
$x\mapsto U_\omega(x)$ has a dual of cost at most $T\geq1$.
Then for every integer $1\leq h\leq N$, a finite random seed $\theta$ can
specify a record $K_\theta(x)$ of at most $\beta$ positions such that:
for every fixed $\theta$, its dual cost is
\[
 O\!\left(\beta T\sqrt h\,L^{3/2}\right);
\]
and for each input $x$, with probability at least $1-1/(100N)$ over $\theta$,
\begin{equation}\label{eq:beta-small-mass}
 \Pr_{U\sim\mu}\{p(U)\nleq p(K_\theta)\}<\frac1{2h}.
\end{equation}
The implicit constant is absolute.  Neither the sampled record nor the output
record is required to be the same on different successful seeds.
\end{lemma}

\Needspace{10\baselineskip}
\begin{proof}
First consider ideal conditional draws.  Initialize the saved union
$S=\varnothing$, call all possible records \emph{marked}, and repeat
the following round:
\begin{enumerate}
 \item Draw $2\beta$ independent records from $\mu$ conditioned on being marked.
       Let $A$ be the union of their index sets, put $S'=S\cup A$,
       and save the full record of $S'$.
 \item A previously marked record with index set $U$ remains marked precisely when
       \[
        p(S'\cup U)\ne p(S').
       \]
       Then replace $S$ by $S'$.
\end{enumerate}
Once unmarked, a record stays unmarked.  All tests use only the
record being tested and the saved records of successive unions.
Keep the marked set fixed during the draws of a round, and stop the ideal
process at zero marked mass.  At the end, return a compression $K$ of
the final saved union $S$.

\smallskip
\noindent\emph{Products of discarded records are dominated.}
If a record with index set $U$ is discarded in a round, then
\[
 p(U)\leq p(S'\cup U)=p(S')\leq p(K).
\]
The last inequality holds because the final saved union contains $S'$
and has the same product as its compression $K$.
Thus the probability in~\eqref{eq:beta-small-mass} is at most the final marked
mass.

\smallskip
\noindent\emph{The expected marked mass halves.}
Fix the current saved union $S$ and marked set $V$ of positive mass
$w=\mu(V)$.
For the analysis, augment the round's $2\beta$ independent draws from
$\mu(\,\cdot\mid V)$ with one further independent test sample from the
same distribution.  The test sample survives the round exactly when
removing it changes the product of the union of $S$ and all $2\beta+1$
samples.  Apply \lemref{lem:beta-essential-subword-list} to the list
consisting of $S$ and these samples.  At most $\beta$ of the sampled sets
can be essential while $S$ is retained.  Symmetry among the $2\beta+1$
random samples gives
\begin{equation}\label{eq:beta-halving}
 \mathbb E\bigl[\mu(V_{\rm next})\mid S,V\bigr]
 \leq \frac{\beta}{2\beta+1}\,w<\frac w2.
\end{equation}
This argument permits nonuniform distributions, overlapping sets, and
duplicate samples.

\smallskip
\noindent\emph{Bounded rejection draws.}
Implement each draw by trying $k$ independent proposals from $\mu$ and
returning the first marked one.  If the current marked mass is $w>0$, then
\[
 \Pr\{\text{failure}\}=(1-w)^k,
 \qquad
 \Pr\{\text{output}=U\mid\text{success}\}
   =\frac{\one{U\text{ marked}}\mu(U)}w.
\]
Indeed, the unconditional probability of returning a particular marked
$U$ is $\mu(U)\sum_{j=0}^{k-1}(1-w)^j$.  Thus a successful draw has
exactly the required conditional law.

Set
\[
 a=\frac1{2h},\qquad
 R=\left\lceil\log_2(400Nh)\right\rceil,\qquad
 \ell=\left\lceil\log_2(400\beta NR)\right\rceil,
 \quad k=2h\ell.
\]
Run at most $R$ rounds.  If a draw fails, discard the incomplete batch,
stop, and compress the current saved union.  While $w\geq a$,
\[
 (1-w)^k\leq(1-a)^k\leq2^{-\ell}\leq\frac1{400\beta NR}.
\]
Couple each bounded draw to an ideal conditional draw until a failure.
There are at most $D=2\beta R$ draws, so the probability of a failure while
the marked mass is at least $a$ is at most $1/(200N)$.
For the ideal process,~\eqref{eq:beta-halving} and Markov's inequality give
\[
 \Pr\{\text{marked mass after $R$ rounds}\geq a\}
 \leq\frac{2^{-R}}a\leq\frac1{200N}.
\]
Once the mass is below $a$, stopping is harmless and further rounds
only decrease it.  The two error bounds prove~\eqref{eq:beta-small-mass}.
This coupling is step by step; conditioning an entire execution on all
draws succeeding would in general bias its history.

\smallskip
\noindent\emph{Dual cost.}
Preassign an independent sampling seed to every possible proposal slot.
For any fixed assignment, the procedure is a deterministic function of
the input, to which \lemref{lem:beta-rejection-dual} applies.  Since
$R,\ell=O(L)$ for $\beta,h\leq N$, its dual cost is
\[
 O(TD\sqrt k)=O(\beta T\sqrt h\,R\sqrt\ell)
             =O(\beta T\sqrt h\,L^{3/2}).
\]
Every output record describes a subword of the input, including on seeds
where a draw fails or the small-mass guarantee fails.  Keeping the square
root of the proposal budget is the source of the improved logarithmic exponent.
\end{proof}

\subsection{From single letters to the whole product}

Call a known set $K$ an \emph{$r$-summary} of an interval $I$ if
$|K|\leq \beta$, $K\subseteq I$, and
\begin{equation}\label{eq:beta-r-summary}
 p(U)\leq p(K)
 \qquad\text{for every }U\subseteq I\text{ with }|U|\leq r.
\end{equation}
The summary need not be unique.

\begin{theorem}[Ordered products: a linear exponent]\label{thm:ordered-beta-product}
Let $M$ be a monoid with a known stable partial order in which the identity
is the minimum element.  Let $G\subseteq M$ be finite, and put
$\beta=\beta_G(M)<\infty$.
If $\beta=0$, the product is the identity and no queries are needed.
Suppose $\beta\geq1$.  There is an absolute constant $C$ such that,
for every integer $n\geq1$,
\[
 Q_{1/3}(\operatorname{Prod}_{M,G,n})
 \leq\min\left\{n,\ (C\beta)^\beta\sqrt n\,
                 \bigl(1+\lceil\log_2(n+2)\rceil\bigr)^{5\beta/2}\right\}.
\]
Moreover, there is an algorithm achieving this bound which also returns
a product-preserving subword of at most $\beta$ positions, with probability
at least $9/10$.
\end{theorem}

\begin{proof}
If $\beta=0$, every allowed letter is the identity and no queries are needed.
If $\beta\geq n$, reading the input gives both the product and a preserving
subword of at most $\beta$ positions.  We may therefore assume $1\leq\beta<n$.

For each rank $r$ and interval $I$ of
length $m$, we construct a family of deterministic record functions
$F_{r,I,\omega}(x)$, indexed by a finite random seed.  Every seed returns
a record of at most $\beta$ positions in $I$.  For each input,
the record is an $r$-summary with probability at least $99/100$ over
the seed.  Separately, for every fixed seed we bound the dual cost of
the exact function $x\mapsto F_{r,I,\omega}(x)$ by $A_r(m)$.
This last requirement includes seeds on which domination fails.

Pad the input with known identities to a power-of-two length
$N$, where $n\leq N<2n$.  Padded positions cost no queries and are always
omitted from returned sets.  Use the balanced binary interval tree on
these $N$ positions, and put $L=1+\log_2 N$.

\paragraph{Base case: dominate all individual letters.}
On an interval of length $m$, generate a uniformly random position and
return its letter record.  On padding, return the empty record.  For
each fixed position this function has constant dual cost.  Apply
\lemref{lem:beta-sampling} with $h=m$ and $T=O(1)$.
If any letter were not dominated by the returned product, the probability
of an undominated sample would be at least $1/m$, contradicting
\eqref{eq:beta-small-mass}.  Thus we have a $1$-summary with probability
at least $99/100$, and dual cost
\begin{equation}\label{eq:beta-base-cost}
 A_1(m)\leq C_0\beta\sqrt m\,L^{3/2}.
\end{equation}
One-position intervals simply return their exact letter record, at every rank.

\paragraph{Induction: dominate subwords of length at most $r$.}
Suppose the $(r-1)$-summary procedure is available.  Consider one depth
of the interval tree on an interval $I$ of length $m$.  There are $h$
nodes at that depth, each with two children of length $m/(2h)$.
Generate a sample as follows:
\begin{enumerate}
 \item Choose one of the $h$ nodes uniformly.
 \item Independently choose a seed for the $(r-1)$-summary on each child.
 \item Compress the union of the two returned sets and output it.
\end{enumerate}
For each fixed choice of node and child seeds, combining the two child
duals gives cost $T=O(A_{r-1}(m/(2h)))$.  Here and when combining depths,
we use the elementary sum rule: a function of several outputs has dual
cost at most a constant times the sum of their dual costs.  Tagging
the witnesses by the preceding outputs makes only the first differing
component contribute; the output-inequality gadget
in~\eqref{eq:ineq-gadget} handles postprocessing.
With probability at least $98/100$, both child summaries are correct.
In that event the sample's product dominates $p(U)$ for every index
set $U$ in that node having at most $r-1$ positions in each child:
stability gives
\begin{equation}\label{eq:beta-split}
 p(U)=p(U_{\rm left})p(U_{\rm right})
 \leq p(K_{\rm left})p(K_{\rm right}).
\end{equation}
The right side is the product of the returned sample, because all left
positions precede all right positions.

Apply \lemref{lem:beta-sampling} to this distribution, with this value of
$h$, to obtain $K_d$ for the chosen depth.  On its success event, $K_d$
dominates \emph{all} subwords described in~\eqref{eq:beta-split}, at
\emph{all} nodes of that depth.  Indeed, if $p(U)$ for one such index
set were not dominated, every successful sample from its node would be undominated.
That has probability at least
\[
 \frac1h\cdot\frac{98}{100}>\frac1{2h},
\]
contradicting~\eqref{eq:beta-small-mass}.  This is why we need neither a
fixed successful output for each node nor a union bound over all its
possible target subwords.

Do this separately at every nonleaf depth, then compress the union of
their records $K_d$.  There are fewer than $N$ depths, each failing
with probability at most $1/(100N)$, so all their guarantees hold
together with probability at least $99/100$.
Every index set $U$ with $2\leq|U|\leq r$ has positions in both children
of the smallest tree interval containing its first and last positions.
Each child then contains at most $r-1$ of its positions.  A singleton
is covered at its leaf's parent, and the empty product is the minimum element.
Thus the returned union, and its compression, is an $r$-summary.
We have not taken a union bound over all child calls.  Unsuccessful
child seeds are already included in the distribution of sampled records,
and its mass of successful records is what the argument uses.  Thus the
success probability remains $99/100$ at every rank.

\paragraph{Dual cost.}
At a depth with $h$ nodes, \lemref{lem:beta-sampling} gives cost at most
\[
 C_1 \beta L^{3/2}\sqrt h\,A_{r-1}\!\left(\frac m{2h}\right),
\]
where $C_1$ absorbs the cost of combining the children.
For $r\geq2$ and $m\geq2$, summing over depths gives
\[
 A_r(m)\leq C_2\beta L^{3/2}
   \sum_{d<\log_2 m}\sqrt{2^d}\,
       A_{r-1}\!\left(\frac m{2^{d+1}}\right).
\]
The important cancellation is
\[
 \sqrt h\,\sqrt{\frac m{2h}}=\sqrt{\frac m2}.
\]
Consequently,~\eqref{eq:beta-base-cost} and induction give
\[
 A_r(m)\leq (C_3\beta)^r\sqrt m\,L^{5r/2}.
\]
Each rank adds $L^{3/2}$ from rejection sampling and at most $L$ from
summing over depths.  The constants are absolute; no carrier-size or
seed-count parameter appears.  All these bounds concern exact
deterministic functions with their seeds fixed, so no quantum errors
accumulate during this induction.

\paragraph{Finish at $r=\beta$.}
Choose an index set $U\subseteq[n]$ of size at most $\beta$ whose subword
preserves the full input product.
A successful $\beta$-summary $K$ satisfies
\[
 [x_1\cdots x_n]=p(U)\leq p(K)\leq[x_1\cdots x_n].
\]
The last inequality uses that $K$ is an actual subword.  Antisymmetry
therefore gives equality, so a random root seed returns a preserving
record with probability at least $99/100$.

For each fixed root seed, adversary tightness now gives a quantum
algorithm computing its exact record function with error at most
$1/16$ in $O(1+A_{\beta}(N))$ queries.  First choose a seed according to the
classical distribution, then run this algorithm.  Its probability of
returning a preserving record is at least
\[
 \frac{99}{100}\cdot\frac{15}{16}>\frac9{10}.
\]
The number of seeds has no effect on the query bound.  Padding coordinates
are fixed before this conversion, so the resulting algorithm uses the
original input oracle and makes no queries to padding.
If desired, replace any malformed output or one with more than $\beta$
positions by the empty record, then requery the reported positions.
This ensures that the recorded letters agree with the input even on error
branches, at an extra cost of at most $\beta$.
Putting $r=\beta$ and using $N<2n$ gives the claimed query bound, with this
extra cost absorbed by the absolute constant.  We may also read and compress
the whole input, giving the minimum with $n$.
\end{proof}

\subsection{Median frontiers and rank doubling}\label{sec:beta-rank-doubling}

The preceding construction raises the summary rank by one.  We now
raise it from $r$ to $2r$ at a cost polynomial in $\beta$ and $\log n$.
Only $O(\log(\beta+2))$ stages are then needed to obtain the whole product.

We continue to use the $r$-summaries defined in~\eqref{eq:beta-r-summary}.
As before, we construct finite families of deterministic record functions,
prove their success probability over the random seeds, and bound their
dual cost for every fixed seed.  The finite alphabet and input length
make the record spaces finite, even when $M$ is infinite.

For the next lemma, the reader may wish to refer to
Figure~\ref{fig:beta-median-frontier} in the overview, which illustrates
the frontier construction.

\begin{lemma}[Median-frontier coverage]\label{lem:beta-median-frontier}
Put a complete binary interval tree on a word of power-of-two length
$N\geq2$.  Assign an $r$-summary $K_I$ to each nonroot interval $I$.
For an internal vertex $v$, let its \emph{frontier} $\mathcal F(v)$ consist
of the two children of $v$ and all siblings off the root-to-$v$ path.
These intervals partition the word.  Let
$C_v=\bigcup_{I\in\mathcal F(v)}K_I$, retaining the full record of this union.
Then every subword of length at most $2r$ has product at most $p(C_v)$
for some internal vertex $v$.
\end{lemma}

\begin{proof}
Let $U=\{i_1,\ldots,i_s\}$ be the target index set, where
$i_1<\cdots<i_s$ and $2\leq s\leq2r$, and put $j=\lfloor s/2\rfloor$.
Take $v$ to be the lowest common
ancestor of $i_j$ and $i_{j+1}$.  Its left child contains at most $j$
target positions, and its right child at most $s-j$, both at most $r$.
Each sibling off the root-to-$v$ path lies entirely before or after $v$,
so contains at most $j-1$ or $s-j-1$ target positions.
Thus every frontier interval contains at most $r$ target positions.

In each frontier interval $I$, the summary property gives
$p(U\cap I)\leq p(K_I)$.  Multiplying these inequalities in interval
order and using stability gives
\[
 p(U)=\prod_{I\in\mathcal F(v)}p(U\cap I)
 \leq\prod_{I\in\mathcal F(v)}p(K_I)=p(C_v).
\]
For a singleton, take the parent of its leaf.  The empty product is the
minimum element and is therefore dominated by every $p(C_v)$.
The conclusion holds for all targets simultaneously
whenever all interval summaries are correct.
\end{proof}

The lemma reduces rank doubling to combining one record per internal
vertex.  These records share interval summaries.  We account for this
sharing directly in the dual, with all summary seeds fixed.

\subsection{A dual for search on a tree}\label{sec:beta-shared-search}

We first record why the costs of exact adaptive computations add.
All duals below satisfy the all-pairs identity~\eqref{eq:all-pairs-dual}.

\begin{lemma}[Adaptive sums]\label{lem:beta-adaptive-sum}
Suppose a deterministic computation makes $s$ calls.  For each fixed
transcript of earlier answers, the function called at step $j$ has a
dual of cost at most $T_j$.  Then the full answer transcript has a dual
of cost at most $\sum_{j=1}^sT_j$.  Any function of that transcript has
a dual of cost at most $2\sum_{j=1}^sT_j$.
\end{lemma}

\begin{proof}
Take the direct sum of the call duals, tagging each summand by the
step and the complete preceding answer transcript.  For two inputs
with different transcripts, only the first differing answer contributes:
earlier answers agree, and later calls have different tags.
The all-pairs identity makes the contribution exactly one.  Each input
uses one summand per step, so the squared-norm costs add.
For postprocessing, tensor with the output-inequality vectors
$\phi_a=e_\star+e_a$ and $\psi_a=e_\star-e_a$ from
\eqref{eq:ineq-gadget}.  Their inner product is $\one{a\ne b}$ and
their squared norms are two.  This incurs the stated factor only once,
after the full transcript has been formed.
\end{proof}

The following theorem specializes the learning-graph-to-adversary
construction of Belovs and Lee~\cite[Theorem~9]{BL11}, combined with
adversary composition~\cite{LMRSS11}.
Associate a subcomputation $g_v$ with each tree vertex $v$, and regard
its output as one coordinate of a virtual input.  Add an initial vertex
representing that no outputs are known, and let the edge entering $v$
represent evaluating $g_v$.  A root path ending at a marked vertex
provides a $1$-certificate for the existence of a marked vertex.
On each positive input, send one unit of flow along one such path.

For edge weights $a_v^2$, the negative learning-graph complexity is
$\sum_v a_v^2$, and the cost of the chosen positive flow is the sum
of $a_v^{-2}$ along its path.  Composing with a dual of cost $t_v$ for
each $g_v$ changes these charges to $t_va_v^2$ and $t_v/a_v^2$, respectively,
giving the weighted estimate~\eqref{eq:beta-tree-weighted-cost} below.
We explicitly optimize the resulting worst-path bound to obtain the
recursive formula~\eqref{eq:tree-recursive-cost}.  The direct proof
below enforces the all-pairs constraints and allows the functions $g_v$
to share input positions.

\begin{theorem}[Tree search with shared subcomputations]
\label{thm:weighted-tree-search}
Let $\mathcal T$ be a finite rooted tree with root $\rho$, and let
$P(v)$ be the root-to-$v$ path, including both endpoints.
For each vertex $v$, let $g_v:\mathcal D\to\Gamma_v$, where
$\mathcal D\subseteq\Sigma^n$, be a deterministic function
with an all-pairs dual of cost at most $t_v>0$.
All input and output alphabets are finite.
A vertex $v$ is marked when a fixed predicate of
$(g_w(x))_{w\in P(v)}$ holds.  Let $F(x)=1$ if some vertex is marked,
and $F(x)=0$ otherwise.  Define
\begin{equation}\label{eq:tree-recursive-cost}
 \mathcal C_v=t_v+
 \sqrt{\sum_{u\text{ child of }v}\mathcal C_u^2},
\end{equation}
with the empty sum equal to zero, so $\mathcal C_v=t_v$ at a leaf.
Then $F$ has an all-pairs dual of cost at most $\mathcal C_\rho$, and
\[
 Q_{1/3}(F)=O(\mathcal C_\rho).
\]
The functions $g_v$ may depend on overlapping input positions.
\end{theorem}

\begin{proof}
For a virtual table $z=(z_v)_{v\in\mathcal T}$, write
$\widehat F(z)=1$ if some vertex is marked, and let $\tau_z(w)$ be
the tuple of entries $z_u$ at the strict ancestors $u$ of $w$, in path order.
Choose a marked vertex $v_z$ for each positive table $z$.
For arbitrary positive weights $a_w$, use two orthogonal sectors and put
\[
\begin{array}{lll}
 \widehat F(z)=1:
 &u_{z,w}=\one{w\in P(v_z)}a_w^{-1}\ket{0,\tau_z(w)},
 &v_{z,w}=\one{w\in P(v_z)}a_w^{-1}\ket{1,\tau_z(w)},\\[2pt]
 \widehat F(z)=0:
 &u_{z,w}=a_w\ket{1,\tau_z(w)},
 &v_{z,w}=a_w\ket{0,\tau_z(w)}.
\end{array}
\]
Equal-output pairs are orthogonal.  If $\widehat F(z)=1$ and $\widehat F(z')=0$,
the tables must differ somewhere on $P(v_z)$, since otherwise $v_z$
would also be marked in $z'$.  Among the vertices $w$ on this path with
$z_w\neq z'_w$, only the first has $\tau_z(w)=\tau_{z'}(w)$.  Its contribution is
$a_w^{-1}a_w=1$, and all other contributions vanish.  The reverse
orientation is identical.  Thus these vectors satisfy the all-pairs
dual identity on virtual tables.

Compose with the duals for $g_w$ using the shared-input construction in the
proof of \thmref{thm:bt-subroutine-composition}.  The resulting cost is at most
\begin{equation}\label{eq:beta-tree-weighted-cost}
 \max\left\{
   \max_v\sum_{w\in P(v)}\frac{t_w}{a_w^2},\quad
   \sum_w t_wa_w^2
 \right\}.
\end{equation}
Put $\lambda_w=t_w/a_w^2$.  A common rescaling of the $\lambda_w$
balances the two terms in~\eqref{eq:beta-tree-weighted-cost}.
Thus optimizing this estimate over the weights gives
\begin{equation}\label{eq:tree-optimal-cost}
 \mathcal W(\mathcal T,t)=
 \inf_{\lambda_w>0}
 \sqrt{\left(\max_v\sum_{w\in P(v)}\lambda_w\right)
       \left(\sum_w\frac{t_w^2}{\lambda_w}\right)}.
\end{equation}
We show that the infimum is attained and equals $\mathcal C_\rho$.
By homogeneity, normalize the maximum path sum to one.
For a single vertex the minimum squared cost is $t_\rho^2$, attained
at $\lambda_\rho=1$.  Otherwise, set $\lambda_\rho=s\in(0,1)$.
Each child subtree then has path budget at most $1-s$.
By induction and scaling, its minimum contribution to the total sum
is $\mathcal C_u^2/(1-s)$.  Consequently,
\[
 \mathcal W(\mathcal T,t)^2
 =\min_{0<s<1}
   \left\{\frac{t_\rho^2}{s}
          +\frac{\sum_{u\text{ child of }\rho}\mathcal C_u^2}{1-s}\right\}
 =\left(t_\rho+
        \sqrt{\sum_{u\text{ child of }\rho}\mathcal C_u^2}\right)^2
 =\mathcal C_\rho^2.
\]
The minimum is attained at $s=t_\rho/\mathcal C_\rho$.
This proves the dual bound, and adversary tightness gives the query bound.
\end{proof}

The recurrence adds the work at a vertex once and combines the costs
of its child branches through the square root of the sum of their
squares.  For a chain it gives $\mathcal C_\rho=\sum_vt_v$; for a root
whose children are leaves it gives
$\mathcal C_\rho=t_\rho+\sqrt{\sum_{v\ne\rho}t_v^2}$.
The following simpler estimate suffices for our application.

\begin{corollary}[Depth bound for tree search]
\label{cor:beta-tree-search-depth}
In the setting of \thmref{thm:weighted-tree-search}, if $\mathcal T$
has depth at most $d$, then
\begin{equation}\label{eq:beta-tree-budget}
 \mathcal C_\rho\leq B=\sqrt{(d+1)\sum_{v\in\mathcal T}t_v^2}.
\end{equation}
In particular, $F$ has a dual of cost at most $B$.
\end{corollary}

\begin{proof}
Take $\lambda_w=1$ in~\eqref{eq:tree-optimal-cost}.
Every root path has at most $d+1$ vertices.
\end{proof}

The recursive bound can be smaller: on a complete binary tree with
$N$ vertices, all leaves at depth $d$, and $t_v=1$ for every vertex,
it gives $\mathcal C_\rho=\sum_{h=0}^d2^{h/2}=O(\sqrt N)$,
whereas the depth bound is $O(\sqrt{N\log(N+1)})$.
Both bounds remain valid when only a specified set of vertices may be
marked, by incorporating that restriction in the marking predicates.
For a fixed vertex $v$, any record determined by the outputs along
$P(v)$ has dual cost at most
\[
 2\sum_{w\in P(v)}t_w\leq 2\mathcal C_\rho\leq 2B,
\]
by \lemref{lem:beta-adaptive-sum} and~\eqref{eq:tree-recursive-cost}.
This accounts for retrieving a record after finding its label.

\begin{lemma}[Seeded batch sampling on a tree]
\label{lem:beta-tree-sampling}
In the setting of \corref{cor:beta-tree-search-depth}, suppose there
are $q\geq1$ labelled vertices, each with a record $C_v(x)$
determined by $(g_w(x))_{w\in P(v)}$; its size need not be at most $\beta$.
For every integer $R\geq1$, a finite random seed specifies a record $K$
of at most $\beta$ positions such that
\[
 \Pr\{p(C_v)\leq p(K)\text{ for every labelled vertex }v\}
 \geq 1-q2^{-R}.
\]
For each fixed seed, the exact output function has dual cost
\begin{equation}\label{eq:beta-tree-combine}
 O\!\left(\beta RB\log(q+2)\right).
\end{equation}
These statements hold for any fixed choice of the functions $g_v$ returning
records of the input, whether or not they satisfy a summary guarantee.
On inconsistent virtual records, use the fixed tie-breaking rule
from the proof of \lemref{lem:beta-rejection-dual} to define unions.
\end{lemma}

\begin{proof}
Given a marked set, choose a uniformly random permutation of its ambient
$q$ labels.  Binary-search for the first marked label in that permutation,
using an initial test for emptiness.  Each of the $O(\log(q+2))$
decisions asks whether some vertex is marked and belongs to a specified
initial segment.  For fixed earlier answers and saved records, its
predicate depends only on the outputs along that vertex's root path, so
\corref{cor:beta-tree-search-depth} supplies a dual of cost $B$.
Obtain the selected record by evaluating the functions $g_w$ along its root path.
By \lemref{lem:beta-adaptive-sum}, this draw, with its permutation fixed,
has dual cost $O(B\log(q+2))$.  Under a random permutation the first
marked label is exactly uniform on the marked set; no rejection cutoff is
needed here.

Apply the sampling process from \lemref{lem:beta-sampling} for $R$
rounds, starting with $S_0=\varnothing$.  In round $j$, make $2\beta$
independent uniform marked draws with replacement, keeping the marked
set fixed until the round is complete.  Let $A_j$ be their union and
put $S_j=S_{j-1}\cup A_j$.  A previously marked record survives
precisely when $p(S_j\cup C_v)\ne p(S_j)$.  Once unmarked, a record
stays unmarked.  Stop if the marked set is empty.  Retain the full
records of successive saved unions for all subsequent insertion tests,
and return a compression of the final saved union.

The leave-one-out argument in~\eqref{eq:beta-halving} gives
$\mathbb E|V_R|\leq q2^{-R}$.  Markov's inequality bounds the
probability of a nonempty final marked set by $q2^{-R}$.
The product of every discarded record is dominated by the final saved product.
This proves the success guarantee, including for overlapping or
repeated records.

Fix all $2\beta R$ permutations in advance.  The entire procedure is now
a deterministic adaptive computation.  Its transcript has the sum of
its draw costs; projecting once to the compressed output costs a
factor of at most two.  This proves~\eqref{eq:beta-tree-combine}.
Every returned record describes a subword of the input, for every seed.
The probability of failure concerns domination alone, not any inexact
evaluation of the fixed-seed functions.
\end{proof}

\subsection{The cost of doubling the rank}\label{sec:beta-doubling-cost}

\begin{lemma}[Rank doubling]\label{lem:beta-rank-doubling}
Let $M$ be a monoid with known multiplication and a known stable partial
order whose minimum is the identity.  Let $G\subseteq M$ be finite and
put $\beta=\beta_G(M)$, with $1\leq\beta<\infty$.
Fix integers $L,r\geq1$, and put
$N=2^L$ and $\delta=100^{-(L+2)}$.
Suppose that, for some $A\geq1$, every dyadic interval $I\subseteq[N]$
of length $m$ has a finite family of deterministic record functions on $G^I$,
indexed by a random seed whose distribution is independent of the input,
with the following properties:
\begin{enumerate}
 \item For every input and every seed, the output $K$ consists of at most
       $\beta$ positions in $I$ and their input letters.
 \item For every input, with probability at least $1-\delta$ over the seed,
       $p(U)\leq p(K)$ simultaneously for every $U\subseteq I$ with
       $|U|\leq r$, where $p(U)$ is the product of the input letters at
       positions in $U$, in input order.
 \item For every fixed seed, the output function has an all-pairs dual
       of cost at most $A\sqrt m$.
\end{enumerate}
Then every dyadic interval $I\subseteq[N]$ of length $m$ admits a finite family of
deterministic record functions with an input-independent seed distribution
whose outputs always contain at most $\beta$ positions in $I$ and their
input letters, and which, for every input, return a record dominating
every subword of length at most $2r$ with probability at least $1-\delta$.
For every fixed seed, the output function has an all-pairs dual of cost
at most
\begin{equation}\label{eq:beta-rank-doubling}
 C\beta AL^4\sqrt m,
\end{equation}
where $C$ is an absolute constant independent of $r$, $\beta$, and $M$.
\end{lemma}

\begin{proof}
Set
\begin{equation}\label{eq:beta-doubling-parameters}
 c=L+2,\qquad R=L+11,\qquad \delta=100^{-c}.
\end{equation}
These parameters are shared by all recursive intervals and stages.
A one-position interval returns its letter.  On an interval of length
$m=2^j\geq2$, independently choose one seed for the $r$-summary of
each nonroot interval in its binary tree.  At an internal vertex $v$,
let $g_v$ return the two child summaries.  Include leaves as dummy
vertices: at each leaf $v$, let $g_v$ return an empty output and set
$C_v=\varnothing$.  For each internal vertex $v$, the outputs along
its root path determine the full frontier record $C_v$
from \lemref{lem:beta-median-frontier}.

For each fixed choice of child seeds, the function $g_v$ at depth $h$
has a dual of cost at most
\begin{equation}\label{eq:beta-frontier-lookup-cost}
 t_v=C_0A\sqrt{m/2^h},
\end{equation}
where $C_0$ absorbs the cost of combining the two child records.
The same positive bound is valid for the constant functions at the leaves.
There are $2^h$ vertices at depth $h$, so
\[
 \sum_vt_v^2\leq C_0^2A^2m(j+1).
\]
The weighted tree dual therefore has cost
\begin{equation}\label{eq:beta-frontier-search-cost}
 B\leq C_0A\sqrt m\,(j+1)=O(A\sqrt m\,L).
\end{equation}
Apply \lemref{lem:beta-tree-sampling} with the $q=2m-1$ vertices and
$R$ rounds.  This gives one raw doubling output of dual cost
$O(\beta A\sqrt m\,L^3)$.

We now analyze its classical success probability.  There are fewer
than $2N$ child summaries, each failing with probability at most
$\delta$.  Thus all are correct except with probability at most
$2N\delta\leq1/400$.  For every fixed table of child records, the
sampling failure probability is at most
$q2^{-R}\leq2N2^{-R}\leq1/400$.
When both guarantees hold, \lemref{lem:beta-median-frontier} makes the
raw output a $2r$-summary.  Its failure probability is therefore at
most $1/200<1/100$.

Finally, take $c$ independent copies of this whole construction and
compress the union of their outputs.  If any copy is a $2r$-summary,
so is this union, by insertion monotonicity.  Every copy returns a subword
of the input, including unsuccessful copies, so their union is also a subword.
Independence makes the failure probability at most
$(1/100)^c=\delta$.  Combining the fixed-seed output duals costs
$O(c)=O(L)$ times the cost of one copy.  This is the fourth logarithm
in~\eqref{eq:beta-rank-doubling}.
\end{proof}

We can now prove the logarithmic-exponent bound of
\hyperref[res:ordered]{Result~\ref*{res:ordered}}.

\begin{theorem}[Ordered products: a logarithmic exponent]
\label{thm:ordered-beta-log-product}
Let $M$ be a monoid with a known stable partial order in which the identity
is the minimum element.  Let $G\subseteq M$ be finite, and put
$\beta=\beta_G(M)$, where $1\leq\beta<\infty$.
For an absolute constant $C$ and every $n\geq1$,
\begin{equation}\label{eq:beta-log-product}
 Q_{1/3}(\operatorname{Prod}_{M,G,n})
 \leq\min\left\{n,\sqrt{n+1}
       \bigl(C(\beta+2)\log(n+2)\bigr)^{C\log(\beta+2)}\right\}.
\end{equation}
Moreover, there is an algorithm achieving this bound which also returns
a product-preserving subword of at most $\beta$ positions, with probability
at least $9/10$.  Its recorded letters agree with the input on every branch.
The constant is independent of $M$, $G$, and $\beta$.
\end{theorem}

\begin{proof}
Assume $1\leq \beta<n$; the remaining cases are handled by reading and
compressing the input.  Put
\[
 L_n=\lceil\log_2(n+1)\rceil,\qquad N=2^{L_n},
\]
so $n<N\leq2n$.  Pad with known identities, which incur no queries,
and use $L=L_n$ and $\delta=100^{-(L+2)}$ throughout the recursion.
The singleton-summary family from~\eqref{eq:beta-base-cost} fails with
probability at most $1/100$.  Compressing the union of $c=L+2$
independent copies reduces this to $\delta$ and gives dual cost
$O(\beta\sqrt m\,L^{5/2})$ on every dyadic interval of length $m\leq N$.

Choose a sufficiently large absolute $C_1$ and set $K=C_1\beta L^4$.
The base cost is at most $K\sqrt m$, and
\lemref{lem:beta-rank-doubling} gives, by induction, a rank-$2^s$
summary family of dual cost at most $K^{s+1}\sqrt m$.
All stages retain failure probability at most $\delta$.
Take $k=\lceil\log_2 \beta\rceil$.  A successful rank-$2^k$ summary
$U$ dominates a preserving subword of the input.  Since $U$ is itself
a subword, antisymmetry gives $p(U)=[x_1\cdots x_n]$.

For each fixed root seed, convert the exact record dual to a quantum
algorithm with error at most $1/16$.  Choosing the seed according to
its finite classical distribution then returns a preserving
record of at most $\beta$ positions with probability at least
\[
 (1-\delta)\frac{15}{16}\geq\frac{99}{100}\frac{15}{16}>\frac9{10}.
\]
This is the only conversion to a quantum algorithm in the recursion.
Remove any padding positions from its output.  If it reports more than
$\beta$ positions, return the empty record; otherwise requery those positions
and replace the recorded letters by their input values.  This ensures that
all recorded letters agree with the input, uses at most $\beta$ extra queries,
and preserves every successful output.

The constant conversion cost, padding factor, and final verification
are absorbed by one further factor of $K$, after increasing $C_1$.
In particular, for an absolute constant $C_2$ the construction gives
\begin{equation}\label{eq:beta-doubling-iterated}
 Q_{1/3}(\operatorname{Prod}_{M,G,n})
 \leq\min\left\{n,\sqrt n\,
   \bigl(C_2(\beta+2)L_n^4\bigr)^{\lceil\log_2(\beta+2)\rceil+2}\right\}.
\end{equation}
The read-all alternative also proves this bound when $\beta\geq n$.
Finally, $L_n=O(\log(n+2))$ and
$\lceil\log_2(\beta+2)\rceil+2=O(\log(\beta+2))$.
Absorbing the fourth power into the exponent gives
\eqref{eq:beta-log-product}, with an absolute constant independent of
$M$, $G$, and $\beta$.
\end{proof}

\begin{remark}[Formalization]
\label{rem:ordered-formalization}
With $L_n=\lceil\log_2(n+1)\rceil$, the Lean formalization proves the
linear-exponent bound
\[
 Q_{1/3}(\operatorname{Prod}_{M,G,n})
 \leq\min\{n,(2^{27}\beta)^\beta\sqrt n\,L_n^{5\beta/2}\}.
\]
For the logarithmic-exponent bound it proves
\eqref{eq:beta-doubling-iterated} with $C_2=2^{15}$ and
\eqref{eq:beta-log-product} with $C=2^{17}$.
These constants use natural logarithms in
\eqref{eq:beta-log-product} and the formal library's value oracle,
which swaps a blank answer with the queried letter and fixes other
answers.  All bounds also hold at error $1/10$ in that model.

Lemmas~\ref{lem:beta-sampling} and~\ref{lem:beta-rank-doubling} are not
stated as standalone lemmas in Lean in the form given here.  However,
\thmref{thm:ordered-beta-log-product} is fully formalized, and its proof
verifies the cost bounds needed from these lemmas for its recursive
construction, with explicit constants.

The formalized verification uses at most $\beta$ additional queries to
ensure that every output lists at most $\beta$ original positions and
their actual input letters, with product preservation holding with
probability at least $9/10$.
\end{remark}

\section{Applications: stock and unitriangular products}\label{sec:applications}

In this section, we apply the product-breadth bounds for stably ordered
monoids whose identity is the minimum element.  We use the linear-exponent
bound of \thmref{thm:ordered-beta-product} for stock summaries and the
logarithmic-exponent bound of \thmref{thm:ordered-beta-log-product} for
unitriangular tropical matrices.
In each case, we bound the product breadth to obtain quantum query bounds
independent of the size of the input alphabet.

\subsection{Best time to buy and sell stock}\label{sec:stock-beta}

Allcock et al.~\cite{ABBLS23} studied the best time to buy and sell stock
problem and obtained an $O(\sqrt{n\log(n+2)})$ quantum query bound.
The problem can be formulated as a monoid product problem in a stably
ordered monoid with identity as the minimum element and product breadth
at most four, so it fits the general framework of
\hyperref[res:ordered]{Result~\ref*{res:ordered}}.
For this fixed breadth, the linear-exponent bound of
\thmref{thm:ordered-beta-product} gives a quantum query upper bound
of $O(\sqrt n\log^{10}(n+2))$, recovering the square-root dependence
on $n$ with a worse logarithmic factor.

Fix a finite set of integer prices $P\subseteq\mathbb Z$.  Recall the stock summaries from
Example~\ref{ex:btbs-breadth}: a nonempty price word $w=x_1\cdots x_m$
over $P$ has summary
\[
 \sigma(w)=\bigl(\min_i x_i,\max_i x_i,\max_{i<j}(x_j-x_i)\bigr),
\]
recording its minimum price, maximum price, and best profit, with
$\max\varnothing=-\infty$.  For adjacent words $u,v$, we have
$\sigma(uv)=\sigma(u)\otimes\sigma(v)$, where
\[
 (a,b,c)\otimes(d,e,f)
 =\bigl(\min(a,d),\max(b,e),\max(c,f,e-a)\bigr).
\]
The third coordinate accounts for transactions within either word and
for buying in the first and selling in the second.  This operation is
associative because it represents concatenation.

Let $M_P$ be the monoid of these summaries, with an identity adjoined for
the empty word, and let $G_P=\{(p,p,-\infty):p\in P\}$ be the singleton
price summaries.  Multiplying the singleton summaries of an input word
gives its best profit in the third coordinate.  The monoid is finite:
the minimum and maximum lie in $P$,
and a profit is either $-\infty$ or a difference of two prices in $P$.

\begin{proposition}[A four-position stock core]\label{prop:stock-beta}
The monoid $M_P$ has a stable partial order in which the identity is the
minimum element, and
\[
 \beta_{G_P}(M_P)\leq4.
\]
Consequently, the complete stock summary can be computed with
$O(\sqrt n\log^{10}(n+2))$ queries, with an absolute implicit constant
independent of $P$.
\end{proposition}

\begin{proof}
Order nonempty summaries by
\[
 (a,b,c)\leq(d,e,f)
 \quad\Longleftrightarrow\quad a\geq d,\quad b\leq e,\quad c\leq f,
\]
and place the empty summary below every nonempty summary.  The multiplication
formula above is monotone in both factors for this order.
For example, on multiplying on the right, decreasing the minimum $a$
can only increase the cross profit $e-a$.  The other coordinates and
left multiplication are checked in the same way.  A product dominates
both factors, so the comparisons involving the empty summary also hold.

Given a price word of length at least two, retain a position attaining its
minimum, one attaining its maximum, and a buy--sell pair attaining its
best profit.  There are at most four retained positions.  Their subword
has the same minimum and maximum and still attains the optimal profit.
Every feasible transaction in the subword is also feasible in the
original word, so its profit cannot be larger.  This preserves the entire
summary whether the best profit is positive, zero, or negative.  A
singleton word is its own core, and the empty word needs no positions.
Apply \thmref{thm:ordered-beta-product} using $\beta_{G_P}(M_P)\leq4$.
\end{proof}

Four positions are sometimes necessary.  The word $(10,2,5,0)$ has
summary $(0,10,3)$.  The first and last positions supply the unique
maximum and minimum; the middle pair supplies the unique positive profit.
Deleting any position changes the summary.  Thus the breadth bound is
attained whenever the alphabet contains these prices.

\subsection{Unitriangular tropical matrices}\label{sec:tropical}

Let $\mathbb T=(\mathbb R\cup\{-\infty\},\max,+)$ be the max-plus tropical
semiring, and let $U_k(\mathbb T)$ be the monoid of $k$-by-$k$ matrices
with diagonal entries zero and entries below the diagonal $-\infty$.
Matrix multiplication is
\[
 (AB)_{st}=\max_v(A_{sv}+B_{vt}).
\]
The identity has $-\infty$ off the diagonal.  Entrywise order is stable,
and this identity is the minimum element.

For $k=2$, the single strict entry computes maximum.  The stock monoid
from the preceding subsection (with integer prices) embeds into
$U_3(\mathbb T)$ by
\[
 (a,b,c)\longmapsto
 \begin{pmatrix}0&-a&c\\-\infty&0&b\\-\infty&-\infty&0\end{pmatrix},
\]
with the empty summary mapped to the identity.  Matrix multiplication
gives the three coordinates of the stock-summary product in
Example~\ref{ex:btbs-breadth}.

More generally, a fixed sign pattern $\varepsilon_1,\ldots,\varepsilon_t$
is represented by a chain of $t+1$ states, with edge $j\to j+1$ at
position $i$ carrying weight $\varepsilon_jx_i$ and all other strict
entries equal to $-\infty$.  The first-to-last entry
is $\max_{i_1<\cdots<i_t}\sum_j\varepsilon_jx_{i_j}$.
The pattern $(-1,+1,\ldots,-1,+1)$ gives the maximum total profit
from a fixed number of transactions.

To apply \hyperref[res:ordered]{Result~\ref*{res:ordered}} to these
problems, we bound the product breadth in terms of $k$ alone using
the following path description.

\begin{lemma}[Short paths preserve the matrix product]\label{lem:path}
For $A_1,\ldots,A_n\in U_k(\mathbb T)$ and $s<t$, the entry
$(A_1\cdots A_n)_{st}$ is the maximum weight of a path from $s$ to $t$
whose strict state transitions occur at increasing input positions.
Every such path has at most $t-s$ strict transitions.
\end{lemma}

\begin{proof}
The matrix product expands over paths with nondecreasing state indices.
A repeated state contributes a diagonal entry of weight zero.  Removing
these stays leaves a sequence of strict transitions through distinct
states between $s$ and $t$, so there are at most $t-s$ of them.
\end{proof}

\begin{corollary}[Product breadth of unitriangular matrices]
\label{cor:unitriangular-beta}
Let $G$ be a finite alphabet in $U_k(\mathbb T)$.  Its generated monoid
$H=\langle G\rangle$ is finite and
\[
 \beta_G(H)\leq V_k:=\sum_{s<t}(t-s)=\binom{k+1}{3}.
\]
For $k\geq2$, an absolute constant $C$ therefore gives
\[
 Q_{1/3}(\operatorname{Prod}_{H,G,n})
 \leq\min\left\{n,\sqrt{n+1}
     \bigl(C(V_k+2)\log(n+2)\bigr)^{C\log(V_k+2)}\right\}.
\]
The exponent is $O(\log k)$, and the bound is
$\widetilde O_k(\sqrt n)$ for every fixed $k$, uniformly over finite
matrix alphabets.  The case $k=1$ is query-free.
\end{corollary}

\begin{proof}
For each finite entry above the diagonal, choose a maximizing path and
retain the input positions of its strict transitions.  Their union has
size at most $V_k$.  Every chosen path survives, and every path through
the retained subword extends to the full word using zero-weight stays.
Thus the subword preserves every entry, including those equal to $-\infty$.

For finite $G$, each entry of a product belongs to the finite set of sums
of at most $k-1$ entries of matrices in $G$, together with $-\infty$.
Hence $H$ is finite.  Its inherited order is stable, with the identity as
the minimum element.  If $\beta_G(H)=0$, the product is the identity and
no queries are needed; otherwise \thmref{thm:ordered-beta-log-product}
applies with $\beta_G(H)\leq V_k$.
\end{proof}

The uniform dependence on $k$ allows the matrix dimension to grow with
the input length.  Since $V_k=O(k^3)$, the query bound above is
\[
 \sqrt n\,\exp\!\left(O\!\left((\log k)^2
       +\log k\,\log\log(n+2)\right)\right).
\]
This is $n^{1/2+o(1)}$ whenever $\log k=o(\sqrt{\log n})$, in particular
when $k$ is any fixed power of $\log n$.  Thus the same
$n^{1/2+o(1)}$ query bound applies when the number of stock transactions
or prescribed signs grows polylogarithmically with $n$.

\section{Sharp bounds for Boolean unitriangular products}
\label{sec:boolean-unitriangular}

Let $\mathbb B=(\{0,1\},\vee,\wedge)$, and let
$\operatorname{UT}_k(\mathbb B)$ be the monoid of upper triangular Boolean
matrices with diagonal entries one.  Each strict entry of a prefix product
can change only once.  This gives a sharper dependence on $k$ than the
general breadth theorem.  We give both bounds here, including the
independent block-search lower bound.

\begin{theorem}[Boolean unitriangular products]\label{thm:boolean-unitriangular}
	The monoid $\operatorname{UT}_k(\mathbb B)$ has product breadth $\binom{k}{2}$.
	For every $k,n\ge1$,
	\begin{align*}
		\ADVpm\bigl(\operatorname{Prod}_{\operatorname{UT}_k(\mathbb B),n}\bigr)
		&\le 8\sqrt{n\min\left\{n,\binom{k}{2}\right\}},\\
		\ADVpm\bigl(\operatorname{Prod}_{\operatorname{UT}_k(\mathbb B),n}\bigr)
		&\ge \sqrt{\frac n2
		\min\left\{n,\left\lfloor\frac{k^2}{4}\right\rfloor\right\}}.
	\end{align*}
	Consequently
	\[
		Q_{1/3}\bigl(\operatorname{Prod}_{\operatorname{UT}_k(\mathbb B),n}\bigr)
		=\begin{cases}
			0,&k=1,\\[1mm]
			\Theta\!\left(\min\{n,k\sqrt n\}\right),&k\ge2,
		\end{cases}
	\]
	with universal implicit constants.  The lower bound already holds for a Boolean
	postprocessing of the full product on a promise.
\end{theorem}

\begin{proof}
	Apply \lemref{lem:bounded-change-scan} to the prefix products
	$p_i=x_1\cdots x_i$ in $\operatorname{UT}_k(\mathbb B)$.  Since every input matrix has
	ones on the diagonal, $p_i\le p_{i+1}$ entrywise.  Each of the $\binom{k}{2}$ strict-upper entries
	can therefore flip from $0$ to $1$ at most once, while the diagonal and lower entries
	are fixed.  Hence the prefix state changes at most
	$C=\min\{n,\binom{k}{2}\}$ times.  The lemma and the algorithm that
	queries every input give the stated upper bound.

	Retaining only the positions where the prefix product changes also gives
	$\beta(\operatorname{UT}_k(\mathbb B))\leq\binom{k}{2}$.
	For the reverse inequality, take one matrix $F_{ij}$ for each $i<j$,
	with ones on the diagonal and at $(i,j)$ and zeros elsewhere.
	Order these matrices by nonincreasing row index $i$.
	No two strict edges can compose in this order, since composing
	$(i,j)$ with $(j,\ell)$ would require an increase in the row index.
	The product of any subword therefore has exactly the strict entries
	contributed by its letters.  All $\binom{k}{2}$ letters are necessary
	to preserve the full product, proving the breadth identity.

	For the lower bound, the case $k=1$ is immediate.  For $k\geq2$, split the ordered states into consecutive sets
	\[
		A=\{1,\ldots,a\},\qquad B=\{a+1,\ldots,k\},
		\qquad a=\lfloor k/2\rfloor,
	\]
	and put $r=|A||B|=\lfloor k^2/4\rfloor$.  For $R\subseteq A\times B$, let $X_R$
	be the Boolean matrix with ones on the diagonal and in the positions of $R$, and zeros
	elsewhere.  No two edges in $A\times B$ compose, so
	\[
		X_RX_S=X_{R\cup S}.
	\]
	Thus these matrices form a free union semilattice of rank $r$ inside
	$\operatorname{UT}_k(\mathbb B)$.

	Choose $d=\min\{r,n\}$ cross edges.  Partition the $n$ query positions into $d$
	balanced nonempty blocks, of sizes $m_1,\ldots,m_d$.  In block $j$, restrict each input
	to the identity or the singleton-edge matrix belonging to edge $j$, with the promise
	that the block contains either zero or one singleton-edge letter.  The corresponding
	entry of the full product is the promised \textsc{Or} of that block.
	Take the parity of these $d$ entries as a Boolean postprocessing.
	For block $j$, let $\Gamma_j$ be the adjacency matrix of the star with
	center the all-identity block and one leaf for each possible location of
	its unique singleton-edge letter.  Its norm is $\sqrt{m_j}$.
	On the Cartesian product of the block promises, set
	\[
	 \Gamma=\sum_{j=1}^d
	 I\otimes\cdots\otimes\Gamma_j\otimes\cdots\otimes I.
	\]
	Every edge changes the parity, so $\Gamma$ is a valid adversary.
	The product of the stars' top eigenvectors shows that
	$\|\Gamma\|=\sum_j\sqrt{m_j}$; the reverse inequality follows from
	the triangle inequality.  Filtering by one input position leaves only
	the corresponding star edge, tensored with identities, and has norm one.
	Since the balanced block sizes satisfy $m_j\geq n/(2d)$,
	\[
	 \ADVpm\geq\sum_{j=1}^d\sqrt{m_j}
	 \geq\sqrt{nd/2}
	 =\Omega\!\left(\min\{n,k\sqrt n\}\right).
	\]
	Extending this promise adversary by zero outside the promise gives an adversary for the
	full product map.
	This Boolean postprocessing lower-bounds computation of the full matrix and completes
	the proof.
\end{proof}

Since $|\operatorname{UT}_k(\mathbb B)|=2^{\binom{k}{2}}$, the theorem may equivalently
be written
\[
	Q_{1/3}\bigl(\operatorname{Prod}_{\operatorname{UT}_k(\mathbb B),n}\bigr)
	=\Theta\!\left(\min\left\{n,
	\sqrt{n\log\bigl|\operatorname{UT}_k(\mathbb B)\bigr|}\right\}\right)
\]
for $k\ge2$.

For a finite $\JJ$-trivial monoid $M$, define its unitriangular division degree by
\[
  \tau(M)=\min\{k\ge1:M\prec\operatorname{UT}_k(\mathbb B)\}.
\]
This minimum exists by Simon's division theorem~\cite{Sim75,ST88}.

\begin{corollary}[Finite $\JJ$-trivial products]
\label{cor:jtrivial-division}
  Every finite $\JJ$-trivial monoid $M$ satisfies
  \[
    Q_{1/3}(\operatorname{Prod}_{M,n})
      =O\!\left(\sqrt{n\min\left\{n,\binom{\tau(M)}2\right\}}\right)
      =O\!\left(\min\{n,\tau(M)\sqrt n\}\right).
  \]
\end{corollary}

\begin{proof}
  Simon's theorem gives $M\prec\operatorname{UT}_{\tau(M)}(\mathbb B)$.
  Proposition~\ref{prop:division-monotonicity} and
  \thmref{thm:boolean-unitriangular} give the claimed bound.  (When $\tau(M)=1$,
  the divisor is trivial.)
\end{proof}

\section{General aperiodic monoids: the AGS bound and Dyck obstruction}\label{sec:ags}

In this section we present the star-free upper bound of Aaronson, Grier, and
Schaeffer~\cite{AGS19} in the language of the product problem.  We follow
their argument and make its dependence on $|M|$ explicit.  This is not
claimed as a new upper bound: our purpose is to establish a baseline for
comparison with the improved size dependence proved in the next section.

Throughout this section $M$ is a finite aperiodic monoid and the input is
$x = x_1 \cdots x_n \in M^n$.  For $m \in M$ let
\[
	P_m = \{x \in M^* : [x] = m\}
\]
denote the set of inputs whose product is $m$.  We construct an adversary dual
for each membership test and combine these $|M|$ duals to compute the product.

For $m \in M$ let $d(m)$ denote the length of the longest
strictly increasing chain of principal two-sided ideals starting at $MmM$:
\[
	d(m) = \max \{ t : MmM \subsetneq Mr_1M \subsetneq \cdots \subsetneq Mr_tM
	\text{ for some } r_1, \ldots, r_t \in M \} \enspace.
\]
Recall from \secref{sec:prelim} that $\dJ(M)$ is the maximum number of strict
containments in a chain of principal two-sided ideals of $M$.  Thus
$d(m) \le \dJ(M)$ for all $m$.  We also have $d(m) = 0$ iff $m = 1$ (if $m \neq 1$ then
$MmM \subsetneq M = M1M$ by \propref{prop:identity}), and $MmM \subseteq MrM$ implies
$d(r) \le d(m)$, strictly if the containment is strict.

We induct on $d(m)$.  A \emph{strict parent} of a target $m$ means any
$s$ with $MmM\subsetneq MsM$, not necessarily an immediate cover in the
ideal order.  The decomposition and searches below use only such targets.
The base case is \propref{prop:identity}: $x \in P_1$ iff every $x_i = 1$,
an unstructured search problem with dual cost $O(\sqrt n)$.
For the inductive step we use Sch\"utzenberger's decomposition of $P_m$
into pieces whose equality targets are strict parents of $m$.  In \cite{AGS19} this decomposition (their Theorem~22) is
quoted from Sch\"utzenberger's proof; since our whole purpose is bookkeeping, we give a
self-contained proof from the Sandwich Lemma.

\subsection{The decomposition theorem}
\label{sec:ags-decomposition}

The first two conditions below detect entry into the target's right and
left ideals.  The other two exclude a letter or an interval whose product
can no longer be extended to the target, even with factors on both sides.

\begin{theorem}[Decomposition~\cite{Sch65,AGS19}]\label{thm:decomp}
	Let $M$ be a finite aperiodic monoid and $m \in M$ with $\rho(m) > 0$.  Define
	\begin{align*}
		E &= \{(r, a) \in M \times M : raM = mM,\ rM \neq mM\} \\
		F &= \{(a, r) \in M \times M : Mar = Mm,\ Mr \neq Mm\} \\
		C &= \{a \in M : m \notin MaM\} \\
		G &= \{(a, r, b) \in M \times M \times M :
			m \in (MarM \cap MrbM) \setminus MarbM\} \enspace.
	\end{align*}
	Then $x \in P_m$ if and only if all four of the following conditions hold:
	\begin{enumerate}
		\item[(U)] some prefix $x_1 \cdots x_i$ with $[x_1 \cdots x_i] = r$ is followed
			by the letter $x_{i+1} = a$, for some $(r, a) \in E$;
		\item[(V)] some suffix $x_{j} \cdots x_n$ with $[x_j \cdots x_n] = r$ is preceded
			by the letter $x_{j-1} = a$, for some $(a, r) \in F$;
		\item[(C)] no letter $x_i$ belongs to $C$;
		\item[(W)] no infix $x_i x_{i+1} \cdots x_j$ with $i < j$ has $x_i = a$,
			$[x_{i+1} \cdots x_{j-1}] = r$, and $x_j = b$ for some $(a, r, b) \in G$
			(when $j = i + 1$ the infix product is $r = 1$).
	\end{enumerate}
	Furthermore, for every $r \in M$ appearing in $E$, $F$, or $G$ we have
	$MmM \subsetneq MrM$, and in particular $\rho(r) < \rho(m)$.
\end{theorem}

\begin{proof}
	First suppose $x \in P_m$, so $p_n = m$.

	(U): The right ideals of prefixes descend:
	$M = p_0 M \supseteq p_1 M \supseteq \cdots \supseteq p_n M = mM$.  Since $\rho(m) > 0$ we
	have $m \neq 1$, and $p_0 M = M \neq mM$: indeed if $mM = M$ then $1 = mq$ for some $q$ and
	$m = 1$ by \propref{prop:identity}.  Let $i + 1$ be minimal with $p_{i+1} M = mM$.  Then
	$r = p_i$ and $a = x_{i+1}$ satisfy $raM = p_{i+1}M = mM$ and $rM = p_i M \neq mM$ by
	minimality, so $(r, a) \in E$ and condition (U) holds.

	(V): symmetric, using suffix products and left ideals.

	(C): every letter $x_i$ satisfies $m = p_n \in M x_i M$, so $x_i \notin C$.

	(W): if some infix had $x_i = a$, $[x_{i+1} \cdots x_{j-1}] = r$, $x_j = b$ with
	$(a, r, b) \in G$, then $m = p_n \in M a r b M$, contradicting $m \notin MarbM$.

	Conversely, suppose (U), (V), (C), (W) all hold; we show $p_n = m$.  We first claim that
	\begin{equation}\label{eq:window}
		m \in M [x_i \cdots x_j] M \qquad \text{for every infix } x_i \cdots x_j .
	\end{equation}
	The proof is by induction on the length $j - i + 1$ of the infix.  For a single letter this
	is exactly condition (C).  For an infix $x_i \cdots x_j$ of length at least $2$, let
	$a = x_i$, $r = [x_{i+1} \cdots x_{j-1}]$, and $b = x_j$.  By induction on the two
	infixes $x_i \cdots x_{j-1}$ and $x_{i+1} \cdots x_j$ we know $m \in MarM$ and
	$m \in MrbM$.  If $m \notin MarbM$ then $(a, r, b) \in G$, contradicting (W).  So
	$m \in MarbM = M[x_i \cdots x_j]M$, proving \eqref{eq:window}.

	Now by (U) there is $(r, a) \in E$ and a prefix with $p_i = r$, $x_{i+1} = a$, so
	$p_{i+1} M = raM = mM$ and hence $p_n \in p_n M \subseteq p_{i+1} M = mM$.  Symmetrically,
	(V) gives $p_n \in Mm$.  By \eqref{eq:window} applied to the whole input,
	$m \in M p_n M$, i.e., $p_n \notin J_m$.  By \lemref{lem:singleton},
	$p_n \in (mM \cap Mm) \setminus J_m = \{m\}$.

	It remains to prove the rank claim.  Let $(r, a) \in E$.  Since $raM = mM$ we have
	$m = raq$ for some $q$, so $m \in rM$ and $MmM \subseteq MrM$, i.e.,
	$\rho(r) \le \rho(m)$.  Suppose for contradiction that $MrM = MmM$, so $r = umv$ for some
	$u, v \in M$.  Substituting $m = raq$ gives $m = umv \cdot aq = u \cdot m \cdot vaq$, so by
	the Sandwich Lemma $m = um$ and $m = m \cdot vaq$.  But then $r = umv = mv \in mM$, so
	$rM \subseteq mM \subseteq rM$, contradicting $rM \neq mM$.  Hence $\rho(r) < \rho(m)$.
	The argument for $F$ is symmetric.

	Finally let $(a, r, b) \in G$.  From $m \in MarM \subseteq MrM$ we get
	$\rho(r) \le \rho(m)$.  Suppose $MrM = MmM$, so $r = umv$ for some $u, v$.  From
	$m \in MarM$ write $m = x a r y$.  Then
	$r = umv = u (xary) v = (ux a) \cdot r \cdot (yv)$, so by the Sandwich Lemma
	$r = r(yv)$.  Similarly from $m \in MrbM$ write $m = x' r b y'$; then
	$r = umv = (ux') \cdot r \cdot (by'v)$, and the Sandwich Lemma gives $r = r(by'v)$.
	Now
	\[
		m = xary = x a \big(r (by'v)\big) y = xarb \cdot y'vy \in MarbM \enspace,
	\]
	contradicting $m \notin MarbM$.  Hence $\rho(r) < \rho(m)$.
\end{proof}

Two comments.  First, all four conditions refer only to products of infixes of the input, so the
theorem really is a statement about the product problem; the alphabet plays no role beyond the
monoid itself.  Second, the index sets have sizes
\[
	|E|, |F| \le |M|^2, \qquad |C| \le |M|, \qquad |G| \le |M|^3 \enspace,
\]
and these sizes are exactly the loop bounds of the algorithm below.  This is one of the two places
the monoid size enters; the other is the depth of the recursion, bounded by $\dJ(M)$.

\subsection{Prefix, suffix, and infix search}
\label{sec:ags-searches}

Throughout, ``an algorithm for $P_r$'' means a bounded-error quantum algorithm deciding whether
the product of its input equals $r$, applicable to any infix of $x$ (an infix of the input is
again oracle input of the same kind).
In this subsection, put
\[
 H_n=\lceil\log_2(n+1)\rceil+2.
\]

\paragraph{Prefix search.}
For a fixed pair $(r,a)\in E$, condition (U) asks whether a prefix has
product $r$ and is followed by the letter $a$.  The next lemma detects
this transition.  Testing all pairs in $E$ then decides (U).

\begin{lemma}[Prefix search, cf.~Lemma~24 of \cite{AGS19}]\label{lem:prefix}
	Let $n\geq0$ and let $r,a\in M$ satisfy $rM \supsetneq raM$.  Suppose we have bounded-error
	algorithms for $P_s$ for every $s$ with $r\in sM$, each of cost at most
	$T(n)$ on every input length at most $n$.  Then there is a bounded-error algorithm which decides, with
	\[
		O\big( |M| \cdot T(n) \cdot \lambda H_n \big)
	\]
	queries, whether some $0\leq i<n$ satisfies $[x_1 \cdots x_i] = r$ and
	$x_{i+1} = a$.
\end{lemma}

\begin{proof}
	For $n=0$, the event is false and no queries are needed.  For $n\geq1$,
	the supplied $P_1$ test on length-one inputs is nonconstant, so $T(n)\geq1$.
	For $n=1$,
	query $x_1$ and accept exactly when $r=1$ and $x_1=a$.
	Below assume $n\geq2$.

	Let $H = \{s \in M : r \in sM\}$.  These are exactly the supplied
    equality targets; each satisfies $MrM\subseteq MsM$.  Membership of a prefix in
	$K = \{ \text{prefixes } x_1 \cdots x_i : r \in p_i M \}$
	can be tested by running the $P_s$ algorithms for all $s \in H$ ($|H| \le |M|$ of them,
	each amplified to error $O(1/(|M|H_n))$, a factor $\lambda$) at cost
	$O(|M| \cdot T(n) \cdot \lambda)$.

	$K$ is prefix closed: if $r \in p_i M$ then $r \in p_{i-1} M$ since
	$p_i M \subseteq p_{i-1} M$.  (The empty prefix is in $K$ as $p_0 = 1$.)  So we can binary
	search for the largest $i^*$ with $x_1 \cdots x_{i^*} \in K$ using at most $H_n-2$ membership
	tests.

	We claim the event holds iff $i^* < n$, $[x_1 \cdots x_{i^*}] = r$, and
	$x_{i^*+1} = a$, which is one more amplified $P_r$ test plus one query (if $i^* = n$ there is no
	next letter and the event is declared false).  If the event holds at position $i$, then
	$x_1 \cdots x_i \in K$; moreover no longer prefix is in $K$: for $j > i$,
	$p_j M \subseteq p_{i+1} M = raM$, and $r \in raM$ would give $rM \subseteq raM$,
	contradicting $rM \supsetneq raM$.  Hence $i = i^*$.
	Including the final equality and letter tests, there are at most $H_n$ steps.
\end{proof}

The symmetric statement (suffix search) detects event (V) with the same
cost, using equality targets $s$ with $r\in Ms$.
Note where the hypothesis $rM \supsetneq raM$ comes from: for $(r, a) \in E$ we have
$raM = mM \subseteq rM$ ($m \in rM$ was shown in the proof of \thmref{thm:decomp}) and
$rM \neq mM$, so the hypothesis holds automatically.

\paragraph{Splitting.}
For the infix event (W) we need to detect an infix whose product lands in a given set, without
knowing where the infix is.  The tool is the splitting observation:

\begin{proposition}[Splitting, cf.~Theorem~23 of \cite{AGS19}]\label{prop:split}
	Let $r \in M$ and let $w$ be a word with $[w] = r$.  For every way of cutting $w$
	into $w = w_1 w_2$ (with $w_1, w_2$ possibly empty) there exist $p, q \in M$ with $pq = r$,
	$[w_1] = p$, $[w_2] = q$, and
    $MrM\subseteq MpM\cap MqM$.  In particular,
    $\rho(p),\rho(q)\leq\rho(r)$.  There are at most $|M|^2$ pairs $(p, q)$ with $pq = r$.
\end{proposition}

\paragraph{Infix search.}
\begin{lemma}[Infix search, cf.~Theorem~18 of \cite{AGS19}]\label{lem:infix}
	Let $n\geq0$ and fix $(a, r, b) \in G$ (from \thmref{thm:decomp}, for target $m$).
	Suppose that, for every $s$ with $MrM\subseteq MsM$ and every length
	$0\leq j\leq n$, we have a bounded-error algorithm for $P_s$ of cost
	at most $B\sqrt j$.  Then event (W)
	restricted to the triple $(a, r, b)$---the existence of an infix $a w b$ of $x$ with
	$[w] = r$---can be decided with
	\[
		O\big( |M|^3 \cdot B \sqrt{n} \cdot \lambda^2 H_n^2 \big)
	\]
	queries.
\end{lemma}

\begin{proof}
	For $n<2$, the event is false and no queries are needed.  Below assume $n\geq2$.

	We first verify that prefix/suffix search applies to the pieces.  Split $r = pq$ as in
	\propref{prop:split}.  We must check the ideal-descent hypotheses: $qM \supsetneq qbM$ and
	$Mp \supsetneq Map$.  Suppose toward contradiction that $qbM = qM$.  Then $q \in qbM$, say
	$q = qbw$, and so
	\[
		ar = apq = ap(qbw) = (apqb)w = (arb)w \in MarbM \enspace,
	\]
	giving $MarM \subseteq MarbM$ and hence $m \in MarM \subseteq MarbM$, contradicting
	$(a, r, b) \in G$.  Hence $qM \supsetneq qbM$, and symmetrically $Mp \supsetneq Map$.

	Now the search.  For a length scale $\ell \in \{2, 4, 8, \ldots\}$ with $\ell\leq n$ (a witness infix $awb$
	has length at least $2$, so these scales suffice; all windows below are clipped to
	$[1, n]$), call an index $j$
	\emph{marked} if for some $(p, q)$ with $pq = r$: the window $x_{j - 4\ell} \cdots x_{j-1}$
	contains a suffix of the form $a w_1$ with $[w_1] = p$ (a suffix-search on a window
	of length at most $4\ell$), and the window $x_j \cdots x_{j + 4\ell - 1}$ contains a prefix of the
	form $w_2 b$ with $[w_2] = q$ (a prefix-search on a window of length at most $4\ell$).
	Every window has length at most $\min\{4\ell,n\}$, so all equality tests
	used within it cost at most $B\sqrt{4\ell}$, and its search depth is at most $H_n$.
	By \lemref{lem:prefix}, testing whether $j$ is marked therefore costs
	\[
		O\big( |M|^3 B\sqrt\ell\,\lambda H_n \big)
	\]
	queries (at most $|M|^2$ split pairs, each a prefix plus a suffix search of cost
	$O(|M|B\sqrt\ell\,\lambda H_n)$).

	If $x$ contains an infix $awb$ with $[w] = r$ of total length in $[\ell, 2\ell]$,
	occupying positions $[s, e]$, then every $j \in (s, e]$ is marked: cutting $awb$ at $j$
	splits it into $a w_1 \cdot w_2 b$ with both parts of length at most $2\ell \le 4\ell$,
	and \propref{prop:split} provides the pair $(p, q)$.  So at least
	$e - s \ge \ell - 1 \ge \ell/2$ \emph{consecutive} indices are marked, and Grover search
	over a grid of integer spacing $\max\{1,\lfloor\ell/4\rfloor\}$ finds a marked index with $O(\sqrt{n/\ell})$ iterations.
	Conversely any marked index exhibits an infix $a w_1 w_2 b$ with
	$[w_1 w_2] = pq = r$, so the test is sound.

	Summing over at most $H_n$ scales, with an amplification factor $\lambda$ for the
	marked test inside Grover:
	\[
		\sum_{\ell} O\Big(\sqrt{\tfrac{n}{\ell}}\Big) \cdot
		O\big(|M|^3 B\sqrt\ell\,\lambda^2 H_n\big)
		= O\big( |M|^3 B \sqrt{n}\, \lambda^2 H_n^2 \big) ,
	\]
	since $\sqrt{n/\ell}\sqrt\ell=\sqrt n$ at every scale.
\end{proof}

\subsection{The localized AGS step and the global bound}
\label{sec:ags-local}

We now analyze the preceding searches by composing exact adversary duals.
This removes the amplification factors from the recursive bound and
isolates the local reduction that we will also use for the cube-root bound.

We use the following dual composition estimates.  An OR of $k$ Boolean
functions, each with dual cost at most $T$, has dual cost at most
$24\sqrt{k+1}\,T$; this follows, with room in the constant, from
\thmref{thm:weighted-tree-search} applied to a star.
Complementing a Boolean output leaves its dual unchanged.  Any function
of a finite family of outputs has dual cost at most twice the sum of
their costs, and adaptive calls obey the same estimate by
\lemref{lem:beta-adaptive-sum}.  A predicate of one input letter has
dual cost at most two.

\begin{lemma}[Localized AGS step]\label{lem:ags-local-step}
  Let $M$ be a finite aperiodic monoid, put $q=|M|$, and fix
  $m\neq1$ and a length bound $n_0\geq1$.  Write
  \[
    H=\lceil\log_2(n_0+1)\rceil+2,
    \qquad A_{n_0}=2^{40}(q+1)^5H^2.
  \]
  Suppose $B\geq1$ and, for every $s$ satisfying $MmM\subsetneq MsM$
  and every $1\leq\ell\leq n_0$, the membership function of
  $P_s\cap M^\ell$ has an all-pairs dual of cost at most $B\sqrt\ell$.
  Then membership in $P_m\cap M^n$, for every $1\leq n\leq n_0$,
  has an all-pairs dual of cost at most
  \[
    A_{n_0}B\sqrt n.
  \]
  There is also an absolute constant $c_0$ such that, if the same
  strict-parent tests instead have query cost at most $B\sqrt\ell$
  for some $B\geq0$, then membership in $P_m\cap M^n$ has query cost at most
  \[
    q^{c_0}L(n_0)^{c_0}(B+1)\sqrt n,
    \qquad L(u)=2+\log_2(u+2).
  \]
\end{lemma}

\begin{proof}
  Apply \thmref{thm:decomp} and the prefix, suffix, and infix
  subroutines of \hyperref[lem:prefix]{Lemmas~\ref*{lem:prefix}}
  and~\ref{lem:infix}, including
  the splitting Proposition~\ref{prop:split}.  The proof of
  \thmref{thm:decomp} gives $MmM\subsetneq MrM$ for every middle
  product $r$ in $E$, $F$, or $G$.  Prefix search replaces such an $r$
  only by $s$ with $r\in sM$, suffix search only by $s$ with $r\in Ms$,
  and hence only by targets with $MrM\subseteq MsM$.

  For the infix routine, a split $uv=r$ gives
  $MrM\subseteq MuM$ and $MrM\subseteq MvM$.  The subsequent suffix
  search from $u$ uses only targets $s$ with $u\in Ms$, while the prefix
  search from $v$ uses only targets $s$ with $v\in sM$.  Thus in both
  cases the two-sided ideal can only grow.  In all cases
  \[
    MrM\subseteq MsM
  \]
  for every equality target $s$ actually queried.  Thus every recursive
  target satisfies the displayed strict containment in the statement;
  no equality test in the $\JJ$-class of $m$ occurs.  In the splitting
  routine we enumerate the set of at most $q^2$ pairs $(u,v)$ with
  $uv=r$; no sharper factorization bound is needed.

For (U), implement the binary search in \lemref{lem:prefix} with exact
dual composition.  Its prefix-membership predicate is an OR over at
most $q$ targets, and its depth, including the final equality and
letter tests, is at most $H$.  The cost for one pair $(r,a)$ is therefore
at most $2H(24\sqrt{q+1}\,B\sqrt n+2)$.  Taking the OR over at most
$q^2$ pairs gives cost at most
\[
 2H(24\sqrt{q+1}\,B\sqrt n+2)\,24\sqrt{q^2+1}
 \leq2^{11}B(q+1)^2H\sqrt n.
\]
Condition (V) has the same bound.  Condition (C) is the complement of
an OR over the positions, with cost at most $48\sqrt{n+1}\leq96\sqrt n$.

For (W), use the windows and cuts from \lemref{lem:infix}, replacing
each search by an OR dual.  At scale $\ell$, the two boundary searches
for a fixed split pair cost at most
$8H(24\sqrt{q+1}\,B\sqrt{4\ell}+2)$.
The windows are clipped to the input, so all recursive lengths remain
at most $n_0$.  The grid of cuts can be chosen with
$4\ell(|\mathrm{grid}|+1)\leq68n$ and
$|\mathrm{grid}|+1\leq n+2$.  Composing the searches over grid cuts,
at most $q^2$ split pairs, and at most $\lfloor\log_2 n\rfloor$ scales
gives the following bound for one triple:
\[
\begin{aligned}
 &8H\cdot24\bigl(24B\sqrt{68(q+1)n}+2\sqrt{n+2}\bigr)\\
 &\qquad{}\times24\sqrt{q^2+1}\cdot
       24\sqrt{\lfloor\log_2 n\rfloor+1}
 \leq2^{25}B(q+1)^2H^2\sqrt n.
\end{aligned}
\]
Here we used $\sqrt{68n}\leq9\sqrt n$,
$\sqrt{n+2}\leq2\sqrt n$, and
$\sqrt{\lfloor\log_2 n\rfloor+1}\leq H$.
Taking the OR over at most $q^3$ triples adds a factor
$24\sqrt{q^3+1}\leq24(q+1)^2$.  Complementing this OR gives (W)
with cost at most $2^{30}B(q+1)^4H^2\sqrt n$.

This last bound also covers (U), (V), and (C).  Combining the four
conditions by a balanced binary AND costs at most sixteen times their
largest dual cost.  The resulting dual has cost
\[
 T_m(n)\leq2^{34}B(q+1)^4H^2\sqrt n
 \leq A_{n_0}B\sqrt n.
\]

For the query statement, adversary tightness converts the assumed
algorithms to all-pairs duals of cost $O(B\sqrt\ell)$.  Enlarge their
common coefficient to $O(B+1)\geq1$ and apply the dual statement.
Converting the resulting dual to a quantum algorithm proves the claim,
since $H=O(L(n_0))$ and $q\geq2$.  The added $1$ covers the letter tests
in condition (C).
\end{proof}

Every recursive equality target in this lemma has a strictly larger
two-sided ideal than the current target.  Iterating the lemma along an
ideal chain gives the global bound below.  In the cube-root argument,
we will instead organize the recursion around regular classes.

\begin{theorem}\label{thm:main-ags}
Let $M$ be a finite aperiodic monoid.  For every $n\geq1$,
\[
 Q_{1/3}(\operatorname{Prod}_{M,n})
 \leq\min\left\{n,\sqrt n\,
   \bigl((|M|+1)\log(n+2)\bigr)^{O(\dJ(M)+1)}
 \right\}.
\]
The implicit constant in the exponent is absolute.
\end{theorem}

\begin{proof}
Fix a length bound $n_0\geq1$ and take $q$, $H$, and $A_{n_0}$ as in
\lemref{lem:ags-local-step}.  We prove by strong induction on $d(m)$ that,
for every $0\leq n\leq n_0$, membership in $P_m\cap M^n$ has an
all-pairs dual of cost
\begin{equation}\label{eq:ags-target-dual}
 T_m(n)\leq A_{n_0}^{d(m)+1}\sqrt n.
\end{equation}
The empty input has cost zero.  Below assume $1\leq n\leq n_0$.

If $d(m)=0$, then $m=1$.  By \propref{prop:identity}, membership in
$P_1$ is the complement of an OR of the predicates $x_i\neq1$.
The OR estimate above gives a dual of cost
\[
 48\sqrt{n+1}\leq96\sqrt n\leq A_{n_0}\sqrt n.
\]
Thus the base case uses one factor of $A_{n_0}$.

If $d(m)=k>0$, every strict-parent target $s$ satisfies $d(s)<k$.
The induction hypothesis bounds its dual cost on any interval of length
$\ell\leq n_0$ by $A_{n_0}^k\sqrt\ell$.  Apply
\lemref{lem:ags-local-step} with $B=A_{n_0}^k\geq1$ to obtain
$T_m(n)\leq A_{n_0}^{k+1}\sqrt n$, completing the induction.

The vector of all $q$ membership answers determines the product:
exactly one target accepts.  Combining their duals and using
$d(m)\leq\dJ(M)$ gives a product dual of cost at most
\[
 2q\,A_{n_0}^{\dJ(M)+1}\sqrt n.
\]
No error reduction is needed in this construction.  Taking $n_0=n$,
adversary tightness~\eqref{eq:adversary-tightness} gives
\[
 Q_{1/3}(\operatorname{Prod}_{M,n})
 =O\!\left(qA_n^{\dJ(M)+1}\sqrt n\right).
\]
Since $A_n=O((q+1)^5\log^2(n+2))$, substituting this estimate and
absorbing absolute constants into the exponent gives the second bound
in the theorem.  Reading all input letters gives the alternative bound $n$.
\end{proof}

\begin{remark}
	No attempt has been made to optimize the exponents; the point is the
	\emph{shape} of the bound.  Two features are essential to the method.  First, the
	recursion reduces the level of the target without requiring the input length to
	decrease, so the per-level overhead multiplies across levels.  The factor
	$\dJ(M)+1$ in the exponent counts the strict ideal descents together with the base case.  Second, for fixed
	$M$ this is the AGS bound $\tilde O(\sqrt n)$.  If $M$ varies with $n$,
	write $q=|M|$ and $D=\dJ(M)$.  Under the condition
	\[
	 (D+1)\bigl(\log(q+1)+\log\log(n+2)\bigr)=o(\log n),
	\]
	the multiplicative overhead in \thmref{thm:main-ags} is $n^{o(1)}$, giving
	\[
	 Q_{1/3}(\operatorname{Prod}_{M,n})\leq n^{1/2+o(1)}.
	\]
	The condition accounts for both the monoid-size factor and the logarithmic
	overhead at each level, and is sufficient rather than necessary for a
	sublinear bound.
\end{remark}

\begin{remark}
	The union semilattice (\exampleref{ex:union}) shows some dependence on the monoid is
	necessary: there $\dJ = m$ and the true complexity is
    $\Theta(\min\{n,\sqrt{mn}\})$.  The gap between this bound and
	\thmref{thm:main-ags} is enormous.  The breadth bounds in Sections~\ref{sec:lattice} and~\ref{sec:beta} explain how additional structure
	can give much smaller coefficients.  For general aperiodic monoids,
	the exponential dependence on $\dJ$ is necessary
	(\secref{sec:dyck-lb}).
\end{remark}

\subsection{The Dyck lower bound and its transition monoid}
\label{sec:dyck-lb}

The AGS bound has an exponential dependence on ideal-chain depth.  We
next show that this dependence cannot in general be replaced by a
polynomial in the monoid size, using bounded-depth Dyck languages.

Ambainis, Balodis, Iraids, Khadiev, K\c{l}evickis, Pr\=usis, Shen, Smotrovs, and
Vihrovs~\cite{ABIKPSSV20} proved the following lower bound for depth-bounded
Dyck languages.  Let $\textsc{Dyck}_{k,n}$ denote the problem of deciding whether
a word of length $n$ over $\{\mathtt{u}, \mathtt{d}\}$ (open/close parenthesis) is balanced with
all prefix heights in $[0, k]$.

\begin{theorem}[\cite{ABIKPSSV20}]\label{thm:dyck-lb}
There is a constant $c>1$ such that, for every even $n\geq2$ and every
$1\leq k\leq\log_2 n$,
\[
 Q_{1/3}(\textnormal{\textsc{Dyck}}_{k,n})=\Omega(c^k\sqrt n).
\]
Moreover, for every fixed $\epsilon>0$ there is a constant
$c'_\epsilon>0$ such that, for every even $n\geq2$ and every
$k\geq c'_\epsilon\log_2 n$,
\[
 Q_{1/3}(\textnormal{\textsc{Dyck}}_{k,n})=\Omega_\epsilon(n^{1-\epsilon}).
\]
\end{theorem}

Both bounds are also proved in Lean by \texttt{dyckLB\_qQuery}, from
explicit block encodings and adversary lower bounds, without assuming
the cited theorem.  The formalization permits the choices
$c=2^{1/20}$ and $c'_\epsilon=2^{\lceil1/\epsilon\rceil+1}+6$.

The language $\textsc{Dyck}_{k}$ is star free, and its syntactic monoid is small and completely
explicit:

\begin{proposition}\label{prop:dyck-monoid}
	Let $M_k$ be the transition monoid of the minimal automaton of $\textnormal{\textsc{Dyck}}_k$ (states
	$\{0, \ldots, k\}$ plus a dead state; $\mathtt{u}: h \mapsto h+1$, $\mathtt{d}: h \mapsto
	h - 1$, falling off $[0,k]$ is death).  Then every nonzero element of $M_k$ is the partial
	map $h \mapsto h + e$ with domain $\{h : a \le h \le k - b\}$, for parameters
	$(a, b, e)$ with $a, b \ge 0$, $a + b \le k$, $-a \le e \le b$ (for a word $w$: $a$ is the
	maximum prefix deficit, $b$ the maximum prefix height, $e$ the total sum).  Consequently
	\[
		|M_k| = \sum_{s=0}^{k} (s+1)^2 + 1 = \frac{(k+1)(k+2)(2k+3)}{6} + 1 = \Theta(k^3)
		\enspace,
	\]
	$M_k$ is aperiodic, and its $\JJ$-classes are exactly the level sets of $a + b$ together
	with the zero element: they form a chain, and $\dJ(M_k) = k + 1$.  For $k\ge1$, $M_k$
	is neither $\RR$-trivial nor $\LL$-trivial.
\end{proposition}

\begin{proof}
	Words act on the left: the element of a word $w = w_1 \cdots w_\ell \in
	\{\mathtt{u},\mathtt{d}\}^*$ sends $h$ to the state reached after reading $w$ from $h$.
	Write $\Sigma_j$ for the sum of the first $j$ letters ($\mathtt{u} = +1$,
	$\mathtt{d} = -1$), and set $a(w) = -\min_j \Sigma_j \ge 0$, $b(w) = \max_j \Sigma_j \ge
	0$, $e(w) = \Sigma_\ell$.  Reading $w$ from $h$ survives iff $0 \le h + \Sigma_j \le k$
	for all $j$, i.e., iff $a(w) \le h \le k - b(w)$, and then the result is $h + e(w)$.  So
	the element of $w$ is the claimed partial map with parameters
	$(a, b, e) = (a(w), b(w), e(w))$, which satisfy $a, b \ge 0$ and $-a \le e \le b$; it is
	the zero element (empty domain) iff $a + b > k$.  Distinct valid triples give distinct
	partial maps (the domain determines $(a, b)$ and the shift determines $e$).  Every valid
	triple is realized: the word $\mathtt{d}^a \mathtt{u}^{a+b} \mathtt{d}^{b-e}$ has prefix
	minimum $-a$, prefix maximum $b$, and sum $e$.  Counting triples with $a + b = s$ gives
	$(s+1)$ choices of $(a,b)$ and $s + 1$ choices of $e \in [-a, b]$, whence the sum
	$\sum_{s=0}^{k}(s+1)^2$ plus the zero element.

	Composition (apply $x = (a_1, b_1, e_1)$ then $y = (a_2, b_2, e_2)$) is
	\[
		x \cdot y = \big( \max(a_1,\, a_2 - e_1),\ \max(b_1,\, b_2 + e_1),\ e_1 + e_2 \big)
	\]
	(saturating to zero when the new parameters exceed budget), since a prefix of the
	concatenated word is a prefix of $x$'s word or $x$'s word followed by a prefix of
	$y$'s.  Aperiodicity: $x^N$ has $e$-parameter $N e_1$; if $e_1 < 0$ its $a$-parameter
	$\max_{1 \le j \le N}(a_1 - (j-1)e_1)$ grows without bound, and if $e_1 > 0$ its
	$b$-parameter does, so $x^N = 0$ for large $N$; if $e_1 = 0$ then $x^2 = x$.  Either way
	$x^{N+1} = x^N$ eventually.

	$\JJ$-classes.  From the composition formula, the parameter $s = a + b$ of any product
	is at least the $s$-parameter of each factor: multiplying $x$ on the left by
	$(a_0, b_0, e_0)$ yields $a' + b' \ge (a_1 - e_0) + (b_1 + e_0) = s(x)$, and similarly
	on the right.  So $s$ never decreases and the level sets are unions of $\JJ$-classes,
	ordered in a chain (with the zero element at the bottom).  Conversely all of level $s$
	is one $\JJ$-class: right-multiplying $(a, b, e)$ by $\mathtt{u}^j$ gives
	$(a, b, e + j)$ as long as $e + j \le b$, and by $\mathtt{d}^i$ gives $(a, b, e - i)$ as
	long as $e - i \ge -a$, so $e$ ranges over all of $[-a, b]$ within a fixed $(a, b)$;
	and left-multiplying by $\mathtt{d}^i$ slides $(a, b, e) \mapsto (a + i, b - i, e - i)$
	while left-multiplying by $\mathtt{u}^j$ slides back, covering all $(a, b)$ with
	$a + b = s$.  Hence the $\JJ$-order is a chain of $k + 2$ classes (levels
	$s = 0, \ldots, k$ and zero) and $\dJ(M_k) = k + 1$.

	For non-$\RR$-triviality, assume $k\ge1$ and let
	$x = \mathtt{ud} = (0, 1, 0)$ and $y = \mathtt{u} = (0,1,1)$.
	Then $y \cdot \mathtt{d} = x$ and $x \cdot \mathtt{u} = y$, so $xM_k = yM_k$ with
	$x \neq y$.  For non-$\LL$-triviality, let $x' = \mathtt{du} = (1, 0, 0)$.
	Then $\mathtt{d} \cdot y = x'$ and $\mathtt{u} \cdot x' = y$, so $M_kx' = M_ky$
	with $x' \neq y$.
\end{proof}

The product of the generator images in $M_k$ determines the state reached
from the initial state $0$.  A word belongs to $\textsc{Dyck}_k$ exactly
when that state is again $0$.  Thus computing the product decides Dyck
membership without further queries.  Combining this reduction with the
size and ideal-depth formulas gives the following lower bound.

\begin{corollary}\label{cor:dyck}
	The family $M_k$ of aperiodic monoids, with $|M_k| = \Theta(k^3)$ and
	$\dJ(M_k) = k + 1$, satisfies, for $n\geq2$ and $1\leq k\leq\log_2 n$,
	\[
		Q_{1/3}(\operatorname{Prod}_{M_k,n}) \ = \
		\Omega\big( \sqrt{n} \cdot c^{\,\dJ(M_k)} \big) \ = \
		\Omega\big( \sqrt{n} \cdot 2^{\Omega(|M_k|^{1/3})} \big) \enspace.
	\]
	In particular: (i) no algorithm achieves $\poly(|M|) \cdot \sqrt{n}$ for all finite
	aperiodic monoids---any bound of the form $f(|M|) \sqrt{n}$ requires
	$f(|M|) = \exp(\Omega(|M|^{1/3}))$; (ii) the exponential dependence on $\dJ$ in
	\thmref{thm:main-ags} is qualitatively necessary: the upper bound is
	$\sqrt{n} \cdot 2^{O(\dJ \log(|M| \log n))}$ against the lower bound
	$\sqrt{n} \cdot 2^{\Omega(\dJ)}$; (iii) for every fixed $0<\epsilon<1$,
	a suitable choice $k=O_\epsilon(\log n)$ gives an aperiodic
	monoid of size $O_\epsilon(\log^3 n)$ whose product problem requires
	$\Omega_\epsilon(n^{1-\epsilon})$
	queries---once the monoid has merely polylogarithmic size, the star-free
	$\tilde O(\sqrt n)$ phenomenon disintegrates entirely.
\end{corollary}

\begin{proof}
Put $n_0=2\lfloor n/2\rfloor$, so $n/2\leq n_0\leq n$ and $n_0$ is even.
Map a Dyck input of length $n_0$ to its letter images in $M_k$ and append
known identities to reach length $n$.  Computing this product decides the
Dyck instance with constant query overhead.  Thus the lower bounds at
length $n_0$ transfer to the monoid product at length $n$.

For integer $k\geq1$, the condition $k\leq\log_2 n$ implies
$2^k\leq n_0$.  The first part of \thmref{thm:dyck-lb}, together with
the size and depth formulas for $M_k$, proves the displayed bound and
consequences (i) and (ii).
For (iii), take $k=\lceil c'_\epsilon\log_2 n\rceil$.  This choice satisfies
$k\geq c'_\epsilon\log_2 n_0$, so the second part of the theorem gives
$\Omega_\epsilon(n_0^{1-\epsilon})=\Omega_\epsilon(n^{1-\epsilon})$.
The monoid size is $O_\epsilon(\log^3 n)$.
\end{proof}

\subsection{Why polynomial dependence on breadth fails}
\label{sec:dyck-breadth}

The same family shows why the breadth upper bounds for commutative and
stably ordered monoids need their structural hypotheses.  We first
compute its product breadth.

\begin{proposition}[Product breadth of the Dyck monoid]\label{prop:dyck-breadth}
  For $k\geq1$ and $G=\{\mathtt{u},\mathtt{d}\}\subseteq M_k$,
  \[
    \beta_G(M_k)=\beta_{M_k}(M_k)=\max\{2k,3k-2\}.
  \]
  In particular, the breadth is $2$ for $k=1$ and $3k-2$ for $k\geq2$.
\end{proposition}

\begin{proof}
  The word $\mathtt{u}^k\mathtt{d}^k$ has maximum prefix height $k$
  and final height zero, so every product core retains all $2k$ letters.
  For $k\geq2$, the word
  $\mathtt{u}^{k-1}\mathtt{d}^k\mathtt{u}^{k-1}$ has maximum prefix
  height $k-1$, minimum $-1$, and final height $k-2$.
  Any subword with the same product must attain its maximum in the first
  run, forcing all $k-1$ initial letters.  Preserving the minimum then
  forces all $k$ middle letters, and preserving the final height forces
  all $k-1$ final letters.  These words give the lower bounds already
  over $G$.

  For the upper bound, first consider a shortest product core
  $x_1\cdots x_\ell$ with nonzero product, allowing arbitrary factors
  $x_i\in M_k$.  Represent each factor by a word over $G$ and concatenate
  these representations.  Let $A$ and $B$ be the minimum and maximum
  heights of the resulting walk, and put $s=B-A\leq k$.
  If $s=0$, the product is the identity and the empty core suffices.
  Otherwise, consider the $\ell+1$ heights at factor boundaries.
  A segment between two equal boundary heights has net displacement
  zero.  Deleting its factors leaves all other heights unchanged, so
  minimality implies that its deletion removes every occurrence of
  $A$ or every occurrence of $B$.

  Each interior height therefore occurs at most three times among the
  factor boundaries: four occurrences would give three disjoint closed
  segments, each of which would have to contain all occurrences of one
  of the two extrema.  Each extreme height occurs at most twice, since
  deleting a closed segment based at that extreme preserves it and must
  remove every occurrence of the other extreme.
  If neither extreme repeats at factor boundaries, then
  $\ell+1\leq3(s-1)+2$, giving $\ell\leq3s-2$.
  If an extreme repeats, every occurrence of the other extreme lies
  between its two visits.  No height can then occur three times at
  factor boundaries: the two resulting closed segments would have to
  contain all occurrences of the two extrema separately, which this
  ordering precludes.  The other extreme cannot repeat either.
  Thus $\ell+1\leq2(s-1)+3$, giving $\ell\leq2s$.
  In either case, $\ell\leq\max\{2k,3k-2\}$.

  Finally, every zero product has a core of length at most $k+1$.
  A zero factor gives a one-letter core, so suppose all factors are
  nonzero, with parameters $(a_i,b_i,e_i)$ as in
  \propref{prop:dyck-monoid}.  In the expanded walk the maximum and
  minimum differ by more than $k$.  By symmetry, suppose a minimum
  occurs in factor $i$ before a maximum in factor $j$.  The two factors
  are distinct, since each factor has span at most $k$, and
  \[
    a_i+e_i+\sum_{r=i+1}^{j-1}e_r+b_j>k.
  \]
  Delete the intervening factors with nonpositive displacement.  The
  remaining contributions to this rise are positive integers, apart
  from possible zero contributions $a_i+e_i$ and $b_j$, whose factors
  can also be omitted.  Retain factors only until these contributions
  first sum to more than $k$.  This uses at most $k+1$ factors and still
  gives a walk of span greater than $k$, hence a zero product.
  Since $k+1\leq\max\{2k,3k-2\}$ for $k\geq1$, this completes the
  upper bound for arbitrary monoid inputs.
\end{proof}

Together with Corollary~\ref{cor:dyck}, the proposition gives a lower bound
$\sqrt n\,2^{\Omega(\beta_{M_k}(M_k))}$ in the range
$1\leq k\leq\log_2 n$.  Thus a bound with only a polynomial factor in
product breadth cannot hold for all finite aperiodic monoids.

The Dyck monoids $M_k$ have nontrivial $\JJ$-classes, so they are not
$\JJ$-trivial.
Consequently the lower bound does not contradict the ordered-monoid
breadth theorem.  Indeed, a stable partial order in which the identity is
the minimum element forces $\JJ$-triviality: if $x=uyv$ and $y=u'xv'$, insertion monotonicity gives
$y\leq x\leq y$, and hence $x=y$.
For finite $\JJ$-trivial monoids, the bound
$O(\min\{n,\sqrt{n|M|}\})$ already follows from
\thmref{thm:rtrivial}, since $\JJ$-triviality implies
$\RR$-triviality and $\dR(M)\leq |M|-1$.
The capped counters in \propref{prop:jtrivial-log-fails} show that
replacing $|M|$ by $\log|M|$ is nevertheless impossible in general.

\section{A cube-root bound in terms of monoid size}
\label{sec:ags-cuberoot-size}

This section proves a bound whose exponent grows at the cube-root scale
in the monoid size, up to a logarithmic factor.  We use the localized AGS
step from \secref{sec:ags-local} in two ways.  First, we organize the
recursion through regular classes to obtain a bound for finite monoids
of matrices in terms of their matrix dimension.  We then represent a
general monoid by several matrix coordinates: small coordinates are
handled by that bound, while a large coordinate can be recovered using
products in a smaller quotient monoid.  Finally, we recover the
information lost by the matrix coordinates and combine the costs into
a recurrence on monoid size.

Throughout, write $N=|M|$ and $L(n)=2+\log_2(n+2)$.
Our goal is \thmref{thm:ags-cuberoot-size}, proved in
\secref{sec:ags-size-recurrence}: for every finite aperiodic monoid $M$
and every $n\geq1$,
\[
 Q_{1/3}(\operatorname{Prod}_{M,n})
 \leq\min\left\{n,\sqrt n\,L(n)^{C(N\log(N+2))^{1/3}}\right\},
\]
for a universal constant $C$.

The one-element monoid requires no queries, so the constructions below may assume
$N\geq2$.  The Dyck lower bound in Corollary~\ref{cor:dyck} already
shows that an overhead exponential in $N^{1/3}$ is unavoidable in the
relevant range.  After proving the upper bound, we return to this family
to interpret its algebraic structure.

\subsection{Regular actions and their simulation}
\label{sec:ags-actions}

A right action of a monoid $H$ on a set $X$ assigns a state $ya$ to
$y\in X$ and $a\in H$, with $y1=y$ and $(ya)b=y(ab)$.  The orbit
$vH$ consists of the states reachable from $v$.  The action graph has
edges $y\longrightarrow ya$; two states lie in the same strongly
connected component when each is reachable from the other.

We use the order
\[
  [x]_{\JJ}\leq_{\JJ}[y]_{\JJ}\quad\Longleftrightarrow\quad
  MxM\subseteq MyM.
\]
A $\JJ$-class is \emph{regular} if it contains an idempotent.  For a
regular class $K$, let
\[
  h(K)=\max\{k:K=K_0<_{\JJ}K_1<_{\JJ}\cdots<_{\JJ}K_k,
                  \text{all $K_i$ regular}\}.
\]
Thus $h(K)$ counts strict comparisons along a longest chain of regular
classes above $K$.  It differs from $d(m)$ in the AGS induction, which
also counts nonregular classes.
For an idempotent $e$, write
$R_e=\{r:r\mathrel{\RR}e\}$ for its $\RR$-class in the monoid under
discussion.  Its
\emph{killed right action} has live states $R_e$ and one absorbing state
$\dagger$:
\[
  r\mathbin{\cdot}a=
  \begin{cases}
    ra,&ra\mathrel{\RR}e,\\
    \dagger,&ra\not\mathrel{\RR}e.
  \end{cases}
  \qquad \dagger\mathbin{\cdot}a=\dagger.
\]
This is a monoid action: if $ra\not\mathrel{\RR}e$, then no further right
multiple of $ra$ can be $\RR$-equivalent to $e$, because principal right
ideals only decrease under right multiplication.
An $R_e$-axis packet, from any supplied live state $r$, returns either the
exact final live state or the first position at which the trajectory leaves
$R_e$, together with its exact live predecessor.  Here and below ``exact''
describes the value returned on the successful branch; every packet is a
bounded-error quantum algorithm.  The robust-composition convention from
Section~\ref{sec:prelim} supplies the required amplification.  We absorb
fixed powers of $N$ and $L(n)$ caused by finite-output amplification into
the notation $N^{O(1)}L(n)^{O(1)}$.

We first record a consequence of the Sandwich Lemma that replaces the
usual finite-semigroup stability lemma.

\begin{lemma}[Green stability]\label{lem:ags-green-stability}
  Let $H$ be a finite aperiodic monoid.  If $a,b\in H$,
  $a\mathrel{\JJ}b$, and $aH\subseteq bH$, then
  $a\mathrel{\RR}b$.  Dually, if $Ha\subseteq Hb$, then
  $a\mathrel{\LL}b$.
\end{lemma}

\begin{proof}
  Write $a=bu$ and, using $a\mathrel{\JJ}b$, write $b=xay$, with all
  factors in $H$.
  Recall the Sandwich Lemma (\lemref{lem:sandwich}): $pqr=q$ implies
  $pq=q=qr$.  Applying it to $a=x a(yu)$ gives $a=xa$, whence
  $b=xay=ay\in aH$.  Thus $aH=bH$.  The other assertion is symmetric.
\end{proof}

Let a finite aperiodic monoid $H$ act on the right of a finite set $X$.
For a strongly connected component $D$, define its killed action on
$D\cup\{\dagger\}$ by retaining $ya$ when it belongs to $D$ and sending
it to $\dagger$ otherwise, with $\dagger a=\dagger$.  Once a trajectory
leaves $D$ it cannot return: every intermediate state on a return path
would belong to the same component.  This verifies the action law.

Fix a base point $v\in D$ and call $u\in H$ a return element if $vu=v$.
Choose an idempotent $e$ whose
$\JJ$-class is minimal among the idempotent powers of return elements.
Here $x^\omega$ denotes an idempotent positive power of $x$, which exists
for every element of a finite monoid.
Such an $e$ exists because $1$ is a return element.  Every positive power
of a return element is again a return element, so $ve=v$.  We call $e$
a \emph{regular owner} of $D$ for this base point.  The next lemma
explains how its regular $\RR$-class simulates the component.

\begin{lemma}[Regular-owner cover]\label{lem:ags-owner-cover}
  The map
  \[
    \pi_e:R_e\longrightarrow D,\qquad \pi_e(r)=vr,
  \]
  is onto.  Extend it by $\pi_e(\dagger)=\dagger$.  Then, for every
  $r\in R_e$ and $a\in H$,
  \[
    ra\in R_e\quad\Longleftrightarrow\quad (vr)a\in D,
    \qquad \pi_e(r\mathbin{\cdot}a)=\pi_e(r)\mathbin{\cdot}a,
  \]
  where the dots denote the two killed actions.  Thus the map
  intertwines these actions, including their transitions to the dead state.
\end{lemma}

\begin{proof}
  Let $y=va\in D$ and choose $b$ with $y b=v$.  Put $r=ea$ and
  $u=eab$.  Then $vu=v$, so the idempotent power $f=u^\omega$ is a return
  idempotent.  Moreover
  \[
    HfH\subseteq HuH\subseteq HrH\subseteq HeH.
  \]
  The minimal choice of $e$ forces equality throughout.  In particular
  $r\mathrel{\JJ}e$, while $rH\subseteq eH$; hence
  $r\mathrel{\RR}e$ by \lemref{lem:ags-green-stability}.  Finally
  $vr=vea=va=y$, proving surjectivity.

  If $r\mathrel{\RR}e$, then $vr$ is reachable from $v$, and a relation
  $e=rb$ shows that $(vr)b=v$; hence $vr\in D$.  This proves the forward
  implication after replacing $r$ by $ra$.

  Conversely, suppose $(vr)a\in D$, and choose $b$ with $vrab=v$.
  The idempotent power of $rab$ is a return idempotent, and
  \[
    H(rab)^\omega H\subseteq HrabH\subseteq HraH
      \subseteq HrH=HeH.
  \]
  Minimality again forces equality throughout.  Thus
  $ra\mathrel{\JJ}r$ and $raH\subseteq rH$, so Green stability gives
  $ra\mathrel{\RR}r\mathrel{\RR}e$.  The intertwining identity is
  associativity.
\end{proof}

The preceding lemma turns packets on regular $\RR$-classes into packets
for arbitrary finite actions.  Indeed, contract the strongly connected
components of a reachable orbit.  A trajectory visits each component at
most once.  In the current component choose the precomputed lift of its
current state under $\pi_e$ and
run the $R_e$ first-exit packet.  A live answer decodes through $\pi_e$.
On an exit, the returned predecessor and one query to the boundary letter
determine the exact state in the next component.  Since a reachable orbit
of a $q$-element monoid has at most $q$ states, robust sequential
composition gives the following statement.

\begin{lemma}[Action compiler]\label{lem:ags-action-compiler}
  Suppose the regular owners of all live strongly connected components in
  one reachable orbit of a $q$-element aperiodic monoid have axis packets of
  cost at most $B\sqrt\ell$ on every interval of length $\ell\leq n$.
  Then the endpoint of that orbit can be computed in
  \[
    q^{O(1)}L(n)^{O(1)}(B+1)\sqrt n
  \]
  queries, uniformly in the supplied live starting state.  If the action
  has no designated dead state, all components count as live.  When it has
  an absorbing dead state, the same bound returns the first-death position
  together with its exact live predecessor.
\end{lemma}

\begin{proof}
  Use the component-by-component procedure just described.  At most $q$
  owner packets and at most $q$ boundary-letter queries occur on any
  branch.  Amplify these finite-output calls and apply the robust-composition
  principle.  If the action has an absorbing dead state, stop at the first
  transition to it; no owner packet is invoked for the dead component.
\end{proof}

To connect these action routines to product equality, fix $t\in M$ and put
\[
  A_t=\{p\in M:t\in pM\}.
\]
Totalize $A_t$ by a dead state and define
\[
  p\mathbin{\cdot}a=
  \begin{cases}
    pa,&pa\in A_t,\\
    \dagger,&pa\notin A_t,
  \end{cases}
  \qquad \dagger\mathbin{\cdot}a=\dagger.
  \tag{\(*\)}\label{eq:ags-target-action}
\]

\begin{lemma}[Target-reachability action]\label{lem:ags-target-action}
  Equation~\eqref{eq:ags-target-action} defines a monoid action.  From the
  initial state $1$, its endpoint on $w$ is $[w]$ if
  $t\in[w]M$, and is $\dagger$ otherwise.  Consequently
  \[
    [w]=t
    \quad\Longleftrightarrow\quad
    \text{the endpoint of \eqref{eq:ags-target-action} is $t$}.
  \]
  Moreover, if $D$ is a live component, $p\in D$, and $e$ is the owner
  supplied by \lemref{lem:ags-owner-cover}, then
  \[
    MtM\subseteq MpM\subseteq MeM.
  \]
  In particular, if a regular class $K$ satisfies
  $K<_{\JJ}[t]_{\JJ}$, then $K<_{\JJ}[e]_{\JJ}$ and
  $h([e]_{\JJ})\leq h(K)-1$.
\end{lemma}

\begin{proof}
  If $pa\notin A_t$, then $t\notin paM$, and
  $pabM\subseteq paM$ shows that no continuation can return to $A_t$.
  This proves the action law, and induction on the word proves the endpoint
  formula.  Since $t\in tM$, the equality test follows.

  Liveness of $p$ gives $t=pz$ for some $z$, so
  $MtM\subseteq MpM$.  Let $v\in D$ be the base point used to choose
  the owner $e$.  By \lemref{lem:ags-owner-cover}, $p=vr$ for some
  $r\mathrel{\RR}e$.  Then $r\in eM$, so $p\in MeM$ and hence
  $MpM\subseteq MeM$.  The last assertion follows from strictness and
  the definition of regular height.
\end{proof}

\subsection{Recursion on regular height}
\label{sec:ags-regular-height}

We now build the regular-axis packets by induction on regular height,
using the query form of \lemref{lem:ags-local-step}.  For a horizon $n_0$, let
\[
  D_h(n_0)=\max
  \left\{\frac{T_e(r,\ell)}{\sqrt\ell}:
  \begin{array}{l}
    1\leq\ell\leq n_0,\ e^2=e,\ h([e]_{\JJ})\leq h,\\
    r\in R_e
  \end{array}\right\},
\]
where $T_e(r,\ell)$ is the cost of the $R_e$-axis packet from $r$.

\begin{proposition}[Regular-axis recurrence]\label{prop:ags-axis-recurrence}
  There is an absolute constant $c_1$ such that
  \[
    D_0(n_0)\leq N^{c_1}L(n_0)^{c_1},
    \qquad
    D_h(n_0)\leq
      N^{c_1}L(n_0)^{c_1}\bigl(D_{h-1}(n_0)+1\bigr)
      \quad(h\geq1).
    \tag{\(**\)}\label{eq:ags-axis-recurrence}
  \]
  Consequently
  \[
    D_h(n_0)\leq N^{O(h+1)}L(n_0)^{O(h+1)}.
    \tag{\(\dagger\)}\label{eq:ags-axis-solved}
  \]
\end{proposition}

\begin{proof}
  The only regular class of height zero is $[1]_{\JJ}$.  Indeed
  $[1]_{\JJ}$ is the
  unique largest $\JJ$-class, and Proposition~\ref{prop:identity} implies
  that it contains only $1$.  Its $\RR$-class is also $\{1\}$.  The first
  exit is therefore the first nonidentity letter, which costs
  $O(\sqrt n)$ by minimum finding.  This proves the base case.

  Now let $h\geq1$ and assume the result at height $h-1$.  Fix $e$ of
  height at most $h$.  If $h([e]_{\JJ})=0$, use the base case.  Otherwise
  fix an incoming state $c\in R_e$ and a target $s\in R_e$; in particular
  $s\ne1$.  The predicate
  \[
    c[w]=s
  \]
  is the ordinary product-equality predicate $P_s$ after prepending the
  known letter $c$.  Apply the construction of \lemref{lem:ags-local-step}
  to this virtual input.  Its recursive targets $t$ satisfy
  $MsM\subsetneq MtM$.  An infix containing the initial $c$ has product
  in $McM=MsM$, so its equality test for such a $t$ is constantly false.
  All other recursive calls use only original unknown positions and
  retain the horizon $n_0$.  The known coordinate costs no query, and
  increasing the virtual length from $\ell$ to $\ell+1\leq2\ell$ changes
  only an absolute factor for $\ell\geq1$.

  To supply the strict-parent tests, use the target-reachability action
  of \lemref{lem:ags-target-action}.  Every one of its live component
  owners has height at most $h-1$, so
  \lemref{lem:ags-action-compiler} and the definition of $D_{h-1}$
  compute $P_t$ with normalized cost
  \[
    N^{O(1)}L(n_0)^{O(1)}(D_{h-1}(n_0)+1).
  \]
  The localized AGS step therefore computes $c[w]=s$ within the
  same form of bound.

  Run these tests for all $s\in R_e$.  Their answers are one-hot: the true
  coordinate, if one exists, is the exact live endpoint, and if none exists
  the killed endpoint is $\dagger$.  There are at most $N$ coordinates, so
  finite-output amplification changes only the polynomial factor in
  \eqref{eq:ags-axis-recurrence}.

  It remains to recover the first exit.  Right ideals decrease along a
  right trajectory, so after leaving $R_e$ the trajectory cannot return.
  Compute the endpoint of the whole interval.  On a dead answer maintain a
  known live cut $\ell$, its exact state, and a known dead upper cut $u$.
  Query the killed endpoint on the difference block from $\ell$ to the
  midpoint of $[\ell,u]$.  A live answer advances $\ell$ and retains the
  new exact state; a dead answer lowers $u$.  To handle integer rounding,
  put $g=u-\ell-1$.  Either child gap has its new value of $g$ at most
  $g/2$, while the queried block has length
  $\lfloor(u-\ell)/2\rfloor\leq g$ whenever $g\geq1$.
  After the initial whole-interval call, the queried lengths are
  therefore bounded by $n-1,(n-1)/2,(n-1)/4,\ldots$.
  The sum of their square roots is $O(\sqrt n)$.  When the gap has one
  letter, the retained state is the exact predecessor and one query gives
  the exiting letter.  Amplify each of the $O(\log n)$ adaptive calls to
  error $O(1/\log(n+2))$; its $O(\log\log(n+2))$ overhead is absorbed by
  $L(n_0)^{O(1)}$.  This proves the recurrence.  To solve it, put
  $A=N^{c_1}L(n_0)^{c_1}$ and enlarge $c_1$ so that $A\geq2$.  Then
  \[
    D_h(n_0)+1\leq2A\bigl(D_{h-1}(n_0)+1\bigr),
  \]
  and iteration proves \eqref{eq:ags-axis-solved}.
\end{proof}

\begin{proposition}[Regular-action compiler by height]
\label{prop:ags-action-height}
  If all live component owners in an action have regular height at most $h$, then
  its endpoint and first-death packet cost
  \[
    q^{O(h+1)}L(n)^{O(h+1)}\sqrt n,
    \tag{\(\ddagger\)}\label{eq:ags-action-height}
  \]
  where $q$ is the order of its transition monoid.
\end{proposition}

\begin{proof}
  Combine \lemref{lem:ags-action-compiler} with
  Proposition~\ref{prop:ags-axis-recurrence}.
\end{proof}

\subsection{A dimension-sensitive matrix bound}
\label{sec:ags-matrix-bound}

The localized AGS construction depends only on the height of chains of
regular $\JJ$-classes.  In a matrix representation, ordinary matrix rank
bounds this height.

\begin{lemma}[Regular matrix-rank ceiling]
\label{lem:ags-matrix-rank-ceiling}
  Let $T\subseteq M_d(\mathbb Q)$ be a finite aperiodic matrix monoid.
  If $a,b\in T$ are regular, $b\in TaT$, and
  $\operatorname{rank}(a)=\operatorname{rank}(b)$, then
  $a\mathrel{\JJ}b$ in $T$.  Consequently every strict chain of regular
  $\JJ$-classes in $T$ has at most $d+1$ classes and at most $d$ strict
  steps.
\end{lemma}

\begin{proof}
  Write $b=xay$ and put $c=xa$.  Then
  \[
    \operatorname{rank}(b)\leq\operatorname{rank}(c)
      \leq\operatorname{rank}(a),
  \]
  so all three ranks are equal.  Thus
  $\operatorname{im}(b)=\operatorname{im}(c)$.  Choose an inner inverse
  $b^{-}\in T$, so $bb^{-}b=b$.  The map $bb^{-}$ is the identity on
  $\operatorname{im}(b)=\operatorname{im}(c)$, and hence
  \[
    c=bb^{-}c,
    \qquad b=cy.
  \]
  Therefore $c\mathrel{\RR}b$ in $T$, and in particular $c$ is regular.

  Similarly, $\operatorname{row}(c)\subseteq\operatorname{row}(a)$ and
  equality of ranks gives equality of row spaces.  If $c^{-}\in T$ is an
  inner inverse of $c$, right multiplication by $c^{-}c$ fixes every row
  in $\operatorname{row}(c)=\operatorname{row}(a)$.  Hence
  \[
    a=ac^{-}c,
    \qquad c=xa,
  \]
  and $a\mathrel{\LL}c$.  We conclude that
  $a\mathrel{\JJ}c\mathrel{\JJ}b$.

  Matrix rank is monotone along the $\JJ$-preorder.  The first assertion
  shows that it changes strictly at every strict comparison between
  regular classes.  The possible ranks are $0,1,\ldots,d$, proving the
  last assertion.
\end{proof}

Combining \lemref{lem:ags-matrix-rank-ceiling} with
Proposition~\ref{prop:ags-action-height}, applied to the right Cayley action of $T$,
gives
\begin{equation}
  Q_{1/3}(\operatorname{Prod}_{T,n})
  \leq \sqrt n\,|T|^{O(d+1)}L(n)^{O(d+1)}.
  \label{eq:ags-matrix-image-compiler}
\end{equation}
Long chains of nonregular $\JJ$-classes cause no additional loss.  The
target-reachability action of \lemref{lem:ags-target-action} charges their
recurrent components to the regular owners supplied by
\lemref{lem:ags-owner-cover}.

\subsection{Principal factors and simple coordinates}
\label{sec:ags-simple-coordinates}

We now construct matrix representations to which the preceding bound
applies.  The relevant representations are indexed by the nonzero regular
$\JJ$-classes.  We first describe the multiplication associated with one
such class.

For a semigroup $S$ and a nonempty two-sided ideal $I$, the \emph{Rees
quotient} $S/I$ identifies all elements of $I$ with one absorbing zero
and keeps distinct elements outside $I$ distinct.  Multiplication is
inherited from $S$, with any result in $I$ replaced by zero.  If $S$ is
a monoid and $I$ is proper, the quotient is a monoid with the same identity.

For a $\JJ$-class $J$ of $M$, its \emph{principal factor} $J^0$ is $J$
with a new zero adjoined and multiplication
\[
  x\circ y=
  \begin{cases}
    xy,&xy\in J,\\
    0,&xy\notin J.
  \end{cases}
\]
Equivalently, put $\mathcal I_J=MJM$ and
$\mathcal I_{<J}=\mathcal I_J\setminus J$.  When
$\mathcal I_{<J}$ is nonempty, the principal factor is
$\mathcal I_J/\mathcal I_{<J}$; otherwise we adjoin a zero.

Suppose now that $J$ is regular.  Index its $\RR$-classes by $I_J$ and
its $\LL$-classes by $\Lambda_J$.  The standard Rees description of a
regular principal factor, together with the triviality of the maximal
subgroups in an aperiodic monoid, identifies its nonzero elements with
pairs $(i,\nu)\in I_J\times\Lambda_J$.  There is a $0$--$1$ matrix
$P_J$, with rows indexed by $\Lambda_J$ and columns by $I_J$, such that
\[
  (i,\nu)\circ(j,\mu)=
  \begin{cases}
    (i,\mu),&(P_J)_{\nu j}=1,\\
    0,&(P_J)_{\nu j}=0.
  \end{cases}
\]
This is the \emph{sandwich matrix}: its entry records whether the product
stays in $J$.  If $X$ and $Y$ are coefficient matrices of linear
combinations of these pairs, their product has coefficient matrix
$XP_JY$.  The ordinary matrix rank
$d_J=\operatorname{rank}_{\mathbb Q}(P_J)$ will be the size of the
associated matrix representation.  The full matrix algebra of that size
has vector-space dimension $d_J^2$.

The monoid algebra $\mathbb Q[M]$ consists of rational linear
combinations of monoid elements, with multiplication extended linearly.
If a monoid has a zero, its \emph{contracted rational algebra} is obtained
by identifying the basis element
$0$ with the zero vector.  If $M$ has no zero, one may equivalently adjoin
a zero and then contract it.  In either case we obtain a unital algebra
$A$ of dimension at most $|M|$ into which the elements of $M$ embed
multiplicatively (with the monoid zero, if present, represented by the
zero vector).

The following form of Munn--Ponizovski\u\i{} theory
\cite{Munn55,Munn57} supplies the matrix coordinates and describes the
information that they lose.  A nonzero class means any class other than
$\{0\}$ when $M$ has a zero; if it has no zero, all regular classes are
included.

\begin{lemma}[Aperiodic Munn--Ponizovski\u\i{} coordinates]
\label{lem:ags-munn-decomposition}
  Let $M$ be a nontrivial finite aperiodic monoid and let $A$ be its contracted
  rational algebra.  For every nonzero regular $\JJ$-class $J$ of $M$, let
  $P_J$ be the $0$--$1$ sandwich matrix of its principal factor and
  $d_J=\operatorname{rank}_{\mathbb Q}(P_J)$.  There are representations
  $\rho_J:M\to M_{d_J}(\mathbb Q)$, one for each such $J$, with the
  following properties.
  \begin{enumerate}
    \item The linear extension $A\to M_{d_J}(\mathbb Q)$ of $\rho_J$ is
      surjective, and $\rho_J$ annihilates every element strictly below
      $J$ in the $\JJ$-order.
    \item The joint map
      \[
        \pi=(\rho_J)_J:A\longrightarrow\prod_{J\text{ nonzero regular}}
          M_{d_J}(\mathbb Q)
      \]
      has a nilpotent kernel $R$, of nilpotency index at most
      $\dim A\leq|M|$.
    \item $d_J^2\leq|J|$ and $\sum_J d_J^2\leq|M|$.
  \end{enumerate}
\end{lemma}

\begin{proof}
  Fix a nonzero regular class $J$, with $\RR$-classes indexed by $I_J$
  and $\LL$-classes by $\Lambda_J$.  Use the principal-factor coordinates
  just described, so
  $|J|=|I_J||\Lambda_J|$.
  Choose a rank factorization $P_J=CD$ with $C$ of size
  $|\Lambda_J|\times d_J$ and $D$ of size $d_J\times|I_J|$, a right
  inverse $D^-$ of $D$ and a left inverse $C^-$ of $C$.

  \emph{The coordinate.}  Left multiplication by $m\in M$ carries each
  $\RR$-class of $J$ either onto an $\RR$-class of $J$ or entirely
  strictly below $J$, and leaving $J$ is absorbing.  Let
  $\lambda(m)\in\{0,1\}^{I_J\times I_J}$ be the matrix of this partial
  map, so that $m\mapsto\lambda(m)$ is a monoid homomorphism.  Right
  multiplication likewise gives $\mu(m)\in\{0,1\}^{\Lambda_J\times\Lambda_J}$,
  and the two actions intertwine through the sandwich,
  \[
    P_J\,\lambda(m)=\mu(m)\,P_J,
  \]
  because both sides have $(\nu,i)$ entry equal to the indicator that
  $x\,m\,y$ stays in $J$ for $x$ in the $\LL$-class $\nu$ and $y$ in the
  $\RR$-class $i$.  Multiplying by $C^-$ on the left gives
  $D\lambda(m)=C^-\mu(m)CD$.  Put
  \[
    \rho_J(m)=D\,\lambda(m)\,D^-.
  \]
  Then
  $\rho_J(m)\rho_J(m')
    =C^-\mu(m)C\,(DD^-)\,D\lambda(m')D^-
    =D\lambda(m)\lambda(m')D^-=\rho_J(mm')$,
  so $\rho_J$ is a representation.  If $m$ lies strictly below $J$ then
  $mx$ lies strictly below $J$ for every $x\in J$, so $\lambda(m)=0$ and
  $\rho_J(m)=0$.  For an element $z=(i,\nu)$ of $J$, left multiplication by
  $z$ sends the $\RR$-class $i'$ to $i$ when $(P_J)_{\nu i'}=1$ and below
  $J$ otherwise, so $\lambda(z)=e_i\,(P_J)_{\nu\cdot}$ and
  \[
    \rho_J(z)=D\,e_i\,(P_J)_{\nu\cdot}\,D^-
      =(De_i)\,(C_{\nu\cdot}DD^-)=D_{\cdot i}\,C_{\nu\cdot},
  \]
  the outer product of column $i$ of $D$ with row $\nu$ of $C$.  The
  columns of $D$ span $\mathbb Q^{d_J}$ and the rows of $C$ span its
  dual, so these outer products span $M_{d_J}(\mathbb Q)$.  This proves
  item~1.

  \emph{Nilpotency.}  For $s\geq0$ let $A_s\subseteq A$ be the span of
  the elements $m$ with $|MmM|\leq s$.  Each $A_s$ is a two-sided ideal,
  $A_0=0$ and $A_{|M|}=A$.  Distinct $\JJ$-classes whose principal ideals
  have the same size are incomparable, so modulo $A_{s-1}$ the elements
  of $A_s$ multiply as in the direct sum of the contracted algebras of the
  principal factors of the classes $J$ with $|MJM|=s$: writing $X_J$ for the
  coefficient matrix of $x\in A_s$ on a regular such class $J$, the
  product $xy$ has coefficient matrix $X_JP_JY_J$ on $J$ modulo
  $A_{s-1}$, and $\rho_J(x)=DX_JC$ by the formula above.  Let
  $K_s=\ker\pi\cap A_s$ and let $x\in K_s$, $y,y'\in A_s$.  On a regular
  class $J$ with $|MJM|=s$ the product $yxy'$ has coefficient matrix
  \[
    Y_JP_JX_JP_JY'_J=Y_JC\,(DX_JC)\,DY'_J=0
  \]
  modulo $A_{s-1}$, and a nonregular class with $|MJM|=s$ has zero
  principal-factor multiplication.  Hence $A_sK_sA_s\subseteq A_{s-1}$, so
  $K_s^3\subseteq K_{s-1}$ and $K_{|M|}^{3^{|M|}}\subseteq K_0=0$: the
  kernel $R=K_{|M|}$ is nilpotent.  For the index, note that $R$ is a
  proper ideal, since $\pi(1)$ is the tuple of identity matrices and the
  class of $1$ is regular.  If $R^{j+1}=R^j$ then $R^j=R^{j+d}$ for every
  $d$, so $R^j=0$; hence the dimensions of the nonzero powers of $R$
  strictly decrease, starting below $\dim A$, and $R^{\dim A}=0$.  This
  proves item~2.

  \emph{Counting.}  Since $P_J$ has $|\Lambda_J|$ rows and $|I_J|$
  columns, $d_J\leq\min\{|I_J|,|\Lambda_J|\}$, so $d_J^2\leq|I_J||\Lambda_J|
  =|J|$.  The classes are disjoint, so $\sum_Jd_J^2\leq\sum_J|J|\leq|M|$.
\end{proof}

The \emph{apex} of a simple representation is the unique minimal
$\JJ$-class on which the representation is nonzero.  This class is regular,
and all classes strictly below it act as zero.
The class $J$ in \lemref{lem:ags-munn-decomposition} is the apex of
$\rho_J$.  We call its full matrix-algebra image the associated
\emph{simple block}.
We call $R$ the \emph{radical} of $A$.  Classically $R$ is the Jacobson
radical, $\pi$ is onto, and one obtains the Munn--Ponizovski\u\i{}
isomorphism $A/\operatorname{Rad}(A)\cong\prod_J M_{d_J}(\mathbb Q)$
\cite{Munn55,Munn57}.  Neither the identification of $R$ with the
Jacobson radical nor the surjectivity of the joint map is used below:
the fixed-apex peel needs only item~1 and the bound $d_J^2\le|J|$, and
the lift of Section~\ref{sec:ags-radical-lift} needs only that $R$ is a
nilpotent ideal of index at most $N$.  The proof above establishes these
properties directly.

For example, the Brandt semigroup of $d$-by-$d$ matrix units consists
of $0$ and elements $e_{ij}$, with $e_{ij}e_{\ell r}=e_{ir}$ when
$j=\ell$ and zero otherwise.  Its nonzero class has sandwich matrix
$\mathrm I_d$, so its degree is $d$ and its cardinality is $d^2$.
For $d\geq2$, adjoining an identity gives a distinct monoid.  We will
return to these examples after the upper-bound proof.

Fix an integer threshold $t\geq2$.  Call a block $J$ \emph{small} if
$d_J\leq t$.  Its image
$\rho_J(M)$ is a finite aperiodic matrix monoid of order at most $N$, so
\eqref{eq:ags-matrix-image-compiler} computes its coordinate with cost
\begin{equation}
  \sqrt n\,N^{O(t+1)}L(n)^{O(t+1)}.
  \label{eq:ags-small-simple-coordinate}
\end{equation}
All small coordinates can be assembled with one further polynomial factor
in $N$: amplify each coordinate to error $O(1/N)$ and combine their finite
outputs.  The resulting amplification factor is absorbed in the notation in
\eqref{eq:ags-small-simple-coordinate}.

\subsection{Reduction at a fixed apex}
\label{sec:ags-apex-reduction}

For a block of large degree, we use a product algorithm for a smaller
Rees quotient.  Fix its apex $J$ and collapse the ideal
$\mathcal I_J=MJM$ to zero.  Outside this ideal, the quotient determines
the product exactly.  When the product enters the ideal, the localized
AGS step recovers the additional information needed for the
$J$-coordinate.

\begin{lemma}[Fixed-apex Rees peel]
\label{lem:ags-fixed-apex-peel}
  Let $J$ be a nonzero, nonidentity regular $\JJ$-class of a finite
  aperiodic monoid $M$, let $\rho_J$ be its simple rational representation,
  and put
  \[
    \mathcal I_J=MJM,
    \qquad \overline M_J=M/\mathcal I_J.
  \]
  Define the normalized quotient cost
  \[
    D_J(n)=\max_{1\leq\ell\leq n}
      \frac{Q_{1/3}(\operatorname{Prod}_{\overline M_J,\ell})}{\sqrt\ell}.
  \]
  Then
  \begin{equation}
    Q_{1/3}(\rho_J\circ\operatorname{Prod}_{M,n})
    \leq N^{O(1)}L(n)^{O(1)}(D_J(n)+1)\sqrt n.
    \label{eq:ags-fixed-apex-peel}
  \end{equation}
  The implicit constants are independent of $d_J$.  Moreover
  \begin{equation}
    |\overline M_J|=|M|-|\mathcal I_J|+1
      \leq |M|-d_J^2+1.
    \label{eq:ags-fixed-apex-drop}
  \end{equation}
\end{lemma}

\begin{proof}
  Write $q:M\to\overline M_J$ for the Rees quotient.  First compute the
  quotient product of the whole word.  If it is nonzero, then it has a
  unique lift $m\in M\setminus \mathcal I_J$.  The desired value $\rho_J(m)$ is a
  free table lookup.

  Suppose instead that the quotient product is zero.  Because $\mathcal I_J$
  is an ideal, there is a first prefix whose product enters $\mathcal I_J$.
  It can be
  found by difference-block binary search without a logarithmic loss in
  normalized cost.  Maintain a known nonzero quotient prefix state at a cut
  $\ell$ and a later zero cut $r$.  Query the quotient product of the
  difference block from $\ell$ to their midpoint, that is, of
  $[\ell,\ \ell+\lfloor(r-\ell)/2\rfloor)$.

  As in the first-exit search in the proof of
  \propref{prop:ags-axis-recurrence}, the gap minus one halves after each
  step.  The queried lengths are bounded by
  $n-1,(n-1)/2,(n-1)/4,\ldots$, whose square roots sum to $O(\sqrt n)$.
  At the final one-letter
  gap, the nonzero predecessor has a unique lift
  $p\in M\setminus \mathcal I_J$;
  querying the crossing letter $a$ returns the exact first entry
  $h=pa\in \mathcal I_J$.

  If $h$ lies strictly below $J$, then $\rho_J(h)=0$, and the entire
  continuation remains zero in this representation.  It remains to treat
  $h\in J$.  While a right continuation remains in $J$,
  \lemref{lem:ags-green-stability} gives
  \[
    hu\mathrel{\JJ}h\quad\Longrightarrow\quad hu\mathrel{\RR}h,
  \]
  because $hu\leq_{\RR}h$.  Thus the continuation either remains in the
  killed $\RR$-axis $R_h$ or exits strictly below $J$, after which the
  $J$-coordinate is zero.

  For each $s\in R_h$, apply \lemref{lem:ags-local-step} to the target
  predicate
  \[
    h\,\operatorname{Prod}(w)=s.
  \]
  Here $1\notin\mathcal I_J$, so the first entry consumes at least one raw
  letter.  Consequently, prepending the known $h$ to the remaining suffix
  still gives a virtual word of length at most $n$ and costs no query.
  Every recursive equality target $u$ used by the localized AGS step satisfies
  \[
    MJM\subsetneq MuM,
  \]
  and hence $u\notin \mathcal I_J$.  The Rees quotient is injective outside
  $\mathcal I_J$,
  so on every queried interval
  \[
    \operatorname{Prod}_M(v)=u
    \quad\Longleftrightarrow\quad
    \operatorname{Prod}_{\overline M_J}(q(v))=q(u).
  \]
  Exact quotient-product output therefore supplies every recursive target
  required by \lemref{lem:ags-local-step}.  Running the resulting
  one-hot family over $s\in R_h$ returns the exact live endpoint, or reports
  that the trajectory has left $R_h$.  Finite target multiplicity, the
  adaptive first-entry search, and bounded-error amplification contribute
  only $N^{O(1)}L(n)^{O(1)}$, proving
  \eqref{eq:ags-fixed-apex-peel}.

  Finally, $\mathcal I_J$ contains $J$, and
  \lemref{lem:ags-munn-decomposition} gives $|J|\geq d_J^2$.
  Collapsing $\mathcal I_J$ to one zero proves
  \eqref{eq:ags-fixed-apex-drop}.
\end{proof}

\begin{remark}
  The representation $\rho_J$ does \emph{not} factor through
  $M/\mathcal I_J$: the quotient collapses $J$, where $\rho_J$ is nonzero.
  \lemref{lem:ags-fixed-apex-peel} is an algorithmic oracle reduction.
  Terminal quotient product alone is insufficient; the first-entry and
  killed-axis packets recover the information lost at the quotient zero.
\end{remark}

\subsection{Lifting through the radical}
\label{sec:ags-radical-lift}

After all simple coordinates have been computed, we know the product only
modulo the radical $R$ of \lemref{lem:ags-munn-decomposition}, the
kernel of the joint coordinate map.  The next lemma lifts exact products through a
square-zero algebra kernel with polynomial overhead.  It is important that
no algebra section is assumed.

The algebraic input is a local-triviality result of Almeida, Margolis,
Steinberg, and Volkov \cite[Lemma~5.4]{AMSV09}.  Our contribution here is
the quantum reconstruction of the original product from quotient-product
calls, together with its query-cost bound and its iteration through
successive radical quotients.

\begin{lemma}[Square-zero product lift]
\label{lem:ags-square-zero-lift}
  Let $B\to C$ be a unital quotient of finite-dimensional rational
  algebras with square-zero kernel.  Let $S\subseteq B$ be a finite
  aperiodic submonoid containing $1_B$, and let $T$ be its image in $C$.
  If length-$\ell$ products in $T$, uniformly for $\ell\leq n$, have
  normalized query cost at most $D$, then products in $S$ have query cost
  \begin{equation}
    N^{O(1)}L(n)^{O(1)}(D+1)\sqrt n,
    \label{eq:ags-square-zero-lift}
  \end{equation}
  where $N$ is any common upper bound on $|S|$ and $|T|$.
\end{lemma}

\begin{proof}
  \emph{The kernel category.}  Let $\phi:S\to T$ be the quotient map.
  We use the two-sided kernel category of Rhodes and Tilson~\cite{RT89},
  following the formulation in~\cite[Section~5]{AMSV09}.  It records how
  a segment product acts between known quotient contexts.
  Its objects are pairs $(p,q)\in T^2$.  A segment product
  $s\in S$ represents an arrow
  \[
    (p,q)\longrightarrow(p\phi(s),q')
    \quad\text{when}\quad \phi(s)q'=q.
  \]
  Two representatives $s,t$ with the same source $(p,q)$ and target
  $(p',q')$ are identified if $xsy=xty$ for every $x,y\in S$ with
  $\phi(x)=p$ and $\phi(y)=q'$.  Thus equivalent arrows can be
  substituted inside any compatible surrounding context.  Composition
  is represented by multiplication in $S$, and identity arrows by $1$.

  \emph{Loops and components.}  Let $K=\ker(B\to C)$.
  The faithful representation of $B$ by left multiplication has block
  upper triangular form in a basis adapted to $K$.  Since $K^2=0$,
  both diagonal actions factor through $B/K$, and the induced action of
  $B/K$ on itself is faithful.  Thus the block-diagonal image of $S$ is isomorphic to
  $T$, so~\cite[Lemma~5.4]{AMSV09} implies that the kernel category is
  locally trivial.  We give a direct verification in our notation.
  Consider a loop represented by $s$, so
  \[
    p\phi(s)=p,
    \qquad \phi(s)q=q.
  \]
  Choose compatible lifts $x,y\in S$ of $p,q$.  Then
  \[
    \delta=x(s-1)y\in K.
  \]
  Also $xs^j-x\in K$ and $(s-1)y\in K$.  Since $K^2=0$,
  \[
    xs^j(s-1)y=x(s-1)y=\delta
  \]
  for every $j\geq0$, and telescoping gives
  \begin{equation}
    xs^\ell y-xy=\ell\delta.
    \label{eq:ags-square-zero-loop}
  \end{equation}
  Two powers of $s$ coincide because $S$ is finite.  Subtracting the two
  corresponding instances of~\eqref{eq:ags-square-zero-loop} and using
  characteristic zero gives $\delta=0$.  Thus every local loop is the
  identity.  It follows that inside a strongly connected component there
  is at most one arrow between any fixed pair of objects.

  \emph{Finding the root arrow.}  For a word $s_1\cdots s_n$, put
  \[
    p_i=\phi(s_1\cdots s_i),
    \qquad q_i=\phi(s_{i+1}\cdots s_n),
    \qquad o_i=(p_i,q_i).
  \]
  The cuts $o_0,o_1,\ldots,o_n$ form a path in the kernel category.
  First compute the full $T$-product $t$, which determines the endpoint
  objects $o_0=(1,t)$ and $o_n=(t,1)$.  Use a fixed dyadic recursion
  on the interval of cuts.  If the two endpoint
  objects of a dyadic interval lie in one strongly connected component,
  the intervening arrow is the unique tabled arrow and no further query is
  needed.  Otherwise split at the midpoint, compute the $T$-products of
  the two child intervals, derive the midpoint cut from the known endpoint
  contexts, and continue.  At a one-letter interval, query that letter
  to determine its arrow.  The condensation graph is acyclic, so the
  cut path changes component at most $|T|^2-1$ times.  Every such change is
  contained in only $O(\log(n+2))$ dyadic intervals.  Hence only
  $N^{O(1)}L(n)^{O(1)}$ quotient-product calls and crossing letters survive.
  Robust adaptive composition yields the cost in
  \eqref{eq:ags-square-zero-lift}.

  \emph{Recovering the product.}  The root arrow determines the exact
  $S$-product.  To see this, note that $\phi^{-1}(1)=\{1\}$.
  If $\phi(s)=1$, then $s=1+k$ with
  $k\in K$ and $k^2=0$, so $s$ is invertible in $B$.  A power of $s$
  stabilizes because $S$ is finite and aperiodic; multiplying the stability
  identity by the algebra inverse gives $s=1$.  Thus the two outside
  context fibres at the root are singletons, and equality of root arrows
  is equality of the represented products.
\end{proof}

For a finite monoid $H$, write
\[
  \mathcal D_H(n)=\max_{1\leq\ell\leq n}
    \frac{Q_{1/3}(\operatorname{Prod}_{H,\ell})}{\sqrt\ell}.
\]
Let $A$ be the contracted rational algebra of $M$ and let
$R$ be its radical, the kernel of the joint coordinate map of
\lemref{lem:ags-munn-decomposition}.  It is nilpotent of index at most
$\dim A\leq N$.  For $r=1,2,4,\ldots$, consider
\[
  A/R^{2r}\longrightarrow A/R^r.
\]
Its kernel $R^r/R^{2r}$ is square-zero.  Choose $k$ so that $2^k$ is at
least the nilpotency index of $R$.
Applying \lemref{lem:ags-square-zero-lift} to the successive finite
images of $M$ uses $k=O(\log N)$ layers and proves
\begin{equation}
  \mathcal D_M(n)+1
  \leq (N L(n))^{O(\log N)}
    \bigl(\mathcal D_{M\bmod R}(n)+1\bigr).
  \label{eq:ags-radical-lift}
\end{equation}
Here $M\bmod R$ means the image of $M$ in
$A/R$, equivalently (since $A/R$ embeds in the product of the blocks)
the tuple of the simple coordinates in
\lemref{lem:ags-munn-decomposition}.  The argument above applies even
when the intermediate quotients do not split as algebras.
At the final layer $R^{2^k}=0$, so the terminal image is the image in $A$;
the contracted basis embedding is faithful and therefore recovers the
exact product in $M$.

\subsection{The size recurrence and the main theorem}
\label{sec:ags-size-recurrence}

Fix a horizon $n_0$, put $L=L(n_0)$, and define
\[
  C_s=\max\left\{
    \frac{Q_{1/3}(\operatorname{Prod}_{H,\ell})}{\sqrt\ell}:
    H\text{ aperiodic},\ |H|\leq s,\ 1\leq\ell\leq n_0
  \right\}.
\]
We use the global upper bounds $N$ and $L$ in every recursive call, so
$C_s$ is monotone in $s$ and all implicit constants below are uniform.

\begin{proposition}[Cube-root carrier recurrence]
\label{prop:ags-cuberoot-recurrence}
  For every $1\leq s\leq N$ and every integer $t\geq2$,
  \begin{equation}
    C_s\leq (N L)^{O(\log N)}
      \left((N L)^{O(t)}+(N L)^{O(1)}C_{s-t^2}\right),
    \label{eq:ags-cuberoot-recurrence}
  \end{equation}
  where the recursive term is omitted when $s\leq t^2$.  Consequently
  \begin{equation}
    C_N\leq
      (N L)^{O\left(t+\left(\frac{N}{t^2}+1\right)\log N\right)}.
    \label{eq:ags-cuberoot-recurrence-solved}
  \end{equation}
\end{proposition}

\begin{proof}
  Fix a monoid $H$ of order at most $s$.  If it has one element, its
  query cost is zero.  Otherwise apply \lemref{lem:ags-munn-decomposition}.
  Every simple coordinate of degree at most $t$ is computed by
  \eqref{eq:ags-small-simple-coordinate}.  There are at most $s$ such
  coordinates, so their complete tuple costs $(NL)^{O(t)}$ in normalized
  form.

  If $d_J>t$, apply \lemref{lem:ags-fixed-apex-peel}.  Since $d_J$ is
  integral,
  \[
    |H/\mathcal I_J|\leq |H|-d_J^2+1\leq s-t^2.
  \]
  Thus this coordinate costs at most $(NL)^{O(1)}C_{s-t^2}$.
  Summing over the at most $s$ large coordinates contributes only another
  polynomial factor.  We have now computed the exact tuple of simple
  coordinates.
  Equation~\eqref{eq:ags-radical-lift} supplies the outer
  $(NL)^{O(\log N)}$ factor; its added $1$ is absorbed by the small-block
  term.  This proves
  \eqref{eq:ags-cuberoot-recurrence}.

  Every recursive call decreases the carrier size by at least $t^2$, so
  a branch has at most $\lceil N/t^2\rceil$ large-apex peels.  At the
  terminal node all blocks are small.  Write $a=(NL)^{O(\log N)}$
  for a uniform upper bound on the multiplier of a recursive term,
  increasing it so that $a\geq2$.  The additive small-block terms form
  a geometric sum bounded by twice its largest term.  Their total is
  therefore $(NL)^{O(t)}a^{O(N/t^2+1)}$, proving
  \eqref{eq:ags-cuberoot-recurrence-solved}.
\end{proof}

We can now state and prove the bound in terms of monoid size.  The choice
of threshold balances the matrix dimension against the number of
reductions to smaller monoids.

\Needspace{18\baselineskip}
\begin{theorem}[Cube-root-size AGS bound]
\label{thm:ags-cuberoot-size}
  There is a universal constant $C$ such that, for every finite aperiodic
  monoid $M$ of order $N$ and every $n\geq1$,
  \begin{equation}
    Q_{1/3}(\operatorname{Prod}_{M,n})
    \leq
    \min\left\{n,
      \sqrt n\,L(n)^{C(N\log(N+2))^{1/3}}
    \right\}.
    \label{eq:ags-cuberoot-final}
  \end{equation}
  Before the standard read-all absorption, the proof gives the convenient
  form
  \begin{equation}
    Q_{1/3}(\operatorname{Prod}_{M,n})
    \leq
    \sqrt n\,(N L(n))^{C(N\log(N+2))^{1/3}}.
    \label{eq:ags-cuberoot-unabsorbed}
  \end{equation}
  Consequently, if $S$ is a finite aperiodic semigroup, the same estimates
  hold after replacing $N$ by $|S|+1$.
\end{theorem}

\begin{proof}
  The one-element monoid needs no queries, so assume $N\geq2$.
  In Proposition~\ref{prop:ags-cuberoot-recurrence}, choose
  \[
    t=\max\left\{2,
      \left\lceil(N\log(N+2))^{1/3}\right\rceil
    \right\}.
  \]
  Then
  \[
    t+\frac{N}{t^2}\log N
      =O\bigl((N\log(N+2))^{1/3}\bigr).
  \]
  Taking $n_0=n$ in
  \eqref{eq:ags-cuberoot-recurrence-solved} proves
  \eqref{eq:ags-cuberoot-unabsorbed}.

  It remains to absorb the powers of $N$ by comparing with the read-all
  bound.  If $L(n)<N^{1/3}$, then $\log_2 n<N^{1/3}$, and reading
  every input letter gives
  \[
    n\leq\sqrt n\,2^{N^{1/3}}
      \leq\sqrt n\,L(n)^{N^{1/3}},
  \]
  since $L(n)\geq2$.  This is at most the desired upper bound because
  $N^{1/3}\leq(N\log(N+2))^{1/3}$.
  If $L(n)\geq N^{1/3}$, then $N\leq L(n)^3$, so the factor
  $(NL(n))^{O((N\log(N+2))^{1/3})}$ in
  \eqref{eq:ags-cuberoot-unabsorbed} is at most
  \[
    L(n)^{O((N\log(N+2))^{1/3})}.
  \]
  Combining the two cases with the read-all upper bound proves
  \eqref{eq:ags-cuberoot-final}.  Adjoining an identity proves the
  semigroup statement.
\end{proof}

\subsection{Structural examples and the remaining gap}
\label{sec:ags-structural-examples}

The upper-bound proof is complete.  We return to the Dyck monoids to
explain how their structure relates to the estimates used in that proof.
Their query lower bound was established in \secref{sec:dyck-lb}; the
following analysis identifies the sizes of their simple representations
and the quotients used in the recursion.

\paragraph{The Dyck principal factors.}
Recall from \propref{prop:dyck-monoid} that a nonzero element of $M_k$
is a partial translation with parameters $(a,b,e)$, and its $\JJ$-class
is determined by $a+b$.

These coordinates show how this family attains the size estimates used in
\thmref{thm:ags-cuberoot-size}: every nonzero principal factor is a full matrix-unit
semigroup, with equality in the simple-degree packing bound, and the recursive deepest-apex
peel has the corresponding exact carrier drop.

\begin{proposition}[Dyck principal factors and exact apex peeling]
\label{prop:dyck-munn-peel}
  For $0\le s\le k$, let $J_s$ be the level-$s$ $\JJ$-class of $M_k$, and write
  \[
    [s;i,j]:=(i,s-i,j-i),
    \qquad 0\le i,j\le s.
  \]
  Thus $[s;i,j]$ is the partial translation
  \[
    [i,i+k-s]\longrightarrow[j,j+k-s],
    \qquad h\longmapsto h+j-i.
  \]
  Put
  \[
    I_s=\{0\}\cup\bigcup_{u=s}^{k}J_u,
    \qquad I_{k+1}=\{0\}.
  \]
  Let $d_{J_s}$ denote the rational rank of the principal factor's sandwich matrix,
  equivalently the degree of its simple Munn block.
  Then:
  \begin{enumerate}
    \item $M_k$ is an inverse monoid, with
      $[s;i,j]^{-1}=[s;j,i]$ and $0^{-1}=0$;
    \item the principal factor $I_s/I_{s+1}$ is the Brandt semigroup of
      $(s+1)$-by-$(s+1)$ matrix units:
      \begin{equation}
        [s;i,j][s;p,q]=
        \begin{cases}
          [s;i,q],&j=p,\\
          0,&j\ne p
        \end{cases}
        \quad\text{in }I_s/I_{s+1};
        \label{eq:dyck-principal-matrix-units}
      \end{equation}
      consequently its sandwich matrix is the identity and
      \begin{equation}
        d_{J_s}=s+1,
        \qquad |J_s|=(s+1)^2=d_{J_s}^{2};
        \label{eq:dyck-block-square}
      \end{equation}
    \item the contracted rational algebra is radical-free and
      \begin{equation}
        \mathbb Q_0[M_k]\cong\bigoplus_{s=0}^{k}M_{s+1}(\mathbb Q),
        \qquad \operatorname{Rad}(\mathbb Q_0[M_k])=0;
        \label{eq:dyck-contracted-decomposition}
      \end{equation}
    \item if $k\ge1$, then $M_k^1J_kM_k^1=I_k=J_k\cup\{0\}$ and
      \begin{equation}
        M_k/I_k\cong M_{k-1}.
        \label{eq:dyck-top-peel}
      \end{equation}
      In particular,
      $|M_k|-|M_k/I_k|=(k+1)^2=d_{J_k}^2$.
  \end{enumerate}
\end{proposition}

\begin{proof}
  The displayed interval description shows that $[s;i,j]$ is a partial bijection onto
  $[j,j+k-s]$, whose set-theoretic inverse is $[s;j,i]$.  This proves the first assertion.

  Two level-$s$ intervals have equal length.  If the range of $[s;i,j]$ and the domain of
  $[s;p,q]$ have the same left endpoint, namely $j=p$, they coincide, and composition gives
  $[s;i,q]$.  If $j\ne p$, their intersection is strictly shorter (or empty), so the
  composite is zero or has level strictly greater than $s$, and hence lies in $I_{s+1}$.
  This proves~\eqref{eq:dyck-principal-matrix-units}.  Its sandwich matrix is
  the identity matrix of size $s+1$, of rational rank $s+1$.
  The pairs $(i,j)$ count the $(s+1)^2$ elements of $J_s$.
  This proves~\eqref{eq:dyck-block-square}.

  Munn's decomposition of a finite inverse-semigroup algebra~\cite{Munn55,Munn57} gives
  one matrix block over the maximal subgroup of every nonzero regular $\JJ$-class.
  Here all maximal subgroups are trivial and the block at $J_s$ has size $s+1$.  Therefore
  \[
    \mathbb Q_0[M_k]/\operatorname{Rad}(\mathbb Q_0[M_k])
      \cong\bigoplus_{s=0}^{k}M_{s+1}(\mathbb Q).
  \]
  The right-hand side has dimension
  $\sum_{s=0}^{k}(s+1)^2=|M_k|-1=\dim_{\mathbb Q}\mathbb Q_0[M_k]$.
  The radical consequently has dimension zero, proving
  \eqref{eq:dyck-contracted-decomposition}.

  Finally, $I_k=J_k\cup\{0\}$ by the $\JJ$-chain description in
  Proposition~\ref{prop:dyck-monoid}.  Send every level-$s$ triple with $s<k$ to the same
  triple in $M_{k-1}$ and send $I_k$ to zero.  The multiplication formula in the proof of
  Proposition~\ref{prop:dyck-monoid} is independent of the height bound: the height only
  decides whether the resulting level survives or becomes zero.  Hence this is a
  surjective homomorphism with zero fibre $I_k$, and it is injective on
  $M_k\setminus I_k$ because distinct surviving triples remain distinct.  It therefore
  induces the Rees-quotient isomorphism~\eqref{eq:dyck-top-peel}.
\end{proof}

Thus the Dyck family saturates $d_J^2\le|J|$, and along
$M_k\to M_{k-1}\to\cdots\to M_0$ every deepest-apex peel attains the corresponding
$d_J^2$ carrier drop.  It has no radical to lift through.  The fixed
$2^{\Theta(N^{1/3})}$ lower bound therefore already occurs in a radical-free family
attaining both equalities, whereas the additional logarithmic loss in the power of $L(n)$
in \thmref{thm:ags-cuberoot-size} is not witnessed by this family.

\paragraph{One Brandt factor.}
The appearance of large Brandt principal factors in
Proposition~\ref{prop:dyck-munn-peel} might suggest that one such factor
already carries the exponential lower bound.  In fact the opposite is true.
For $k\geq2$, put $[k]=\{1,\ldots,k\}$ and write
\[
  B_k^1=\{\identity,0\}\cup\{e_{ij}:i,j\in[k]\},
  \qquad
  e_{ij}e_{\ell r}=
  \begin{cases}
    e_{ir},&j=\ell,\\
    0,&j\ne\ell,
  \end{cases}
\]
where $\identity$ is the identity and $0$ is absorbing.

\begin{proposition}[Brandt products reduce to the two-state case]
\label{prop:brandt-two-colour}
  Put
  \[
    q_k(n)=Q_{1/3}(\operatorname{Prod}_{B_k^1,n}).
  \]
  There is a universal constant $C$ such that, for every $k\geq2$ and
  $n\geq1$,
  \[
    q_2(n)\leq q_k(n)
    \leq C\bigl(q_2(n)+\sqrt n\bigr).
  \]
  Consequently $q_k(n)=\Theta(q_2(n))$ uniformly in $k$.  In particular,
  there is an absolute constant $A$ such that, uniformly for all $k\geq2$,
  \[
    \Omega(\sqrt n)
    \leq Q_{1/3}(\operatorname{Prod}_{B_k^1,n})
    \leq O\!\left(\sqrt n\,(\log(n+2))^A\right).
  \]
\end{proposition}

\begin{proof}
  Suppose first that the input contains no zero letter.  Delete its identity
  letters and write the remaining word as
  \[
    e_{i_1j_1}e_{i_2j_2}\cdots e_{i_tj_t}.
  \]
  Its product is $\identity$ if $t=0$.  If $t>0$, its product is nonzero
  precisely when
  \[
    j_s=i_{s+1}\qquad(1\leq s<t),
  \]
  in which case it equals $e_{i_1j_t}$.

  Choose $h:[k]\to[2]$ by taking $h(1),\ldots,h(k)$ independently and
  uniformly from $[2]$, and encode each queried letter locally by
  \[
    \Phi_h(\identity)=\identity,
    \qquad \Phi_h(0)=0,
    \qquad \Phi_h(e_{ij})=e_{h(i),h(j)}\in B_2^1.
  \]
  This map need not be a homomorphism.  Nevertheless, if the original word
  has nonzero product, its encoded word has nonzero product for every $h$.
  If the original product is zero because
  $j_s\ne i_{s+1}$ for some $s$, then
  $\Pr_h[h(j_s)\ne h(i_{s+1})]=1/2$, and on this event the encoded word has
  zero product.  An explicit zero
  letter is preserved by every colouring.  Take six independent colourings
  and amplify each $B_2^1$ product call to error at most $1/100$.  Declare
  the original product zero if any call returns zero.  On a nonzero input
  the error is at most $6/100$; on a zero input it is at most
  $2^{-6}+6/100$.  Thus this is a bounded-error zero test.  A coherent query
  to an encoded letter is simulated using two original queries: query the
  original letter into a work register, compute its encoded value
  reversibly, and then unquery the work register.

  If the zero test reports nonzero, quantum first- and last-marked search
  finds the first and last matrix-unit letters using $O(\sqrt n)$ queries.
  If there is no matrix-unit letter, the product is $\identity$; otherwise,
  if the two extremal letters are $e_{ij}$ and $e_{\ell r}$, the product is
  $e_{ir}$.  Constant amplification of these searches proves the upper
  bound with bounded error.

  For the lower bound, the matrix units on any fixed two indices, together
  with $\identity$ and $0$, form a copy of $B_2^1$ inside $B_k^1$.  Restricting
  the input alphabet to this copy gives $q_2(n)\leq q_k(n)$.  Restriction
  further to $\{\identity,0\}$ is unstructured search, so
  $q_2(n)=\Omega(\sqrt n)$ and the additive term in the upper bound is
  absorbed.  Finally, $B_2^1$ is a fixed nontrivial aperiodic monoid, so the
  final displayed estimate follows from
  \thmref{thm:fixed-monoid-trichotomy}.
\end{proof}

The matching breadth lower bound for commutative monoids does not extend
to this family.  The word
$e_{12}e_{23}\cdots e_{k-1,k}e_{k1}$ has product $e_{11}$ and no proper
product-preserving subword, so $\beta_{B_k^1}(B_k^1)\geq k$.
Taking $k=n$, a putative lower bound
$\Omega(\min\{n,\sqrt{n\beta}\})$ would give $\Omega(n)$ queries,
contradicting the uniform upper bound above.

Consequently no lower bound of the form $\Omega(c^k\sqrt n)$, for a fixed
$c>1$, can hold uniformly as $k$ grows.  This conclusion uses the present
value-oracle model, in which one query reveals both indices of a matrix unit.
One isolated Brandt factor therefore does not explain the Dyck lower
bound.  The Dyck family has a chain of $k+1$ such factors; the comparison
does not assert that an arbitrary chain of Brandt factors is hard.
Indeed $B_k^1$ has only the
three $\JJ$-classes $\{\identity\}$, $\{e_{ij}:i,j\in[k]\}$, and $\{0\}$,
so $\dJ(B_k^1)=2$ independently of $k$, although the Munn degree of its
matrix-unit class is $k$.  Moreover, the fact that a Brandt principal factor
divides $M_k$ gives a reduction from its product problem \emph{to} the
product problem in $M_k$, which is the wrong direction for transferring
the Dyck lower bound to the factor.

\paragraph{The worst-case bound.}
For $N\ge1$, let $\mathfrak A_N$ be a fixed set of representatives of the isomorphism
classes of finite aperiodic monoids of order at most $N$.  For integers $N,n\ge1$, define
the worst-case aperiodic envelope
\[
  \mathcal Q_{\rm ap}(N,n)=
  \max\left\{
    Q_{1/3}(\operatorname{Prod}_{M,n}):
    M\in\mathfrak A_N
  \right\}.
\]

\begin{theorem}[Worst-case aperiodic envelope]
\label{thm:aperiodic-envelope}
  There is a universal constant $C>0$ such that, for all $N,n\ge1$,
  \[
    \mathcal Q_{\rm ap}(N,n)
    \le
    \min\left\{n,
      \sqrt n\,2^{C N^{1/3}}
      L(n)^{C N^{1/3}\log^{1/3}(N+2)}
    \right\}.
  \]
  Moreover, there are universal constants $a,b>0$ and an integer $N_0$ such that,
  whenever
  \[
    N_0\le N\le a\bigl(\log(n+2)\bigr)^3,
  \]
  one has
  \[
    \mathcal Q_{\rm ap}(N,n)\ge\sqrt n\,2^{bN^{1/3}}.
  \]
  Thus throughout this regime the fixed exponential $2^{\Theta(N^{1/3})}$ is optimal;
  only the power of $L(n)$ remains unmatched.
\end{theorem}

\begin{proof}
  The upper bound is \thmref{thm:ags-cuberoot-size}, applied to each monoid of order
  at most $N$.

  For the lower bound put $k=\lfloor N^{1/3}\rfloor-2$.  After increasing $N_0$, we have
  $k=\Omega(N^{1/3})$, while Proposition~\ref{prop:dyck-monoid} gives
  \[
    |M_k|=\sum_{j=1}^{k+1}j^2+1\le(k+2)^3\le N.
  \]
  Choose $a>0$ small enough that
  $N\le a(\log(n+2))^3$ implies $k\le\log_2 n$.  Corollary~\ref{cor:dyck} then gives
  \[
    Q_{1/3}(\operatorname{Prod}_{M_k,n})
      =\Omega\!\left(\sqrt n\,2^{bN^{1/3}}\right)
  \]
  for a sufficiently small universal $b>0$.  Decreasing the exponent constant and
  increasing $N_0$ absorbs
  the implicit multiplicative constant.
\end{proof}

\begin{corollary}[The polylogarithmic-size transition]
\label{cor:aperiodic-envelope-transition}
  For every fixed $0<\epsilon<1$, there is a constant $C_\epsilon$ such that, for all
  sufficiently large $n$,
  \[
    N\ge C_\epsilon\bigl(\log(n+2)\bigr)^3
    \quad\Longrightarrow\quad
    \mathcal Q_{\rm ap}(N,n)=\Omega_\epsilon(n^{1-\epsilon}).
  \]
  In addition, for $n\ge2$,
  \[
    N\ge\lfloor n/2\rfloor+1
    \quad\Longrightarrow\quad
    \mathcal Q_{\rm ap}(N,n)=\Theta(n).
  \]
\end{corollary}

\begin{proof}
  For the first statement, take the monoid $M_k$ supplied by
  Corollary~\ref{cor:dyck}(iii).  Its order is $O_\epsilon(\log^3n)$,
  and its product problem requires $\Omega_\epsilon(n^{1-\epsilon})$ queries.
  The claim follows because $\mathcal Q_{\rm ap}$ is monotone in $N$.

  For the second statement take $K=\lfloor n/2\rfloor$.  The capped counter
  $C_{K+1}=\{0,1,\ldots,K\}$ has order $K+1$, is aperiodic, and
  Proposition~\ref{prop:jtrivial-log-fails} gives
  $Q_{1/3}(\operatorname{Prod}_{C_{K+1},n})=\Omega(\sqrt{nK})=\Omega(n)$.
  Reading all inputs gives the matching upper bound.
\end{proof}

\begin{remark}[Where the logarithmic loss occurs]
  A large apex of degree $d_J$ owns at least $d_J^2$ physical monoid
  elements, which is the source of the cube-root balance.  The fixed-apex
  peel itself costs only a polynomial factor.  The extra
  $\log^{1/3}(N+2)$ in the exponent of $L(n)$ comes from paying the
  $O(\log N)$ radical-doubling lift at every large-apex descent.  A
  coefficient-one radical lift shared across all peels would remove this
  loss; it is not used here.
\end{remark}

\clearpage
\section{AI methodology}\label{sec:ai-methodology}

We use the categories proposed in the AI and TCS working-group
report~\cite[Section~2.1]{AITCS26}.
\begin{center}
\begin{tabular}{@{}ll@{}}
\textbf{Activity} & \textbf{AI used} \\
Asking the research question & Yes \\
Coming up with the approach & Yes \\
Proof development & Yes \\
Writing and exposition & Yes \\
Checking for bugs & Yes \\
Other supporting tasks & Yes
\end{tabular}
\end{center}

\paragraph{Details.}
This research began as an attempt to generalize the techniques in our
earlier paper on quantum divide and conquer~\cite{ABBLS23}.  We recognized
the monoid product problem as a general setting capturing a large class
of divide-and-conquer algorithms.  This led us to investigate what the
AGS bound implies quantitatively as a function of $|M|$, and when monoid
structure can improve it.

We had already carried out substantial work on the commutative idempotent
case before we found AI models capable of independently proving results
in this area.  We supplied our research notes to ChatGPT (GPT-5.6 Sol),
which fully resolved that case.

At the same time, we were building a quantum query complexity library in
Lean~4 with Claude Opus~4.8 and Claude Fable~5.  Finding an explicit
adversary dual solution for maximum with cost $O(\sqrt n)$ proved
surprisingly difficult.  GPT-5.6 produced such a construction, which was very nice.  We
recognized that the method could be generalized and, working with the
model, developed the essential-width theorem
(\thmref{thm:essential-width}).  This led to the resolution of the full
commutative case (\thmref{thm:commutative-beta}).

As a generalization of the best time to buy and sell stock problem,
Claude Fable~5 suggested upper unitriangular matrix multiplication over
the max-plus tropical semiring.  Close collaboration with GPT-5.6 eventually
led to a logarithmic exponent bound for this case.  A separate thread of work focused on
product breadth, with the initial aim of bounding query complexity for
general $\JJ$-trivial monoids.  This investigation identified insertion
monotonicity as a useful additional property.  This property holds in
stably ordered monoids with identity as the minimum element and abstracts
the role of the zero-weight diagonal entries in the tropical setting.
The general query bound emerged through further interactive work on
sampling short subwords and combining their positions in the original
input order.  Refinements eventually gave an exponent logarithmic in
product breadth, recovering the $O(\log k)$ dependence of our specialized
bound for $k$-by-$k$ unitriangular matrices.

The cube-root bound was also proved by GPT-5.6 over many sessions of
prompting and guidance from the authors.

We have developed over 100K lines of Lean~4 code for this project,
including quantum query complexity~\cite{palomar-2026-09-29-000001-v1}
and sunflower bounds~\cite{Lee26} libraries, both registered on Palomar.
The Lean formalization specific to this paper is available on GitHub and archived on
Zenodo~\cite{MonoidProduct26}.  It covers every named result in the paper, with the
qualifications described below.  The repository's correspondence table records, for each numbered
statement, the Lean declaration and its precise form.
Some intermediate results are expressed in Lean as bounds on adversary
dual costs rather than directly on quantum query complexity.  The two cost
estimates discussed in \remarkref{rem:ordered-formalization} are verified
within the iterative proof of \thmref{thm:ordered-beta-log-product}, rather
than stated as standalone lemmas.
The formalization was developed primarily with Claude models, most
recently Opus~5.5.

The Lean development also informed revisions to the
paper.  For example, comparison of the AGS estimate with its Lean
formulation led us to correct the accounting of the base-case cost and
revise the displayed bound in \thmref{thm:main-ags}.  Sometimes the process of formalization in Lean
led to a simpler proof than the original argument in the paper, which would then
feedback into improving the paper presentation of the proof.

AI tools also helped draft proof overviews, reorganize arguments, develop
examples, search the literature, and check claims against source papers.
We have extensively revised AI written drafts to improve readability and clarity.
We take full responsibility for the material in this paper.

\end{document}